\documentclass[10pt,a4paper]{article}
\usepackage[utf8]{inputenc}
\usepackage{amsmath, mathtools}
\usepackage{amsfonts, bbm, mathrsfs, yfonts}
\usepackage{amssymb, amsthm, upgreek}
\usepackage{enumitem, hyperref,xcolor, textcmds}
\usepackage{xargs}
\usepackage{booktabs}
\usepackage{geometry}
\usepackage{tikz}
\usepackage{caption, subcaption, placeins}
\usepackage{newtxtext}
\usepackage[subscriptcorrection]{newtxmath}

\usepackage{graphicx}
\usepackage{natbib}

\graphicspath{{../../figures/}}

\usepackage[plain,noend]{algorithm2e}

\usepackage{upgreek}
\usetikzlibrary{calc,arrows.meta,positioning, fit}
\DontPrintSemicolon

\makeatletter
\renewcommand{\algocf@captiontext}[2]{#1\algocf@typo. \AlCapFnt{}#2} 
\def\@algocf@capt@plain{top}
\renewcommand{\algocf@makecaption}[2]{%
  \addtolength{\hsize}{\algomargin}%
  \sbox\@tempboxa{\algocf@captiontext{#1}{#2}}%
  \ifdim\wd\@tempboxa >\hsize
    \hskip .5\algomargin%
    \parbox[t]{\hsize}{\algocf@captiontext{#1}{#2}}
  \else%
    \global\@minipagefalse%
    \hbox to\hsize{\box\@tempboxa}
  \fi%
  \addtolength{\hsize}{-\algomargin}%
}
\makeatother

\def\dist{\operatorname{dist}}
\def\trace{\textrm{Tr}}

\renewcommand{\d}{{\rm d}}
\def\Leb{\mathrm{Leb}}
\newcommand{\argmax}{\operatorname*{arg\,max}}

\title{Geodesic slice sampling on Riemannian manifolds}

\author{Alain Durmus\thanks{\'Ecole Polytechnique, France, Email: alain.durmus@polytechnique.edu} 
	\and Samuel Gruffaz\thanks{Universit\'e Paris Saclay, France, Email: samuel.gruffaz@ens-paris-saclay.fr} 
	\and Mareike Hasenpflug\thanks{University of Passau, Germany, Email: mareike.hasenpflug@uni-passau.de, daniel.rudolf@uni-passau.de} 
	\and Daniel Rudolf\footnotemark[3]}

\theoremstyle{definition}
\newtheorem{theorem}{Theorem}
\newtheorem{lemma}[theorem]{Lemma}

\newtheorem{remark}[theorem]{Remark}
\newtheorem{example}[theorem]{Example}
\newtheorem{assumption}{Assumption}

\begin{document}
\maketitle

\begin{abstract}
We propose a theoretically justified and practically applicable slice sampling based Markov chain Monte Carlo (MCMC) method for approximate sampling from probability measures on Riemannian manifolds. 
The latter naturally arise as posterior distributions in Bayesian inference of matrix-valued parameters, 
for example belonging to either the Stiefel or the Grassmann manifold.
Our method, called geodesic slice sampling, 
generalizes Hit-and-run slice sampling on $\mathbb{R}^{d}$ to Riemannian manifolds by 
abstracting straight lines to geodesics.
It is reversible with respect to the distribution of interest, and converges to it in total variation distance.
We demonstrate the robustness of our sampler's performance compared to other MCMC methods dealing with manifold valued distributions through extensive numerical experiments, on both synthetic and real data.
In particular, we illustrate its remarkable ability to cope with anisotropic target densities,
without using gradient information and preconditioning. 
\end{abstract}


\section{Introduction}
 Parameters in statistical models are often constrained to ensure identifiability or to incorporate expert knowledge \citep{holbrook2016bayesian,holbrook2020nonparametric, mantoux2021understanding}.
 These constraints impose a manifold structure on the parameter space, providing a geometric interpretation of the model.
Extracting information with Bayesian inference from such models requires the ability to sample (at least
approximately) from the usually highly intractable posterior distribution on the manifold \citep{lieAccepteddimension}.
However, due to the involved geometric features, standard $\mathbb{R}^d$ approaches are not directly applicable \citep{zappa2018monte}.
This highlights the importance of developing efficient sampling techniques specifically designed for manifolds.

In this paper we tackle this problem and consider a target measure $\pi$ on a general Riemannian manifold $\mathsf{M}$ that has a density with respect to the Riemannian measure $\nu_{\mathfrak{g}}$ of the manifold, i.e., that is of the form
\begin{equation}\label{Eq: Target - introduction}
	\pi(\d x) = \frac{p(x)}{\int_{\mathsf{M}} p(y)\, \nu_{\mathfrak{g}}(\d y)} \nu_{\mathfrak{g}}(\d x),
\end{equation}
where $p : \mathsf{M} \to [0, \infty)$ satisfies $\int_{\mathsf{M}} p(y) \ \nu_{\mathfrak{g}}(\d y) \in (0, \infty)$.
Such frameworks are encountered  in various applications, 
e.g., brain connectivity network analysis \citep{mantoux2021understanding}, dimensionality reduction \citep{holbrook2016bayesian}, computer vision \citep{lui2012advances}, texture analysis \citep{kunze2004bingham} and protein conformation modelling \citep{hamelryck2006sampling}.
In most instances, the Riemannian manifolds that arise are matrix manifolds, such as the Stiefel manifold or the Grassmann manifold \citep{edelman1998geometry}.

One popular way to approximately sample from intractable distributions are Markov chain Monte Carlo (MCMC) methods. 
We follow this approach and propose a practical MCMC algorithm based on slice sampling techniques.
This method, that we refer to as \emph{geodesic slice sampler} (GSS), incorporates the geometry of the underlying Riemannian manifold by using the geodesics.
We briefly describe the transition mechanism of GSS.
It crucially exploits the fact that on a Riemannian manifold, for every pair $(x,v)$, where $x \in \mathsf{M}$ and $v\in T_x\mathsf{M}$ is an element of the tangent space at $x$, 
there exists a unique geodesic $\gamma_{(x,v)}$ emanating from $x$ in the direction $v$.

More precisely, given the current state $x \in \mathsf{M}$ with $p(x)> 0$, a single transition of the GSS targeting the distribution $\pi$ proceeds in three steps.
First, a level $t$ is uniformly sampled from the interval $(0, p(x))$.
Second, a geodesic $\gamma_{(x,v)}$ that passes through $x$ is randomly chosen.
Lastly, a point is generated from the intersection of the geodesic $\gamma_{(x,v)}$ and the level set $L(t) := \{ y \in \mathsf{M} \mid p(y) > t\}$. This final step presents the most significant challenge and requires special care to ensure invariance of  the target distribution. To this end, we carefully adapt Neal's stepping-out and shrinkage procedure \cite{neal2003slice} to our manifold setting.

Random walk Metropolis-Hastings, Hamiltonian Monte Carlo and Langevin-type algorithms have already been adapted to the Riemannian manifold setting.
Moreover, there exist several tailor-made algorithms for specific manifolds.
Consult Section \ref{Sec: Literature review} for a literature review of  MCMC-methods on Riemannian manifolds.
However, to the best of our knowledge GSS represents the first practical slice sampling-based MCMC method applicable to general Riemannian manifolds.
Following a slice sampling paradigm is appealing because, by design, the length of the transition step is fitly chosen for each transition.
This is especially advantageous when efficiently exploring the target distribution requires  varying the length of transition steps based on the position and direction of the move, e.g., due to anisotropy.
It is worth noting that some slice sampling algorithms achieve this without any hyperparameters, that need to be tuned. 
However, they are usually introduced by a practical implementation.
Nonetheless, it is expected that the resulting slice sampler algorithms' performance will be more robust with regard to the choice of these parameters compared to, for example, the sensitivity of a Metropolis-Hastings or Hamiltonian Monte Carlo algorithm to the selection of step size. We refer to \cite{neal2003slice} for some more details on these properties of slice sampling; see also
\citep{murray2010elliptical} for a further discussion of advantages and limitations of a slice sampling approach.

To conclude this introduction, our main contributions can be summarized as follows:
\begin{itemize}
	\item We propose a slice sampling based MCMC-method, which we call geodesic slice sampling, to target distributions of the form \eqref{Eq: Target - introduction}.
	It combines a geodesical Hit-and-run algorithm with a 1-dimensional slice sampler arriving at a method that generalizes Hit-and-run slice sampling 
	to Riemannian manifolds.
	\item We demonstrate the applicability of GSS in numerical experiments and evaluate its strong suits as well as its drawbacks. Its main feature is the robustness of its performance to the shape
	of the target, without using gradient information and with an easy tuning of parameters using the diameter of the manifold.
	\item 
	 We verify the correctness of GSS by showing 	reversibility with respect to $\pi$ and convergence to $\pi$ in total variation distance.
\end{itemize}

The structure of the paper is as follows.
First we provide an introduction to slice sampling on $\mathbb{R}^d$ in Section \ref{Sec: Slice sampling on Rd}, before turning to GSS on general Riemannian manifolds in Section \ref{Sec: Geodesic slice sampling}.
This is followed by a literature review of  MCMC-methods on Riemannian manifolds in Section \ref{Sec: Literature review}.
Readers that are interested in more details on the differential geometry background used in Section \ref{Sec: Methodology} may find it in Supplementary material \ref{Sec: Manifolds}.
Section \ref{Sec: Applications} is devoted to numerical experiments.
 Further numerics can also be found in Supplementary material \ref{Sec: Numerics supplement}.
The proof of the reversibility and ergodicity of GSS with respect to $\pi$ is given in Supplementary material \ref{Sec: Validity}.
A formal treatment of the stepping-out and shrinkage procedure is also included there.

\subsection{General notation}

We introduce some general notation that is valid throughout the whole paper.
Let $ \mathbb{N} $ be the set of strictly positive integers and call $ \mathbb{N}_0 := \mathbb{N} \cup \{0\} $.
We denote by $\Leb_d$ the $d$-dimensional \emph{Lebesgue measure} on $\mathbb{R}^d$
and by $\mathcal{B}(\mathbb{R}^{d})$ the \emph{Borel-$\sigma$-algebra}.
Similar, for a set $\mathsf{S} \in \mathcal{B}(\mathbb{R}^{d})$ we write $\mathcal{B}(\mathsf{S})$ for the trace Borel-$\sigma$-algebra on $\mathsf{S}$.
Moreover, we set $\mathbb{S}^{d-1} := \{ x \in \mathbb{R}^d \mid \|x\| = 1\}$ to be the $d-1$-dimensional \emph{Euclidean unit sphere}.
For $x \in \mathbb{R}^d$, let $x^\top$ be its transpose, and write $\mathrm{Id}_d \in \mathbb{R}^{d \times d}$ for the \emph{identity matrix}.
For any set $\mathsf{S} \in \mathcal{B}(\mathbb{R}^{d})$, finite or
satisfying $\Leb_d(\mathsf{S})\in (0, \infty)$, 
denote the discrete, respectively continuous, \emph{uniform distribution} on $\mathsf{S}$ as $\mathrm{Unif}(\mathsf{S})$.
Whenever we introduce random variables in the sequel, we assume them to be defined on some rich enough probability space $(\Omega, \mathcal{F}, \mathbb{P})$.
Let $(\mathsf{X}, \mathcal{X})$ be a measurable space and let $R$ be an $(\mathsf{X}, \mathcal{X})$-valued random variable.
Then we denote by $\mathbb{P}^{R}(\mathsf{A}) := \mathbb{P}(R \in \mathsf{A})$, $\mathsf{A} \in \mathcal{X}$,
the \emph{distribution of $R$}.
If $\mathbb{P}^R = \mu$ for some probability measure $\mu$, then we also write $R \sim \mu$.
For a set $\mathsf{A} \in \mathcal{X}$, its corresponding indicator function $\mathbbm{1}_{\mathsf{A}}$ is defined as $\mathbbm{1}_{\mathsf{A}}(x) = 1$ if $x \in \mathsf{A}$ and $\mathbbm{1}_{\mathsf{A}}(x) = 0$ if $x \notin \mathsf{A}$.
The Dirac measure at a point $x \in \mathsf{X}$ is denoted as $\delta_x$, i.e. we have $\delta_x(\mathsf{A}) = \mathbbm{1}_{\mathsf{A}}(x)$ for all $\mathsf{A} \in \mathcal{X}$.
 Given another probability measure $\eta$ on $(\mathsf{X}, \mathcal{X})$, we denote by $d_{\mathrm{tv}}(\mu, \eta)=\sup_{\mathsf{A} \in \mathcal{X}} |\mu(\mathsf{A})- \eta(\mathsf{A})|$ the \emph{total variation distance} between $\mu$ and $\eta$.
For a Markov kernel $K: \mathsf{X} \times\mathcal{X} \to [0,1]$ define inductively over the natural numbers $K^1 = K$ and $K^{n+1}(x,\mathsf{A}) = \int_{\mathsf{X}} K^n(y,\mathsf{A})\ K(x, \d y)$ for all $x \in \mathsf{X}$ and $\mathsf{A} \in \mathcal{X}$.
Furthermore, let $\mathsf{Y}$ be a set and let $f: \mathsf{X} \to \mathsf{Y}$ be a map from $\mathsf{X}$ to $\mathsf{Y}$.
For some set $\mathsf{S} \subseteq \mathsf{X}$, we denote by $f\vert_{\mathsf{S}}$ the \emph{restriction of $f$ to $\mathsf{S}$}.
Now equip also $\mathsf{Y}$ with a $\sigma$-algebra and turn it into the measurable space $(\mathsf{Y}, \mathcal{Y})$.
Additionally assume $f$ to be measurable and let $\mu$ be a measure on $(\mathsf{X}, \mathcal{X})$.
We call $f_\sharp\mu(\mathsf{A}) := \mu(f^{-1}(\mathsf{A}))$, $\mathsf{A} \in \mathcal{Y}$,
the \emph{push forward measure} of $\mu$ under $f$, and for $\mathsf{S} \in \mathcal{X}$ we call $\mu\vert_{\mathsf{S}}(\mathsf{A}) := \mu(\mathsf{S} \cap \mathsf{A})$, $\mathsf{A} \in \mathcal{X}$,
the \emph{restriction of $\mu$ to $\mathsf{S}$.}
To emphasize that a union is disjoint we write $\sqcup$.
 We denote the Gaussian distribution on $\mathbb{R}^d$ with mean $x \in \mathbb{R}^{d}$ and covariance matrix $\Sigma \in \mathbb{R}^{d\times d}$ by $\mathcal{N}(x, \Sigma)$. 
For $\kappa \in \mathbb{R}^{d}$, let $\mathrm{diag}(\kappa) \in \mathbb{R}^{d \times d}$ be the diagonal matrix with diagonal entries given by the components of $\kappa$.
Moreover, we define 
$\boldsymbol{0}_d \in \mathbb{R}^d$ to be the vector and 
$\boldsymbol{0}_{d,k}\in \mathbb{R}^{d\times k}$ to be the matrix with all entries zero, and write $\mathrm{tr}(\Sigma)$ for the trace of a square matrix $\Sigma$.
Finally, for simplicity we also use $\land$ and $\lor$ to denote the minimum respectively the maximum between two real numbers, i.e., $r \land s := \min\{r,s\}$ and $r \lor s := \max\{r,s\}$ for $r,s \in \mathbb{R}$.

\section{Methodology: Geodesic slice sampling}\label{Sec: Methodology}

\subsection{Slice sampling on $\mathbb{R}^d$}\label{Sec: Slice sampling on Rd}

In this section we revisit slice sampling on $\mathbb{R}^d$ in order to provide a general introduction to the ideas of (uniform simple) slice sampling.
For simplicity, we assume the target density to be strictly positive throughout this subsection.
The slice sampler, as all MCMC-methods, defines a Markov chain which, for a given measurable unnormalized density $p : \mathbb{R}^d \to (0, \infty)$
that satisfies $\int_{\mathbb{R}^d} p(y) \Leb_d(\d y) \in (0,\infty)$,
can be used to approximately sample from the probability measure
\[
\pi(\d x)= \frac{p(x)}{ \int_{\mathbb{R}^d} p(y) \Leb_d(\d y)} \ \Leb_d(\d x).
\]
At its heart are the \emph{(super) level sets} of $p$ defined by
\[
L(t) := \{x \in \mathbb{R}^d \mid p(x)> t\}, \qquad t \in (0, \infty),
\]
which contain all points that have a function value with respect to $p$ that
is greater than a specified value.
The uniform simple slice sampler, which we also call \emph{idealized slice sampler}, generates approximate samples from $\pi$ by drawing suitably from these level sets.

We denote by $(Y_k)_{k \in \mathbb{N}}$ the Markov chain that corresponds to the idealized slice sampler.
A transition from $Y_k = x \in \mathbb{R}^d$ to $Y_{k+1}$ works as follows:
\begin{enumerate}
	\item Draw a random level $T_{k+1} \sim \mathrm{Unif}\big((0, p(x))\big)$, call the result $t$. This specifies a level set $L(t)$.
	\item Draw $Y_{k+1}\sim \mathrm{Unif}(L(t))$ uniformly from this level set $L(t)$.
\end{enumerate}
A graphic representation of the transition mechanism for $d=1$ can be found in Figure \ref{Fig: Idealized slice sampling}.
\begin{figure}
	\centering
	\subfloat[Sample the level $t$ uniformly from $(0, p(x))$.]{\includegraphics[width=0.3\textwidth]{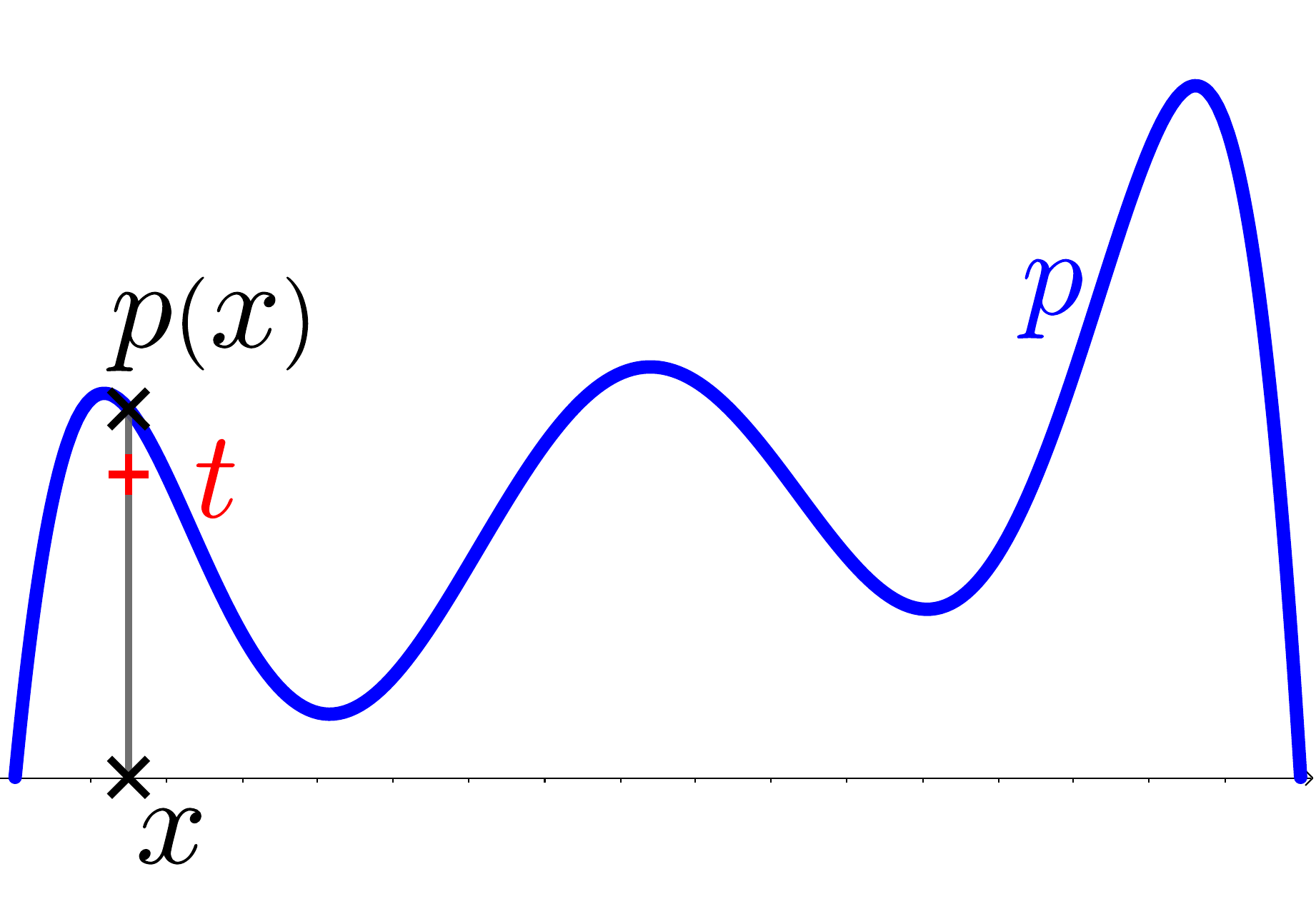}}
	\hspace*{0.015\textwidth}
	\subfloat[The level $t$ specifies a level set $L(t)$.]{\includegraphics[width=0.3\textwidth]{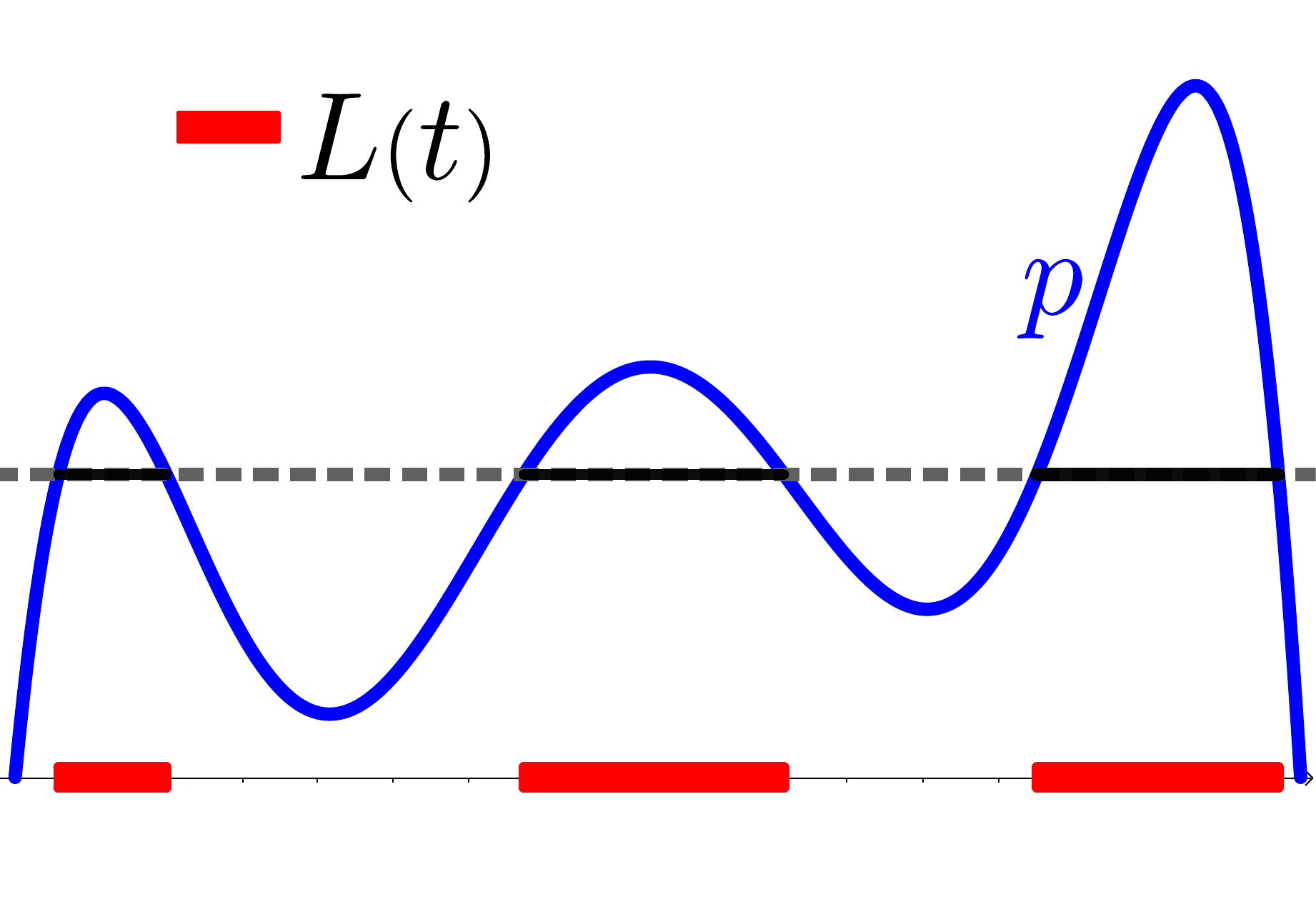}}
	\hspace*{0.015\textwidth}
	\subfloat[Sample the next point $y$ uniformly from $L(t)$.]{\includegraphics[width=0.3\textwidth]{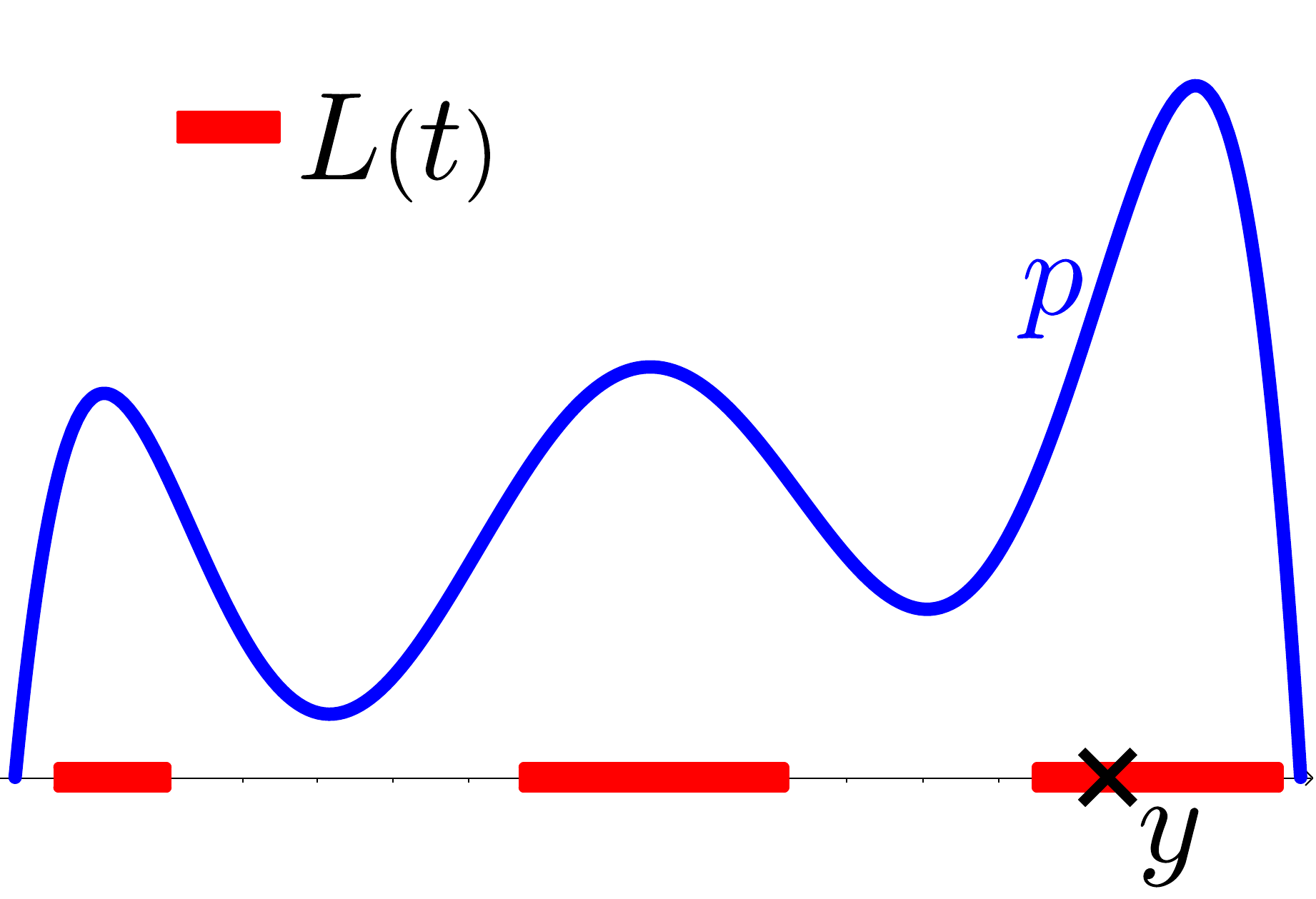}}%
	\caption{Transition mechanism of the idealized slice sampler for initial point $x$.}
	\label{Fig: Idealized slice sampling}
\end{figure}
We can also describe the idealized slice sampler through its transition kernel given by
\begin{align*}
	H: \mathbb{R}^d\times \mathcal{B}(\mathbb{R}^d) &\to [0,1]\\
	(x,\mathsf{A}) & \mapsto \frac{1}{p(x)} \int_{(0, p(x))} \frac{1}{\Leb_d\big(L(t)\big)}\int_{L(t)} \mathbbm{1}_{\mathsf{A}}(y)\ \Leb_d(\d y)\, \Leb_1(\d t).
\end{align*}
Since slice sampling was brought to the attention of the statistics community in \citep{besag1993spatial}, 
the properties of idealized slice sampling, 
have been investigated in several works, e.g., \citep{mira2002efficiency,natarovskii2021quantitative,roberts2002convergence, rudolf2013positivity, rudolf2018comparison}.  
The following result illustrates a major advantage of slice sampling:
\citet[Corollary 3.7]{natarovskii2021quantitative} show that the spectral gap of $H$ only depends on the level set function $t \mapsto \Leb_d(L(t))$, i.e., on the volume of the level sets, not their shape.
This means that the performance of the idealized slice sampler is ignorant of the introduction of, e.g., multimodality, local modes or anisotropy as long as the volume of the level sets is not modified.

Unfortunately, each transition of the idealized slice sampler requires to sample from the uniform distribution $\mathrm{Unif}(L(t))$, $t > 0$, of a level set.
Since in general there is no efficient algorithm to tackle this problem,
this is a major prevention for the implementation of the idealized slice sampler.
One modification strategy to obtain a practical algorithm is called \emph{hybrid slice sampling}, see \citep{latuszynski2024convergence}.
Here, the uniform distribution on the level sets is replaced by a family of kernels $(H_t)_{t > 0}$
such that for any $t>0$ (where it is well-defined) $\mathrm{Unif}(L(t))$ is invariant for $H_t$.
This leads to a transition kernel of the form
\[
(x,\mathsf{A}) \mapsto \frac{1}{p(x)} \int_{(0, p(x))} H_t(x,\mathsf{A})\ \Leb_1(\d t), \qquad x \in \mathbb{R}^d, \mathsf{A} \in \mathcal{B}(\mathbb{R}^d).
\]
A concrete example for the case $d=1$ is the \emph{stepping-out} and a \emph{shrinkage} procedure introduced in \citep{neal2003slice}. 
To provide a basic understanding of these two schemes,
we give a verbal description and a visualization in Figure \ref{Fig: Step-out and shrink}.
The stepping-out procedure is treated in more detail 
in Supplementary material \ref{Sec: Stepping-out}. This includes pseudocode and a careful definition of the generated distributions.
For an extensive treatment of the shrinkage procedure, we refer to \citep{ReversibilityEllipticalSliceSampler}.

The stepping-out and shrinkage based hybrid slice sampler takes a point $x \in \mathbb{R}$ (current state of Markov chain) and a level set $L(t)$ (level generated as in the first step of the transition mechanism of the idealized slice sampler), and proceeds in two steps.
\begin{description}
	\item [ \textbf{1. Stepping-out:} ] Under the specification of two hyper parameters $w > 0$ and $m \in \mathbb{N}$, the stepping-out procedure chooses a random segment of $\mathbb{R}$ containing $x$. 
	To this end, an interval of length $w$ is placed randomly around $x$ by sampling the left interval boundary point $L_1 \sim \mathrm{Unif}\big((x-w,x)\big)$
	and setting the right interval boundary to $R_1 = L_1 + w$.
	Then this interval is extended iteratively to the left by intervals of length $w$ until for the first time the left boundary leaves the level set $L(t)$, or the maximal number of stepping-out steps to the left $\upiota$ is reached.
	Here
	$\upiota$ is obtained by randomly splitting $m+1$, the maximal number of total stepping-out steps, into two summands $\upiota$ and $m +1 - \upiota$. 
	Similarly the interval is extended iteratively to the right by intervals of length $w$ until the right boundary hits $\mathbb{R}\setminus L(t)$ for the first time, or $m+1-\upiota$ steps have been performed. 
	This provides a randomly generated interval $I = (L^\ast,  R^\ast)$, where 
	\[
	L^\ast = L_1 - \big(\uptau_\ell \land (\upiota-1) \big)w \qquad \text{and} \qquad R^\ast =  R_1 + \big(\uptau_r\land (m - \upiota)\big)w
	\]
	with 
	\[
	\uptau_\ell := \inf\{k \geq 0 \mid L_1 - k w \notin L(t)\} \qquad \text{and} \qquad \uptau_r:= \{k \geq 0 \mid R_1 + kw \notin L(t)\}.
	\]
	\item [ \textbf{2. Shrinkage:} ] We generate a point from $I\cap L(t)$ with the shrinkage procedure.
	Roughly speaking, it is an adaptive acceptance/rejection scheme
	that shrinks the proposal area with each rejection.\footnote{For convenience of applying the results from \cite{ReversibilityEllipticalSliceSampler} we describe here a scheme that differs slightly from the original one in \citep{neal2003slice} and rather resembles the one of the elliptical slice sampler, see \citep{murray2010elliptical}.
		Crucially, we view here the interval $I$  as a circle, i.e., if one ``leaves'' the interval at the right boundary, it is immediately ``reentered'' at the left boundary.}
	Observe that a set
	$\mathsf{J} \subseteq \mathbb{R}$  (viewed as a circle) is divided by two points $y,z \in \mathsf{J}$ into two segments, namely $(y \land z, y \lor z)\cap \mathsf{J}$ and $\mathsf{J}\setminus (y \land z, y \lor z)$.
	If  $\mathsf{J}$ contains the initial point $x$, we set
	\[
	\mathbb{J}(y,z,\mathsf{J}) :=\begin{dcases}
		(y \land z, y \lor z)\cap \mathsf{J},& \text{if } x \in (y \land z, y \lor z),\\
		\mathsf{J} \setminus (y \land z, y \lor z), & \text{otherwise},
	\end{dcases}
	\]
	to be the segment containing $x$.
	The shrinkage procedure now builds a sequence of such segments.
	First we sample $Y_1 \sim \mathrm{Unif}(I)$, and set $\mathsf{J}_1 = I$.
	For $k \in \mathbb{N}$, we let $Y_{k+1}$ be a random variable with conditional distribution
	\[
	\mathbb{P}\left(Y_{k+1} \in \cdot \mid Y_1, \ldots, Y_k, \mathsf{J}_1, \ldots, \mathsf{J}_k\right) = \mathrm{Unif}(\mathsf{J}_{k}),
	\]
	that is, we draw the next proposal uniformly from the current segment.
	Observe that $Y_{k+1}$ divides $\mathsf{J}_k$ into two segments.
	Set $\mathsf{J}_{k+1} = \mathbb{J}(Y_{k}, Y_{k+1}, \mathsf{J}_{k})$ for $k \in \mathbb{N}$, i.e., we keep the segment of $\mathsf{J}_k$ that contains the initial point $x$.
	This is continued until we generate a proposal that lies in $L(t)$.
	Hence, overall the shrinkage procedure yields a random point $Y^\ast = Y_\uptau$ where $\uptau := \inf\{ k \in \mathbb{N} \mid Y_k \in L(t)\} $.
	If the stepping-out and shrinkage scheme is embedded into a 1-dimensional hybrid slice sampler this point is then the next random variable of the chain.
\end{description}
We comment on the hyperparameters of the stepping-out procedure.
\begin{remark}\label{R: Hyperparameters of stepping-out}
	The output of the stepping-out procedure can be viewed as a ``loose'' approximation of the set $L(t)$.
	The maximal possible length of this approximation is given by $mw$, but also the individual choice of $m$ and $w$ affects the quality of the approximation depending on the shape of $L(t)$.
	Choosing $m$ larger and $w$ smaller can lead to an interval that lies ``tighter'' around $L(t)$.
	However, if $L(t)$ has ``holes'', it is also more likely that parts of $L(t)$ are ``cut off''.
	Moreover, increasing $m$ increases the upper bound for the computational cost of the stepping-out procedure. 
\end{remark}
\begin{figure}
	\centering
	\subfloat[Stepping-out procedure with one step to the left and three steps to the right.]{%
		\makebox[\textwidth]{
			\begin{minipage}{\textwidth}
				\includegraphics[width=0.24\textwidth]{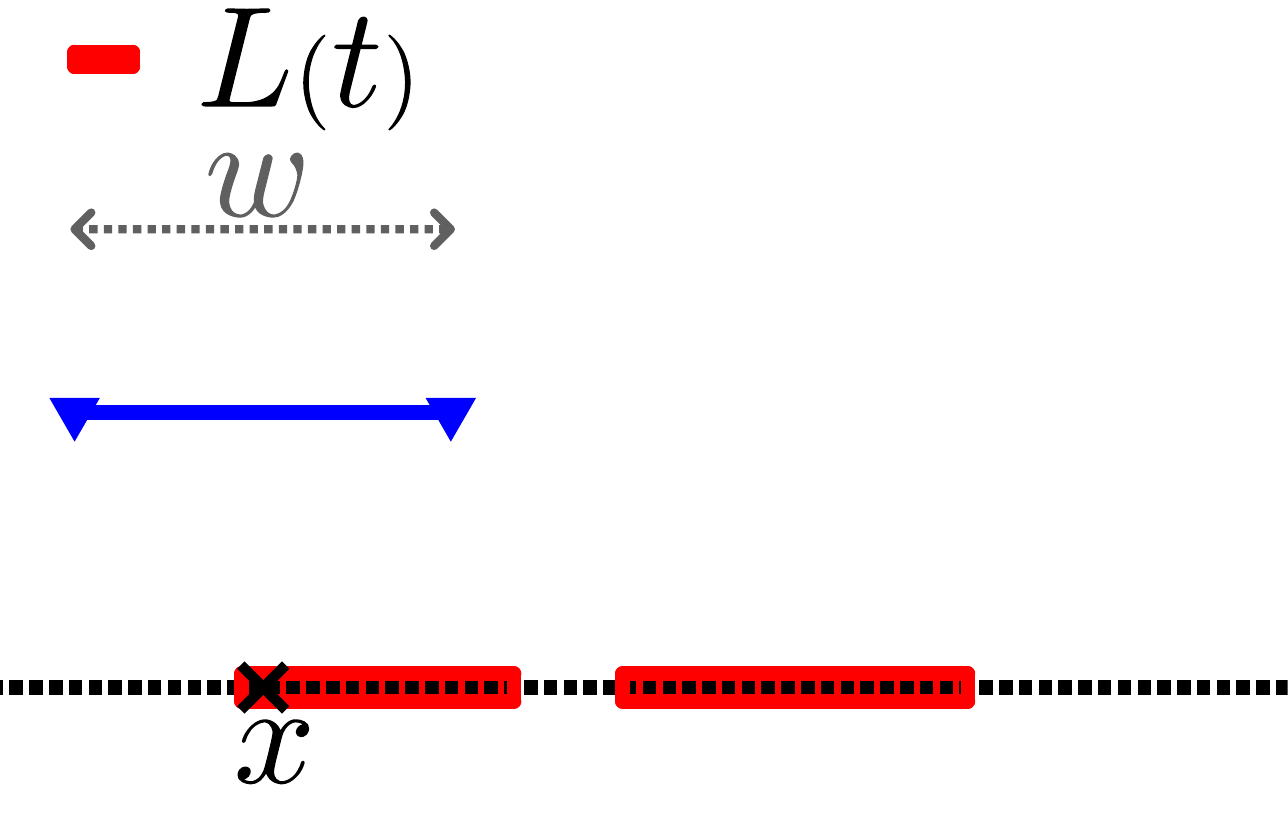}
				\hfill
				\includegraphics[width=0.24\textwidth]{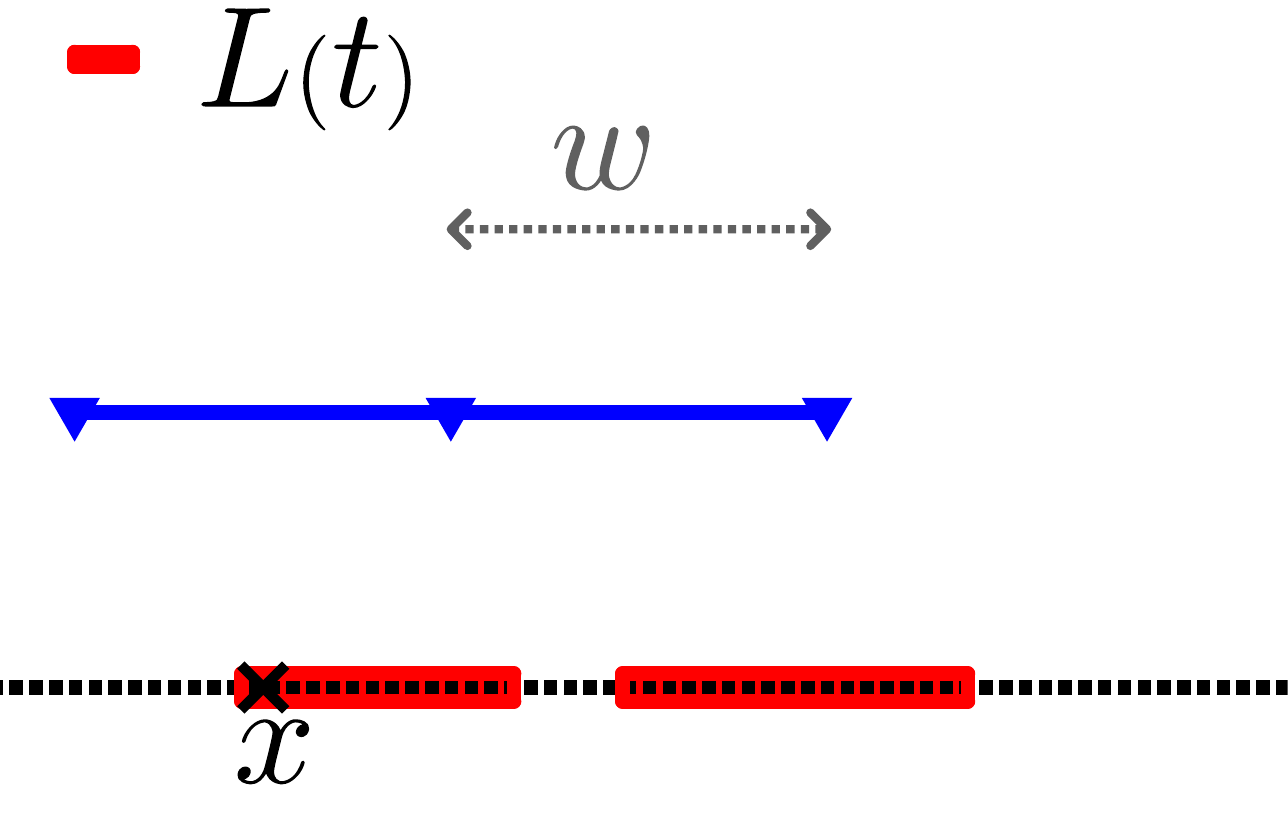}
				\hfill
				\includegraphics[width=0.24\textwidth]{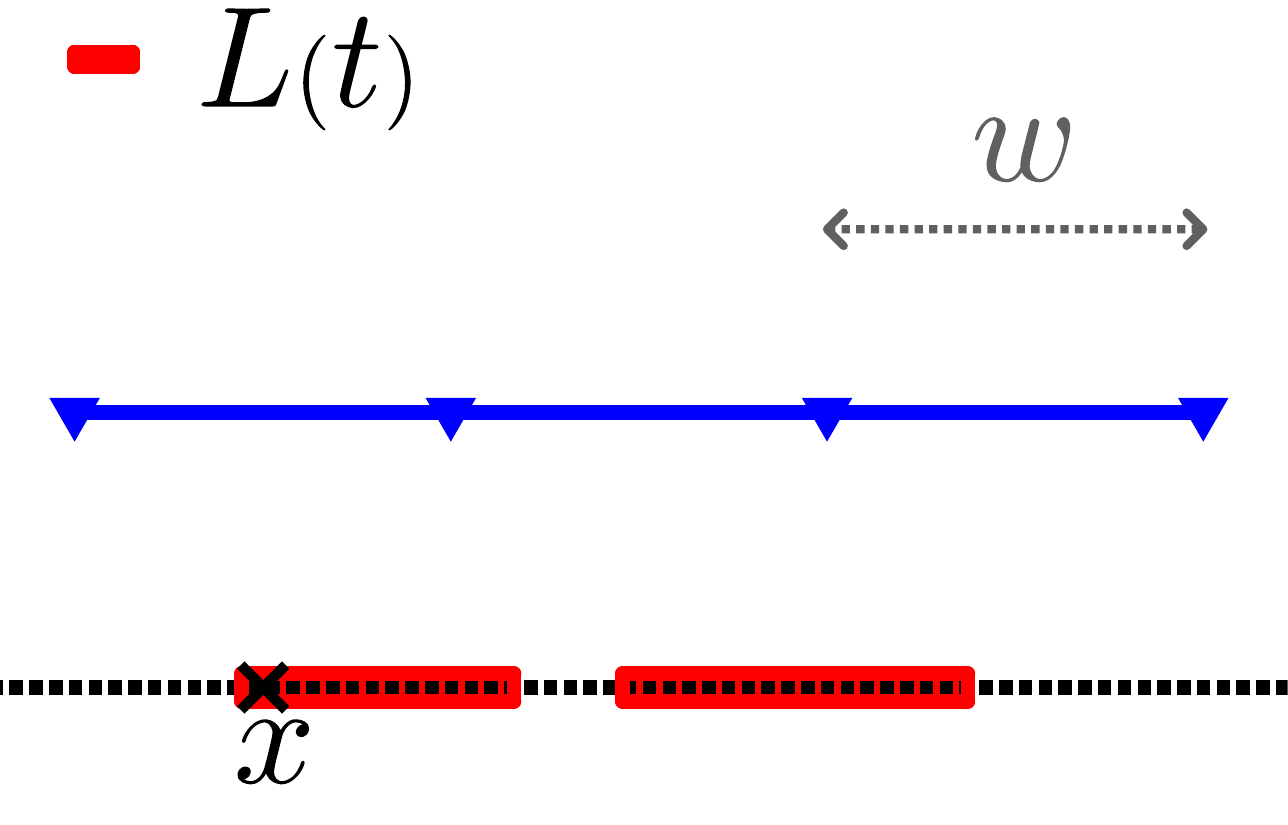}
				\hfill
				\includegraphics[width=0.24\textwidth]{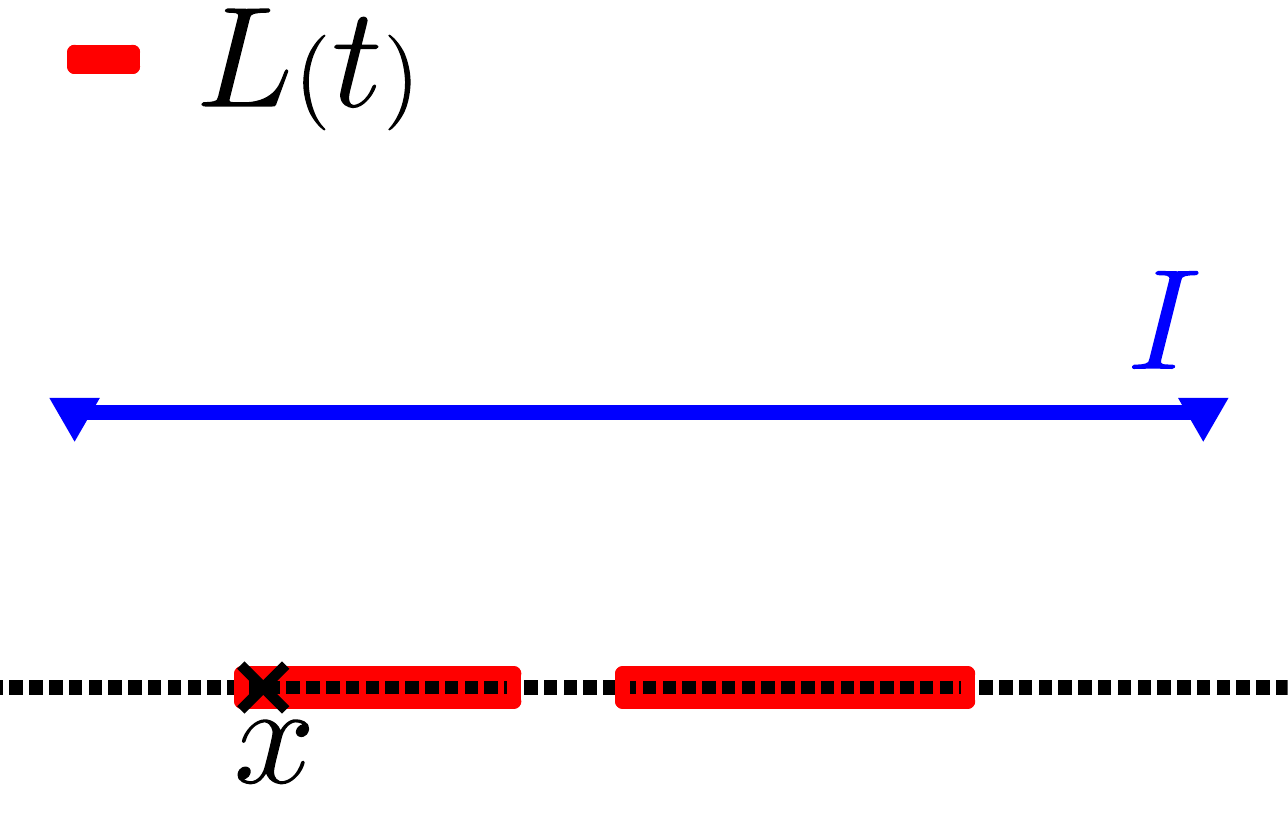}
			\end{minipage}
	}}
	
	\subfloat[Shrinkage procedure where the fourth proposal is accepted.]{%
		\makebox[\textwidth]{
			\begin{minipage}{\textwidth}
				\includegraphics[width=0.24\textwidth]{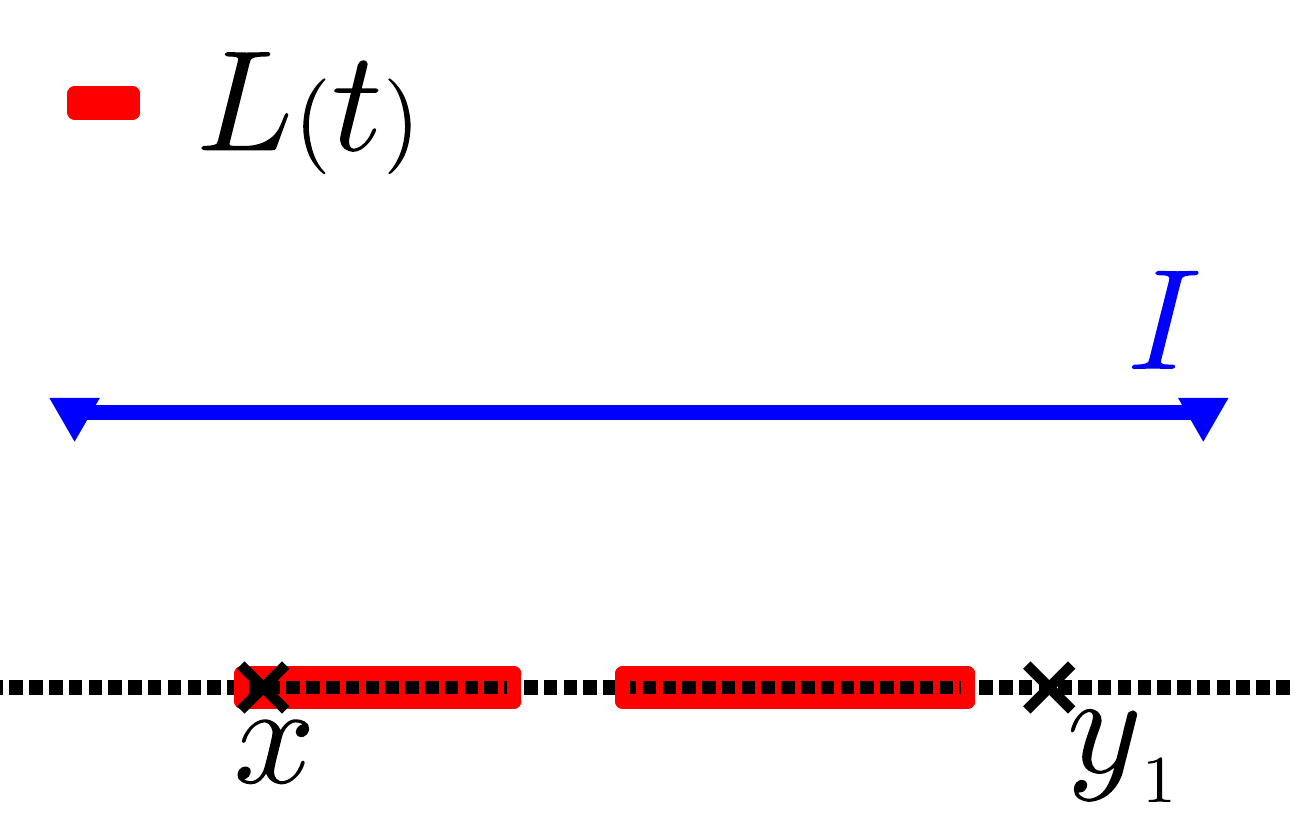}
				\hfill
				\includegraphics[width=0.24\textwidth]{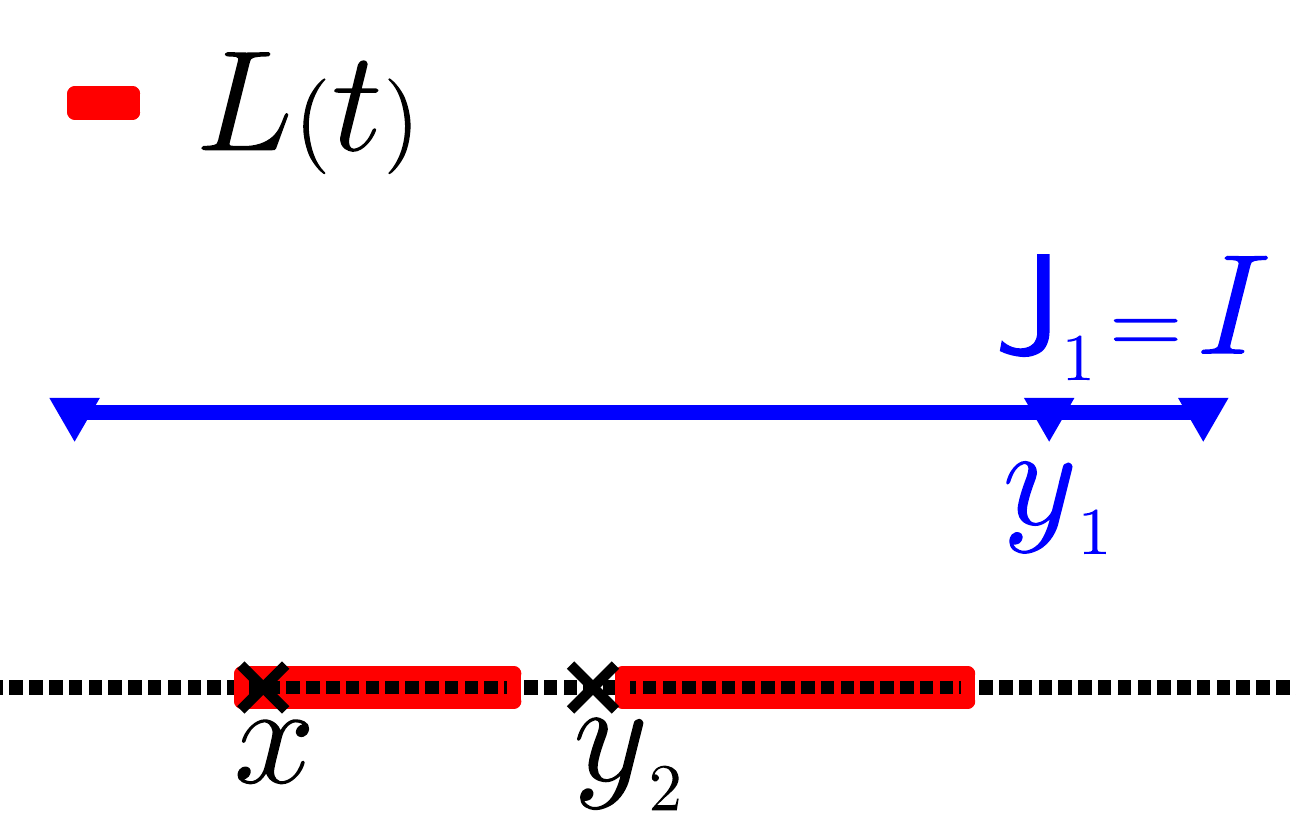}
				\hfill
				\includegraphics[width=0.24\textwidth]{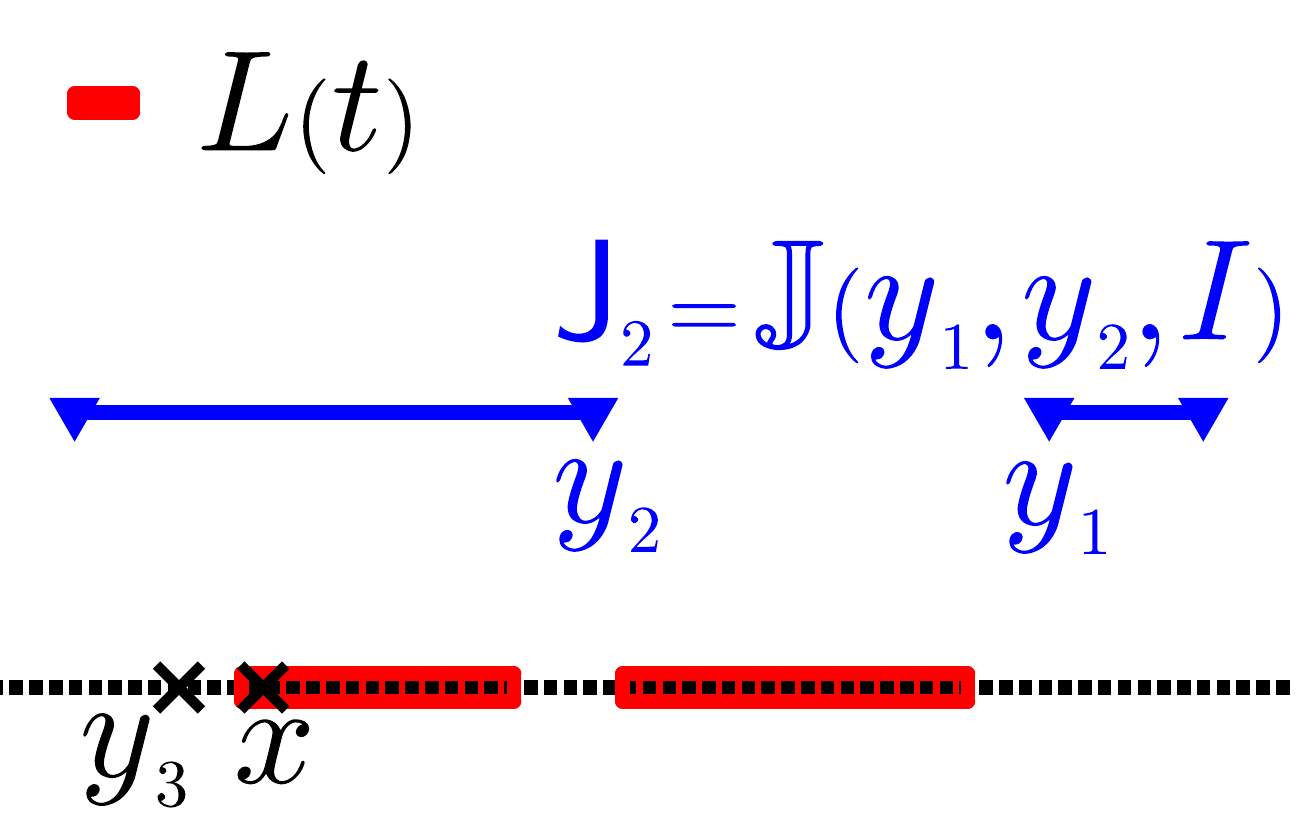}
				\hfill
				\includegraphics[width=0.24\textwidth]{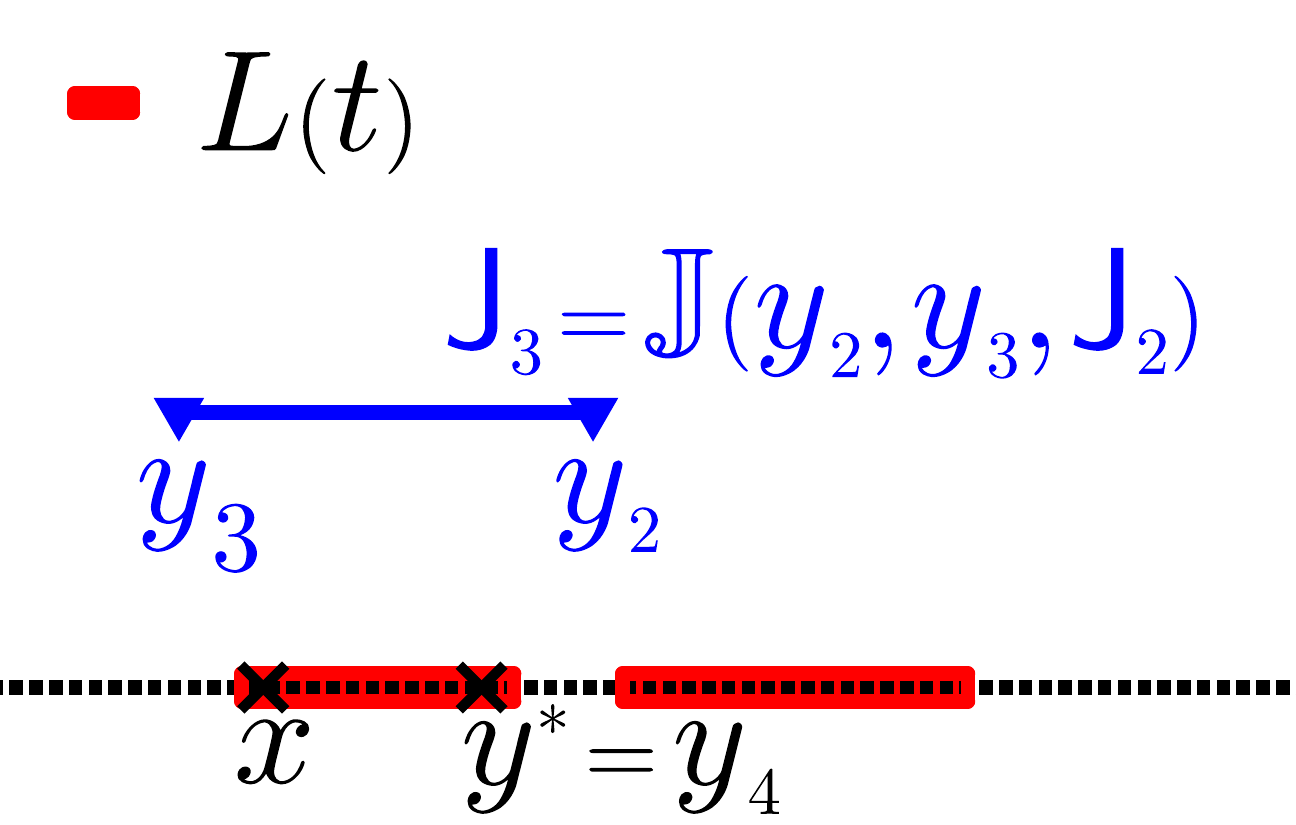}
			\end{minipage}
	}}
	\caption{One run of the stepping-out and shrinkage procedure each.}
	\label{Fig: Step-out and shrink}
\end{figure}
Now we have a practical slice sampling algorithm on $\mathbb{R}$.
One way to lift the stepping-out and shrinkage approach to $\mathbb{R}^d$ is to combine it with the Hit-and-run algorithm arriving at something called the \emph{Hit-and-run slice sampler}\footnote{Hit-and-run slice sampling is already mentioned in \citep[Section 29.7]{Mackay}.
 Convergence and comparison results for this sampler are provided in \citep{latuszynski2024convergence, rudolf2018comparison}, and it is used as a benchmark approach in \citep{murray2010elliptical, schaer2023gibbsian}.},
which essentially samples a random line through the current point and then runs a slice sampler on this line.
For $x \in \mathbb{R}^d$ and $v \in \mathbb{S}^{d-1}$ we define 
\[  
\gamma_{(x,v)}(\theta) = x + \theta v, \qquad \theta \in \mathbb{R},
\]
to be the \emph{line through $x$ in the direction $v$}, and for $t > 0$ 
\[  
L(x,v,t) := \{\alpha \in \mathbb{R} \mid p\left(\gamma_{(x,v)}(\alpha) \right)> t\} = \{\alpha \in \mathbb{R} \mid \gamma_{(x,v)}(\alpha) \in L(t)\}
\]
to be the parametrized intersection of a straight line with a level set.
The transition mechanism of the Hit-and-run slice sampler then proceeds as depicted in Figure \ref{Fig: Hit-and-run slice sampling}.
\begin{figure}[h]
	\begin{tikzpicture}[->,>={Stealth[round,sep]},shorten >=1pt,auto, node distance = 0.5cm,tight-matrix/.style={every outer matrix/.append style={inner sep=+0pt}}]
		
		\node (input) [draw, rectangle, align=center] {\textbf{Input:} current state $x \in \mathbb{R}^d$};
		
		\matrix (dummy1) [tight-matrix, below = of input, nodes = draw, column sep = 0.25\textwidth] {
			\node (level) [align=center, rectangle] {Sample level $t$ uniformly\\ distributed from $(0,p(x))$.}; &
			\node (direction) [ align=center, rectangle] 
			{Sample direction $v$ uniformly\\ distributed from $\mathbb{S}^{d-1}$.};	\\
		};

		\matrix (step-out) [draw, below = of dummy1, column sep=2pt, dotted, very thick, matrix anchor = north, anchor = center]{
			\node (step-out-text) [draw, align=center, rectangle, solid, thin] {Generate $I$ with stepping-out\\ procedure for $L(x,v,t)$.}; &
			\node (step-out-pic) [draw, align=center, rectangle, solid, thin] {
				\includegraphics[width=0.15\textwidth]{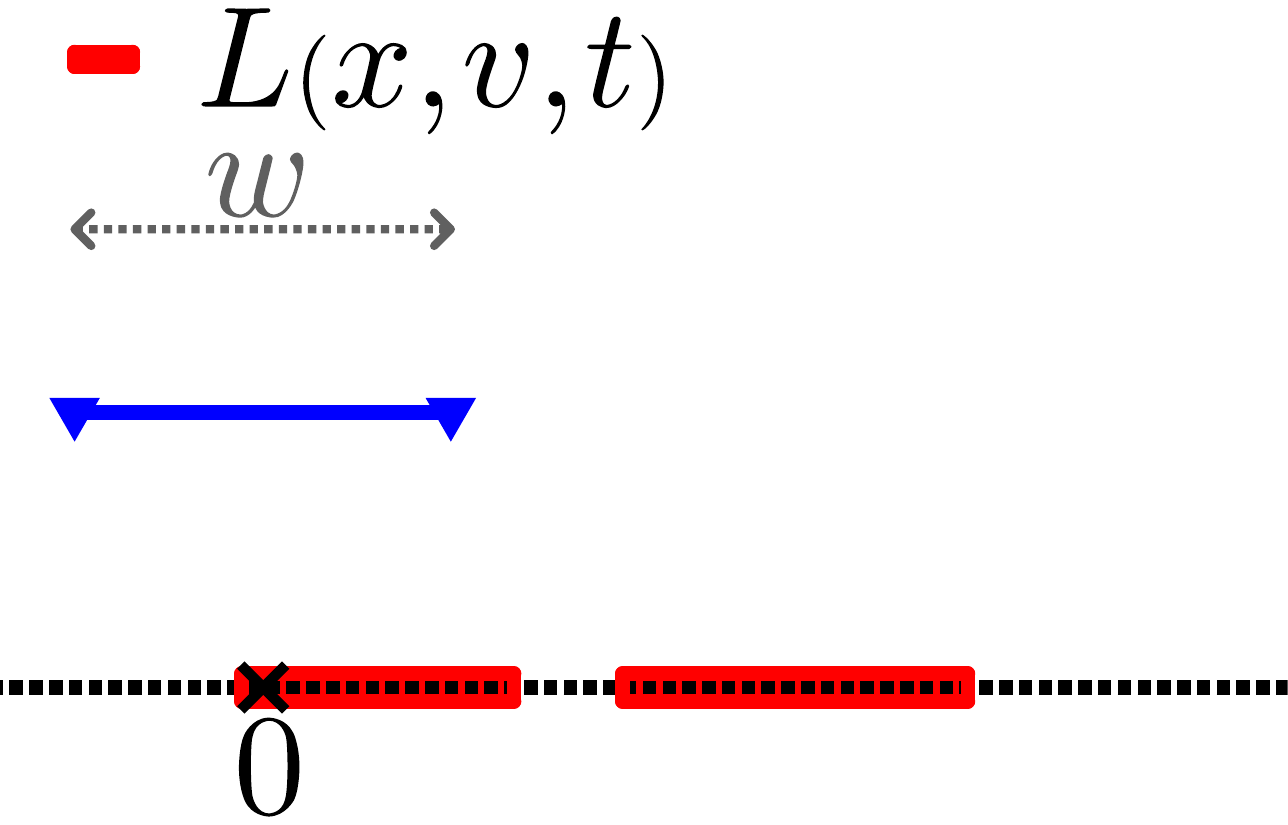}
				\includegraphics[width=0.15\textwidth]{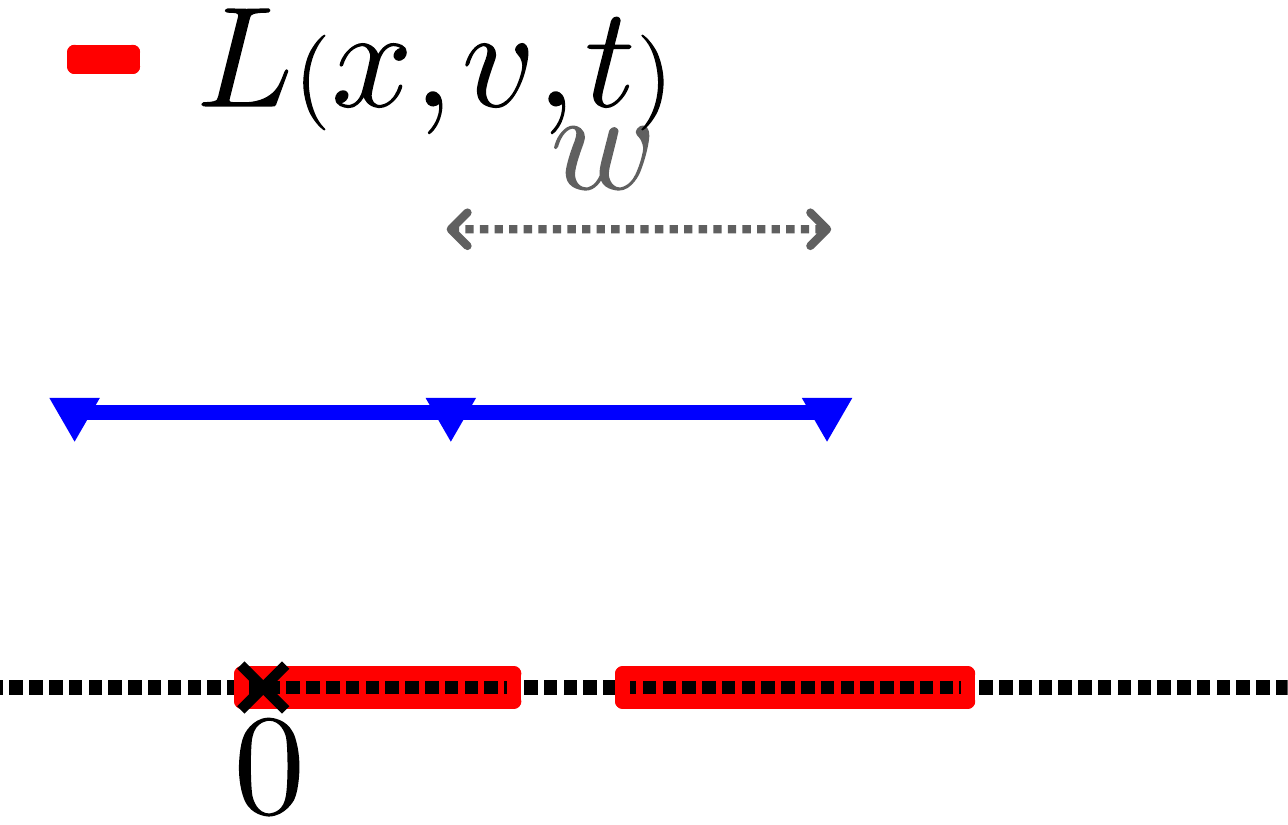}
				\includegraphics[width=0.15\textwidth]{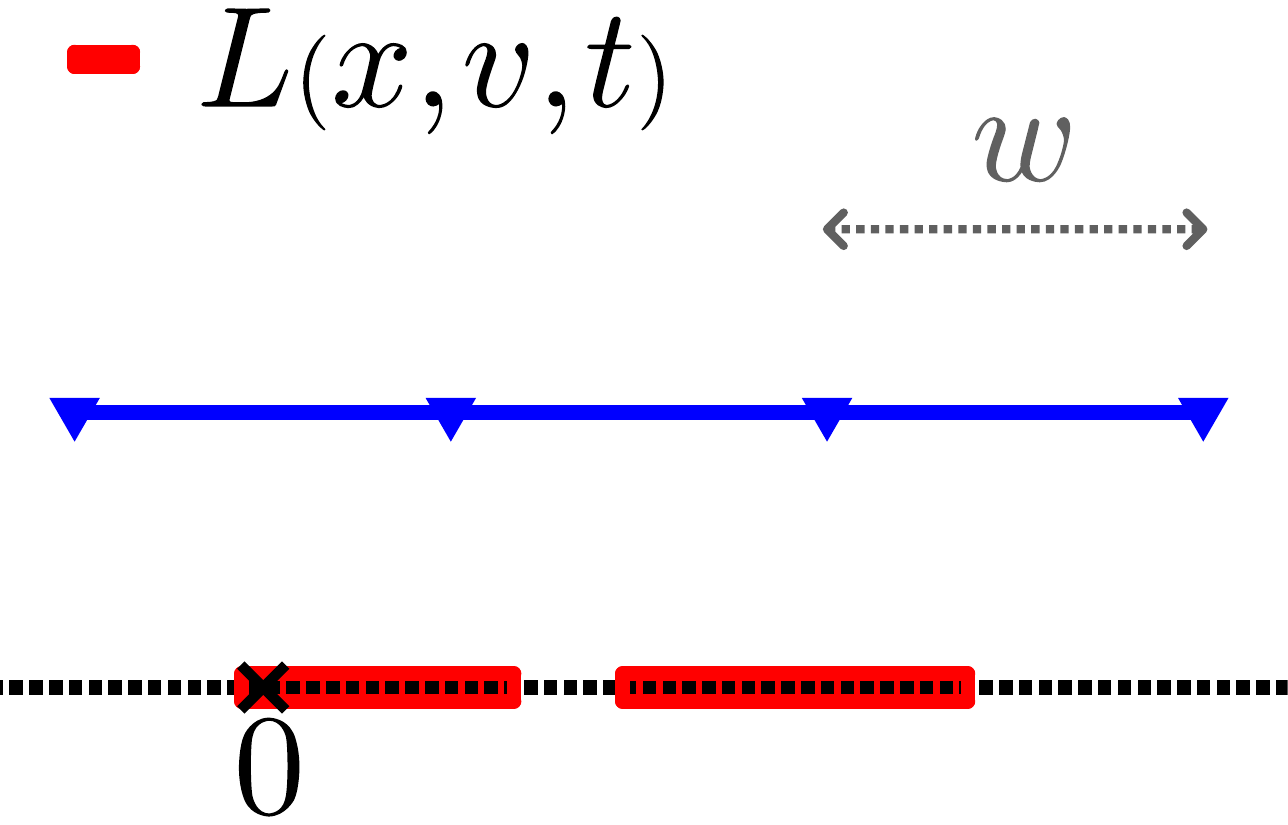}
				\includegraphics[width=0.15\textwidth]{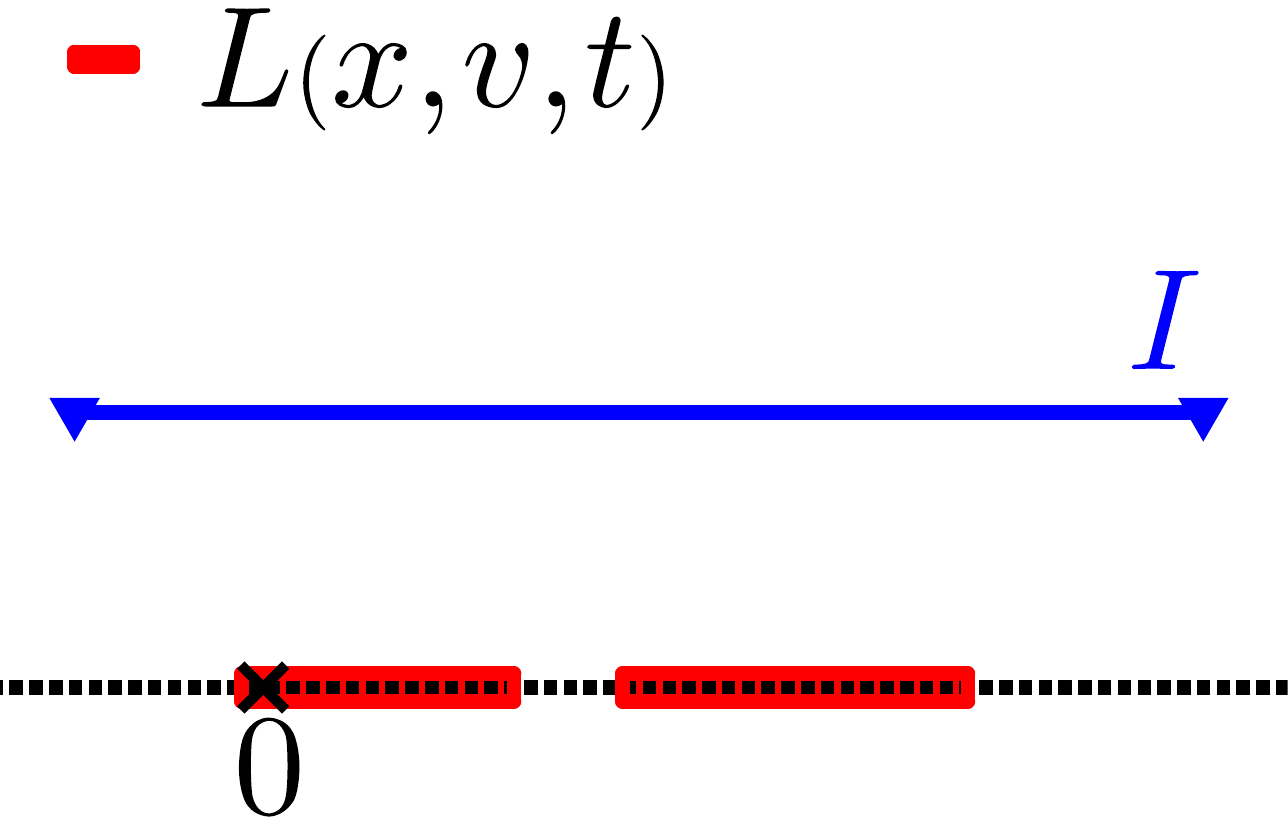}};\\
		};
		
		\matrix (shrink) [draw, below = of step-out, nodes = draw, column sep=2pt, dotted, very thick, matrix anchor = north, anchor = center]{
			\node (shrink-text) [align=center, rectangle, solid, thin] {Generate $\theta$ from $I \cap L(x,v,t)$\\ with shrinkage procedure.}; &
			\node (shrink-pic)[align=center, rectangle, solid, thin]{
				\includegraphics[width=0.15\textwidth]{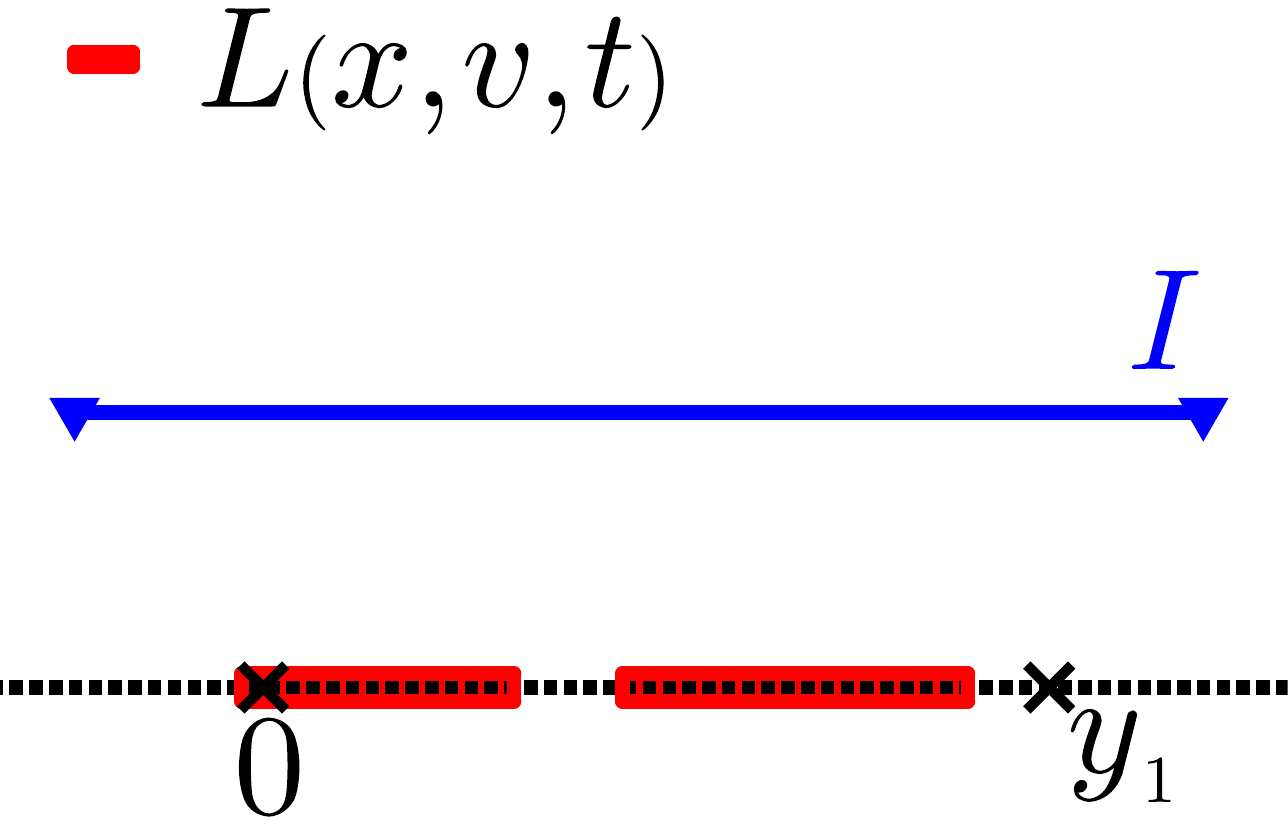}
				\includegraphics[width=0.15\textwidth]{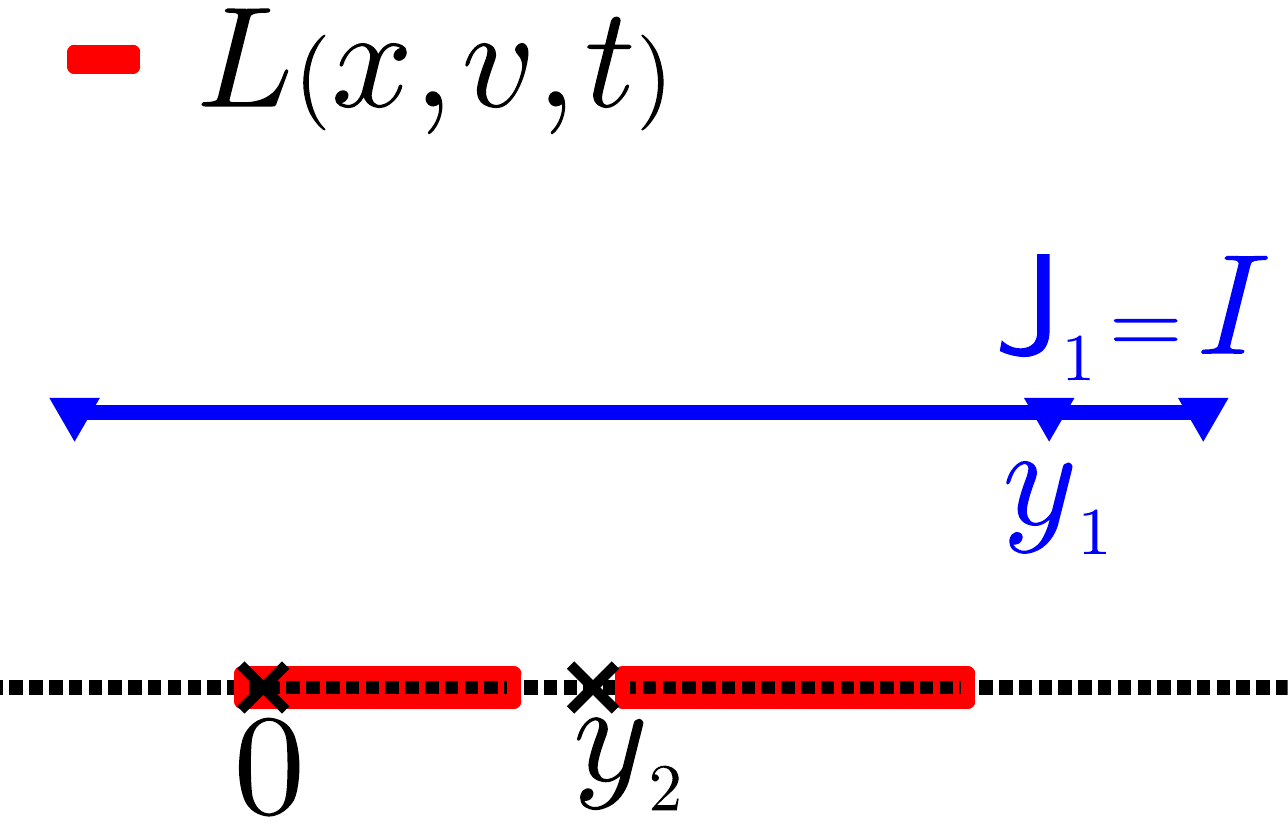}
				\includegraphics[width=0.15\textwidth]{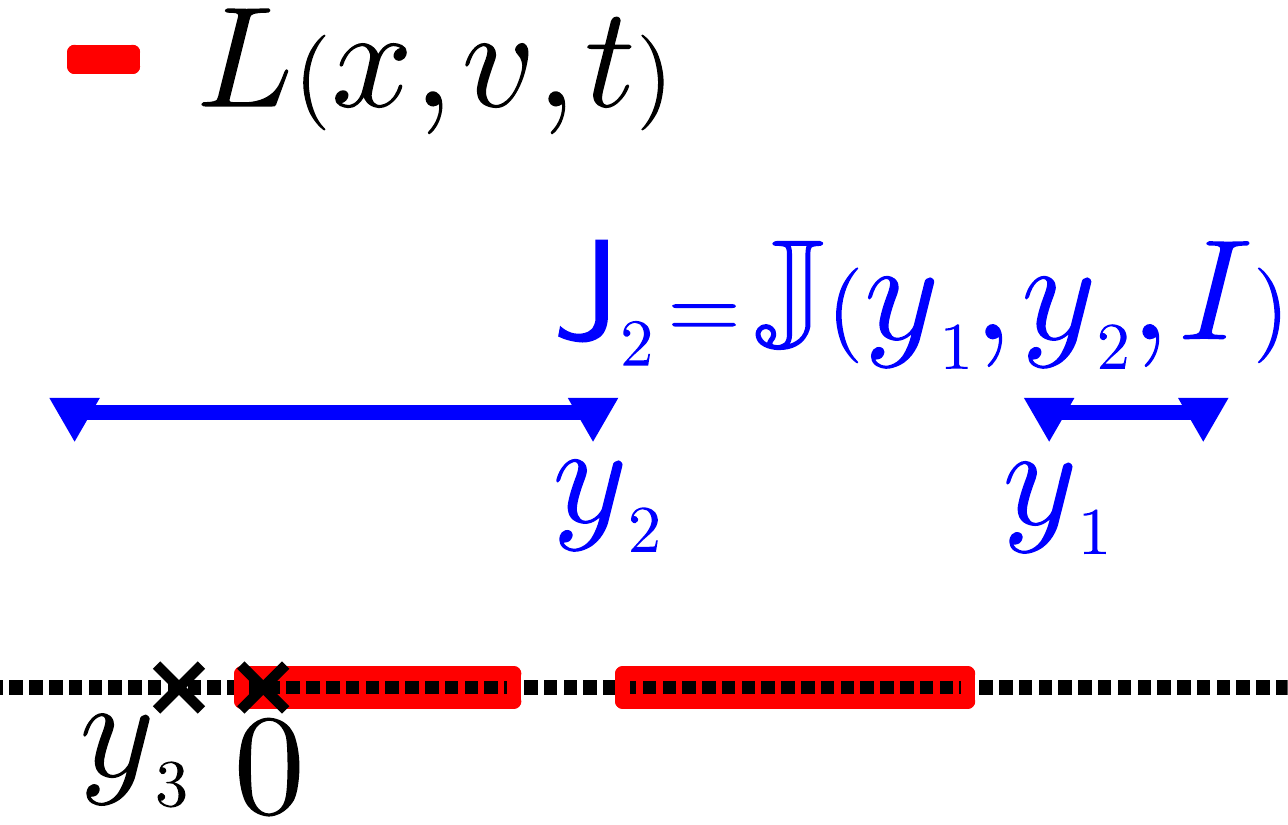}
				\includegraphics[width=0.15\textwidth]{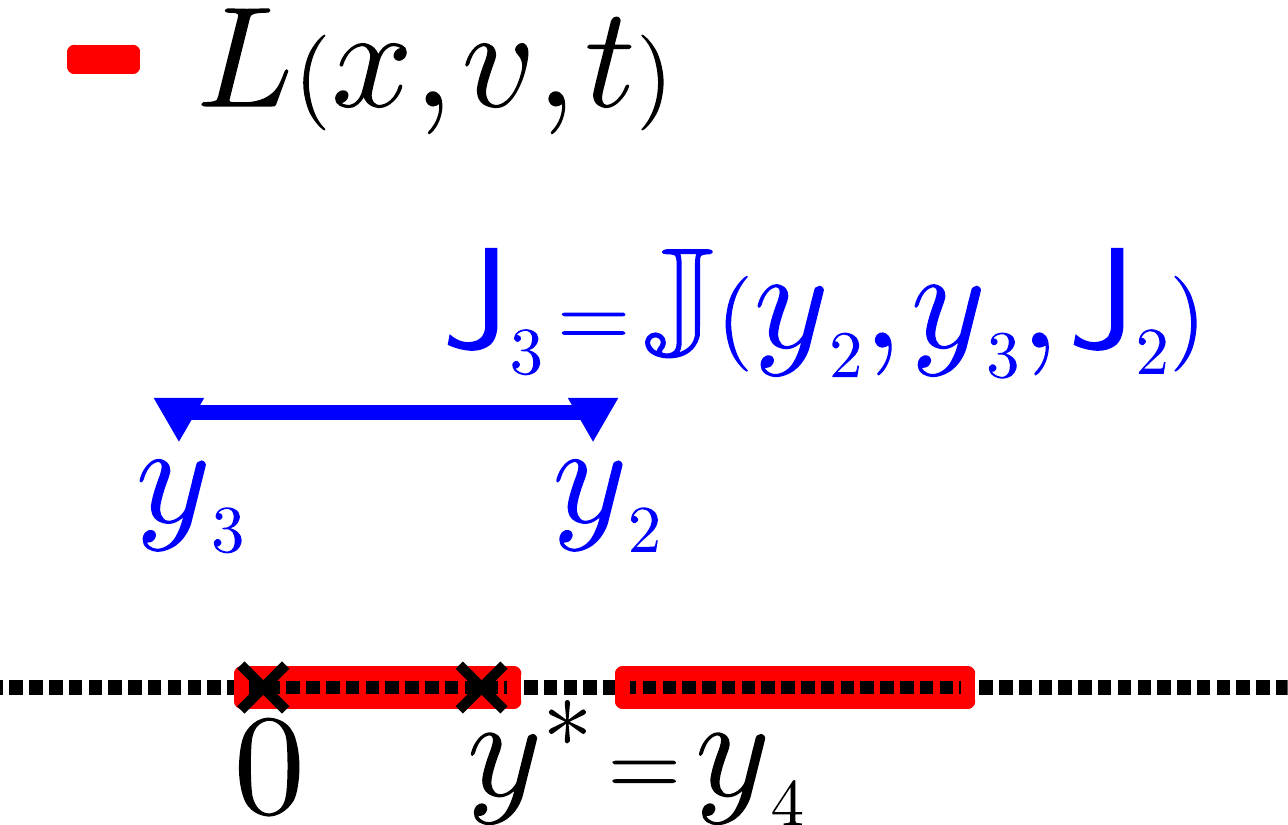}}; \\
		};	
		
		\node (output) [draw, align=center, rectangle, below = of shrink] {\textbf{Output:} next state $y=\gamma_{(x,v)}(\theta)$};
		
		\draw (input) -- (level);
		\draw (input) -- (step-out);
		\draw (level) -- (step-out);
		\draw (direction) -- (step-out);
		\draw (step-out) -- (shrink);
		\draw (shrink) -- (output);    
	\end{tikzpicture}
	\caption{Flow chart describing the transition mechanism of the Hit-and-run slice sampler.}
	\label{Fig: Hit-and-run slice sampling}
\end{figure}
%
For a complete algorithmic description see Algorithm \ref{A: GSS} with $\sigma_{d-1}^{(x)}$ being the uniform distribution on $\mathbb{S}^{d-1}$ for all $x \in \mathbb{R}^d$.
Equivalently, the transition mechanism can be described as first sampling a direction $v$ uniformly from $\mathbb{S}^{d-1}$, and then running a 1-dimensional stepping-out and shrinkage based slice sampler for the unnormalized density $\alpha \mapsto p\left(\gamma_{(x,v)}(\alpha)\right)$ with initial point 0.

Our interest in this framework arises from the fact that, by generalizing straight lines to geodesics, it can be leveraged to general Riemannian manifolds, which we discuss in the next section.

\subsection{Geodesic slice sampling}\label{Sec: Geodesic slice sampling}
We now turn to slice sampling on Riemannian manifolds.
For the sake of brevity, we keep the introduction of objects from differential and Riemannian geometry to a bare minimum here
and focus on providing some intuition for this setting.
 References and a more formal summary
 can be found in Supplementary material \Ref{Sec: Manifolds}.
\begin{remark}[Riemannian manifold]
 Essentially, manifolds are topological spaces that are locally homeomorphic to open subsets of $\mathbb{R}^d$.
 This preserves several local properties of Euclidean space, e.g. local compactness,
 and allows to generalise a number of notions from $\mathbb{R}^d$ to manifolds.
 However, many properties of $\mathbb{R}^{d}$ are also lost.
 For example, in general a manifold does not admit a vector space structure.
 For an illustration of some of the notions discussed below see Figure \Ref{fig:tangent_manifold}.
 
 If the local homeomorphisms of a manifold $\mathsf{M}$, which are called coordinate neighbourhoods, 
 are compatible in a smooth way, 
 we speak of a smooth manifold.
 For this kind of manifolds, one can define for all points $x \in \mathsf{M}$ a tangent space $T_x\mathsf{M}$.
 Conceptually, $T_x\mathsf{M}$ can be viewed as the collection of the gradients of curves  through $x$ on $\mathsf{M}$.
 Every tangent space $T_x\mathsf{M}$ is a vector space of the same dimension as the image domain of the coordinate neighbourhoods of $\mathsf{M}$.
 This number is also called the dimension of $\mathsf{M}$.
 
 The tangent spaces can be turned into inner product spaces in the following way: 
 A Riemannian metric $\mathfrak{g}$ is a map providing a smooth assignment of points $x\in \mathsf{M}$ to a positive definite inner product $\mathfrak{g}_x(\cdot,\cdot)$ defined on the corresponding tangent space $T_x\mathsf{M}$.
 A Riemannian manifold $(\mathsf{M},\mathfrak{g})$ is simply a smooth manifold $\mathsf{M}$ endowed with a Riemannian metric $\mathfrak{g}$.
 
 The Riemannian metric $\mathfrak{g}$ induces a measure $\nu_{\mathfrak{g}}$ on $\mathsf{M}$ called the \emph{Riemannian measure}.
 It can be viewed as an extension of the Lebesgue measure to Riemannian manifolds, see e.g. \citep[Section XII.1]{AnalysisIII}.
 The measure $\nu_{\mathfrak{g}}$ is defined on the Borel-$\sigma$-algebra induced by the topology of $\mathsf{M}$ which is denoted by $\mathcal{B}(\mathsf{M})$.
 
 To generalize the notion of a straight line in $\mathbb{R}^d$ to a general Riemannian manifold $\mathsf{M}$,
  one leverages the property that a straight line is characterized by a vanishing second derivative:
 For a curve $\gamma: \mathsf{I} \to \mathsf{M}$, where $\mathsf{I} \subseteq \mathbb{R}$ is an interval,
 its velocity vector field assigns to each point $s \in \mathsf{I}$ an element $(\d \gamma)/(\d t)\vert_s \in T_{\gamma(s)}\mathsf{M}$,
 which we may interpret as the gradient of $\gamma$ at $s$.
 Defining  a derivative of the velocity vector field of a curve on a general Riemannian manifold $\mathsf{M}$ is not straightforward because each $(\d \gamma)/(\d t)\vert_s$  belongs to a different space.
 This issue can be solved by defining a notion of derivative called \textit{covariate derivative} $D/ dt$ related to the local geometry of the manifold through the use of a Riemannian metric $\mathfrak{g}$.
 The curve $\gamma$ is now called \textit{a geodesic} if the covariant derivative of its velocity vector field is zero, 
 i.e., if the equation $(D/dt)\big((\d \gamma)/(\d t)\big)=0$ holds.
 A Riemannian manifold $(\mathsf{M},\mathfrak{g})$ is complete if for any $x \in \mathsf{M}$ and any $v \in T_x\mathsf{M}$, there exists a unique geodesic
 \begin{equation}\label{Eq: Geodesic through x in direction v}
 	\gamma_{(x,v)}: \mathbb{R} \to \mathsf{M}
 \end{equation}
 satisfying $\gamma_{(x,v)}(0) = x$ and $(\d \gamma_{(x,v)})/(\d t)\vert_0 = v$.
 We may interpret $\gamma_{(x,v)}$ as the geodesic emanating from $x$ in direction $v$.
 Moreover, we define by $\mathrm{Exp}_x:  T_x\mathsf{M} \to \mathsf{M}, v \mapsto \gamma_{x,v}(1)$ the exponential map at $x \in \mathsf{M}$. 
 Note that the relation $\gamma_{(x,v)}(\theta) = \mathrm{Exp}_x(\theta v)$ holds for any $\theta \in \mathbb{R}$ and $ v \in T_x\mathsf{M}$.
\end{remark}



The need to sample from measures $\pi$ defined on connected, complete Riemannian manifolds occurs in several applications in Bayesian statistics, see e.g. \citep{holbrook2016bayesian,holbrook2020nonparametric,lieAccepteddimension, mantoux2021understanding}.
\begin{figure}
	\begin{center}
		\includegraphics[width=0.85\textwidth]{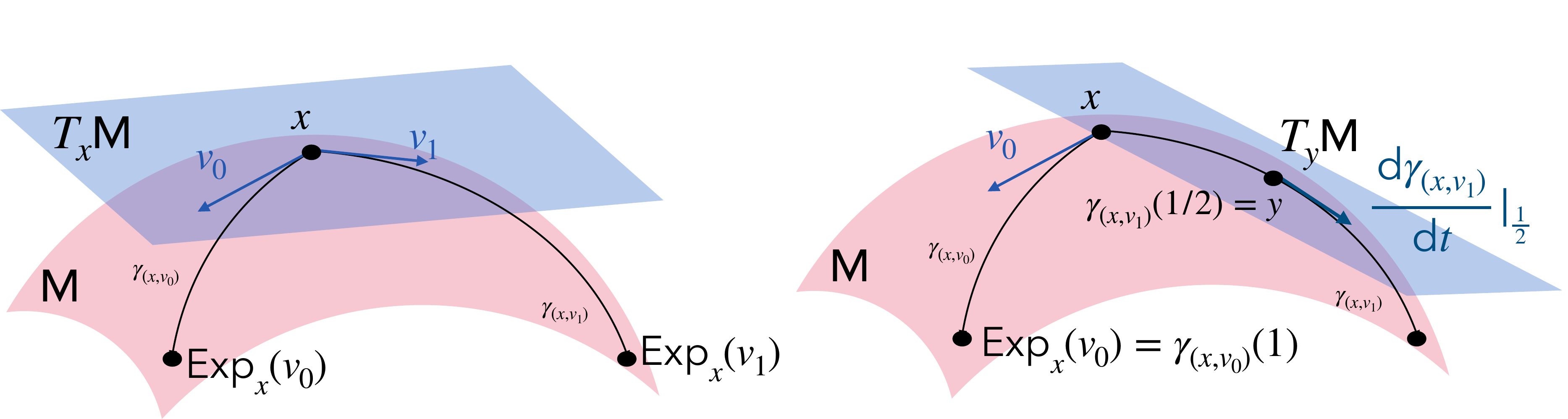}
	\end{center}
		\captionof{figure}{Illustration of tangent spaces, geodesics and the Exponential map. Note that $T_x\mathsf{M}\neq T_y\mathsf{M}$.
		}
	\label{fig:tangent_manifold}
\end{figure}
More precisely, in the following we consider sampling from probability measures $\pi$ defined on a state space $\mathsf{M}$ satisfying the following assumption:
\begin{assumption}\label{Ass: Manifolds}
	Let $\mathsf{M}$ be a $d$-dimensional, smooth, connected manifold.
	In addition we assume that $\mathsf{M}$ is endowed with a Riemannian metric $\mathfrak{g}$ that renders $\mathsf{M}$ complete.
\end{assumption}
We assume in this paper that the probability measure $\pi$ on $\mathsf{M}$ admits a density with respect to the Riemannian measure $\nu_{\mathfrak{g}}$, i.e.,
\begin{equation}\label{Eq: Definition target density}
	\pi(\d x) := \frac{p(x)}{\int_{\mathsf{M}} p(y) \ \nu_{\mathfrak{g}}(\d y)} \nu_{\mathfrak{g}}(\d x),
\end{equation}
where $p: \mathsf{M} \to [0, \infty)$ is the unnormalized target density.
We provide some examples for manifolds satisfying Assumption \ref{Ass: Manifolds}.
\begin{example}\label{Ex:Manifolds}
	\begin{enumerate}
		\item \label{Ex:Manifolds_i} The most simple example for a $d$-dimensional, smooth, connected manifold is $\mathbb{R}^d$.
		Since for each $x \in \mathbb{R}^d$ the tangent space to $\mathbb{R}^d$ at $x$ is again $\mathbb{R}^d$, we can equip it with the Riemannian metric 
		$\mathfrak{g}_x(v_1,v_2)=v_1^\top v_2$, for any $x \in \mathbb{R}^d, v_1,v_2 \in T_x\mathsf{M}$,
		rendering $\mathbb{R}^d$ complete.
		The induced Riemannian measure is the Lebesgue measure.		
		\item \label{Ex:Manifolds_ii} We equip the Euclidean unit sphere $\mathbb{S}^{d-1}$ with the standard Riemannian metric $\widehat{\mathfrak{g}}$ induced by its embedding $\mathrm{Id}: \mathbb{S}^{d-1} \to \mathbb{R}^d$ in $\mathbb{R}^d$, that is,
		$\widehat{\mathfrak{g}}_x = \left(\mathrm{Id}_*( v_1)\right)^\top \mathrm{Id}_*( v_2)$, for any $x \in \mathbb{S}^{d-1}, v_1, v_2 \in T_x\mathbb{S}^{d-1},$
		where $\mathrm{Id}_*$ is the map on the tangent spaces induced by $\mathrm{Id}$.
		Then $(\mathbb{S}^{d-1}, \widehat{\mathfrak{g}})$ satisfies Assumption \ref{Ass: Manifolds}.
		The corresponding Riemannian measure $\nu_{\widehat{\mathfrak{g}}}$ is the standard volume measure.
		For example applications involving directions give rise to the use of spheres in statistical models \citep{habeck2023geodesic}.
		\item \label{Ex:Manifolds_iii} Let $k,n \in \mathbb{N}$ with $k \leq n$.
		The $k(k-1)/2 + k(n-k)$-dimensional, smooth, connected \emph{Stiefel manifold} $\mathcal{V}(n,k) := \{\Gamma \in \mathbb{R}^{n \times k} \mid \Gamma^\top \Gamma = \mathrm{Id}_k\}$
		consists of (ordered)	$k$-tuples of vectors in $\mathbb{R}^n$ that form an orthonormal system.
		This means that each point on the Stiefel manifold describes (not uniquely) a $k$-dimensional subspace of $\mathbb{R}^n$.
		To characterize the tangent space $T_\Gamma\mathcal{V}(n,k)$ to a point $\Gamma \in \mathcal{V}(n,k)$, we need a matrix $\Gamma_\perp \in \mathbb{R}^{n \times (n-k)}$ such that the columns of $\Gamma$ and $\Gamma_\perp$ form an orthonormal basis of $\mathbb{R}^n$.
		Then $	T_\Gamma\mathcal{V}(n,k):= \{\Gamma \Pi + \Gamma_\perp \Sigma \in  \mathbb{R}^{n \times k} \mid \Pi \in \mathbb{R}^{k \times k} \text{ skew symmetric}, \Sigma \in \mathbb{R}^{(n-k)\times k}\}$.
		If we introduce the Riemannian metric
		\[
		\mathfrak{g}_\Gamma(\Delta_1, \Delta_2) = \mathrm{tr}\big(\Delta_1^\top (\mathrm{Id}_n - \frac{1}{2} \Gamma \Gamma^\top) \Delta_2\big) = \frac{1}{2} \mathrm{tr}(\Pi_1^\top \Pi_2) + \mathrm{tr}(\Sigma_1^\top \Sigma_2), 
		\]
		where $\Delta_1 = \Gamma \Pi_1 + \Gamma_\perp \Sigma_1, \Delta_2 = \Gamma \Pi_2 + \Gamma_\perp \Sigma_2 \in T_\Gamma\mathcal{V}(n,k)$,
		Assumption \ref{Ass: Manifolds} holds true for $\mathcal{V}(n,k)$.
		For more details see \citep{edelman1998geometry}.
		The Stiefel manifold commonly emerges in statistical analyses involving matrix structures.
		The next example provides a concrete instance for this.
		
	\end{enumerate}
\end{example}

\begin{example}[Analysis of functional connectivity networks of brains]\label{Ex: Brain network}
	In Section \ref{Sec: Applications},
	we consider a model introduced in \citep{mantoux2021understanding} that aims to infer the structure of adjacency matrices from functional connectivity networks of brains.
		Consider $J$ adjacency matrices $(\Upphi^{(j)})_{j \in \{1, \ldots, J\}}$ of different networks with $n \in \mathbb{N}$ nodes each.
		Let $k \in \mathbb{N}$ with $k \leq n$.
		The model that we consider has parameters $\theta=(\sigma_\kappa^2,\sigma^2_\epsilon,\mu,F)$, $\sigma_\kappa^2,\sigma^2_\epsilon>0,\mu \in \mathbb{R}^k$, and $F\in \mathbb{R}^{n\times k}$, and is defined 
		as
		\begin{equation}
		\begin{aligned}
			\Upphi^{(j)}&= \Gamma^{(j)}\operatorname{diag}(\kappa^{(j)}) (\Gamma^{(j)})^{\top}+\mathcal{E}(\epsilon^{(j)}), \qquad
			\epsilon^{(j)}\overset{\text{i.i.d.}}{\sim} \mathcal{N}(0,\sigma_\epsilon^2 \mathrm{Id}_{n(n+1)/2}),&\\
				\kappa^{(j)}&\overset{\text{i.i.d.}}{\sim} \mathcal{N}(\mu,\sigma_\kappa^2 \mathrm{Id}_k), \qquad\Gamma^{(j)}\overset{\text{i.i.d.}}{\sim}\operatorname{vMF}(F), &j \in \{1,\ldots,J\},
		\end{aligned}\label{eq:model_adjanc_matrices}
		\end{equation}
		where $\Gamma^{(j)}\in \mathcal{V}(n,k)$, $\kappa^{(j)} \in \mathbb{R}^k$ is called pattern weight vector, 
		and $\epsilon^{(j)} \in \mathbb{R}^{n(n+1)/2}$ represents the symmetric residual noise (i.e., $\mathcal{E}$ maps a vector in $\mathbb{R}^{n(n+1)/2}$ to the symmetric matrix in $\mathbb{R}^{n \times n}$ determined by its components).
		Moreover, $\operatorname{vMF}(F)$ denotes the matrix von Mises–Fisher distribution on the Stiefel manifold $\mathcal{V}(n,k)$
		which has unnormalized density
		\begin{equation}
			\label{eq:Mises_fisher_distribution}
			p_{\operatorname{vMF}(F)}(\Gamma) = \exp\left(\mathrm{tr}(F^\top \Gamma )\right), \qquad \Gamma \in \mathcal{V}(n,k),
		\end{equation}  
		with respect to the Riemannian measure on $\mathcal{V}(n,k)$.
		The unobserved variables $\Gamma^{(j)}\in \mathcal{V}(n,k)$ and $\kappa^{(j)}$ determine the individual-level specificity of network $j$. 
\end{example}

Very much in parallel to the (special) case $\mathbb{R}^d$, we are now able to develop a slice sampling approach
to target $\pi$ defined on $\mathsf{M}$. 
We can easily extend the uniform simple slice sampler 
from $\mathbb{R}^d$ to $\mathsf{M}$.
(In fact this works for every measure space, see Supplementary material \ref{Sec: Slice sampling in measurable spaces}.)
The level sets, containing all points where the unnormalized density $p$ is greater than a given value $t \in (0, \infty) $, take the form
$L(t) := \{x \in \mathsf{M} \mid p(x) > t\}. $
Define the set $\mathsf{W}:= \{x \in \mathsf{M} \mid p(x)  > 0\}$ where $p$ is strictly positive. 
The idealized slice sampler on $\mathsf{M}$ then proceeds analogously to $\mathbb{R}^d$, arriving at the transition kernel
\begin{align*}
	H: \mathsf{W}\times \mathcal{B}(\mathsf{W}) \to [0,1],\qquad
	(x,\mathsf{A})  \mapsto \frac{1}{p(x)} \int_{(0, p(x))} \frac{1}{\nu_{\mathfrak{g}}\big(L(t)\big)}\int_{L(t)} \mathbbm{1}_{\mathsf{A}}(y)\ \nu_{\mathfrak{g}}(\d y)\, \Leb_1(\d t).
\end{align*}
The implementation of the kernel requires sampling from the manifold-equivalent $\nu_{\mathfrak{g}}(L(t))^{-1}\nu_{\mathfrak{g}}\vert_{L(t)}$ of the uniform distribution on a level set $L(t)$.
Doing this efficiently poses a problem, as in $\mathbb{R}^d$, because in general we have no knowledge about the shape of
the level sets other than that they are $d$-dimensional, measurable sets.
Therefore, we propose a hybrid slice sampler that extends Hit-and-run slice sampling from $\mathbb{R}^d$ to the general Riemannian manifold $\mathsf{M}$.
As generalization of straight lines to $\mathsf{M}$, it uses the geodesics $\gamma_{(x,v)}$.
Like in $\mathbb{R}^d$, we want to index the geodesics emanating from a point $x \in \mathsf{M}$ by a ``unit sphere of directions'', which naturally is given by the \emph{unit tangent sphere} $\mathbb{S}_x^{d-1} := \{ v \in T_x\mathsf{M} \mid \mathfrak{g}_x(v,v) = 1\}$ in $T_x\mathsf{M}$.
The natural immersion of $\mathbb{S}^{d-1}_x$ into the inner product space $(T_x\mathsf{M}, \mathfrak{g}_x)$ via the identity induces a Riemannian metric $\widehat{\mathfrak{g}}_x$ on $\mathbb{S}_x^{d-1}$.
We call the normalization $\sigma_{d-1}^{(x)} :=  \nu_{\widehat{\mathfrak{g}}_x}(\mathbb{S}^{d-1}_x)^{-1}\nu_{\widehat{\mathfrak{g}}_x}$, for any $x \in \mathsf{M}$,
of the Riemannian measure $\nu_{\widehat{\mathfrak{g}}_x}$, induced by $\widehat{\mathfrak{g}}_x$, \emph{uniform distribution on $\mathbb{S}_x^{d-1}$} .
Then for $x \in \mathsf{M}$, $v \in \mathbb{S}_x^{d-1}$ and $t > 0$ we immediately obtain
$L(x,v,t) := \{\alpha \in \mathbb{R} \mid p\left(\gamma_{(x,v)}(\alpha) \right)> t\} = \{\alpha \in \mathbb{R} \mid \gamma_{(x,v)}(\alpha) \in L(t)\} $
as the parametrized intersection of a geodesic with a level set.

We now present the extension of the Hit-and-run slice sampler to manifolds replacing straight lines by geodesics
under observance of the subtleties  of the Riemannian setting.
We call this sampler the \emph{geodesic slice sampler}.
Roughly speaking, we arrive at a transition mechanism
which at each step randomly chooses a geodesic and then runs a stepping-out and shrinkage based 1-dimensional slice sampler on this geodesic:
Given a point $x \in \mathsf{M}$ with $p(x) > 0$, a new point $y \in \mathsf{M}$ is generated by first sampling a level $t$ uniformly from $(0, p(x))$, and a direction $v$ uniformly from $\mathbb{S}^{d-1}_x$, i.e., from $\sigma_{d-1}^{(x)}$.
The sampled level $t$ defines a level set $L(t)$, and the sampled direction $v$ specifies a geodesic $\gamma_{(x,v)}$ emanating from $x$ in direction $v$.
Now we use Neal's stepping-out and shrinkage techniques described in Section \ref{Sec: Slice sampling on Rd} to generate a point $\theta \in \mathbb{R}$ from the intersection $L(x,v,t)$ of the level set $L(t)$ and the geodesic $\gamma_{(x,v)}$.
The new point $y \in \mathsf{M}$ is then given by $y = \gamma_{(x,v)}(\theta)$.
\begin{algorithm}
	\caption{Geodesic slice sampler.}
	\label{A: GSS}
	\KwIn{point $x \in \mathsf{M}$ with $p(x) > 0$, hyperparameters $w \in (0, \infty)$ and $m \in \mathbb{N}$}
	\KwOut{point $y \in \mathsf{M}$ with $p(y)> 0$}
	\BlankLine
	Draw $T \sim \mathrm{Unif}\big((0, p(x))\big)$, call the result $t$.\;
	Draw $V \sim \sigma_{d-1}^{(x)}$, call the result $v$.\;
	Generate a realization of $(L,R) = \texttt{Step-out}_{w,m}(x,v,t)$, call the result $(\ell,r)$.\;
	Generate a realization of $\Theta= \texttt{Shrink}_{\ell, r}(x,v,t)$, call the result $\theta$.\;
	\Return{$y=\gamma_{(x,v)}(\theta)$.}
\end{algorithm}

    \begin{minipage}{0.55\textwidth}
        \begin{algorithm}[H]
			\label{A: Stepping-out}
			\caption{Stepping-out procedure.\\ Call as $\texttt{Step-out}_{w,m}(x,v,t)$.}
				\KwIn{point $x \in \mathsf{M}$, direction $v \in \mathbb{S}_x^{d-1}$, level $t \in (0, p(x))$,  hyperparameters $w \in (0, \infty)$ and  $m \in \mathbb{N}$}
				\KwOut{two points $\ell, r \in \mathbb{R} $ such that $\ell < 0 < r$}
				Draw $\Upsilon \sim \mathrm{Unif}\big((0,w)\big)$, call the result $u$.\;
				Set $\ell : =-u$ and $r := \ell + w$.\;
				Draw $J \sim \mathrm{Unif}(\{1, \ldots, m\})$, call the result $\upiota$.\;
				Set $i = 2$ and $j = 2$.\;
				\While{$i \leq \upiota$ and $p\left(\gamma_{(x,v)}(\ell)\right)>t$}{Set $\ell = \ell - w$.\;
				 Update $ i  = i+1$.}
				\While{$j \leq m + 1 -\upiota$ and $p\left(\gamma_{(x,v)}(r)\right)>t$}{Set $r = r + w$.\;
				Update $j = j+1$.}
				\Return{$(\ell, r)$}
		\end{algorithm}
    \end{minipage}
	\hfill
    \begin{minipage}{0.45\textwidth}
        \begin{algorithm}[H]
			\label{A: Shrinkage}
			\caption{Shrinkage procedure.\\ Call as $\texttt{Shrink}_{\ell, r}(x,v,t)$.}
				\KwIn{point $x \in \mathsf{M}$, direction $v \in \mathbb{S}_x^{d-1}$, level $t \in (0, p(x))$ and parameters $\ell < 0 < r$}
				\KwOut{point $\theta \in L(x,v,t)\cap [\ell, r)$}
				Draw $ \Theta \sim \mathrm{Unif}\big((0, r-l)\big) $, call the result $\theta_h$.\;
				Set $ \theta := \theta_h - \mathbbm{1}_{\{\theta_h > r \}} (r-l)$.\;
				Set $\theta_{\min} := \theta_h$ and $\theta_{\max} := \theta_h $.\;
				\While{$p\left(\gamma_{(x,v)}(\theta)\right)\leq t$}{\If{$\theta_h \in [ \theta_{\min}, r-l)$}{Set $\theta_{\min}=  \theta_h$.}
				\Else{Set $\theta_{\max}=  \theta_h$.}
				Draw $\Theta\sim \mathrm{Unif}\big((0, \theta_{\max})\cup [\theta_{\min}, r-l)\big)$, call result $\theta_h$.\;
				Set $\theta = \theta_h - \mathbbm{1}_{\{\theta_h > r \}} (r-l)$.}
				\Return{$\theta$.}
		\end{algorithm}
\end{minipage}
A complete algorithmic description of the geodesic slice sampler in pseudo code can be found in Algorithm \ref{A: GSS}.
It calls Algorithm \ref{A: Stepping-out} and Algorithm \ref{A: Shrinkage} representing the stepping-out and shrinkage procedure on the geodesic respectively.
We comment on the prerequisites of the geodesic slice sampler.
\begin{remark}
	\begin{enumerate}
		\item In order to implement Algorithm \ref{A: GSS} we need to be able to perform the following operations:
		\begin{itemize}
			\item Evaluation of the unnormalized density $p(x)$ at every $x \in \mathsf{M}$.
			\item Sampling from $\sigma_{d-1}^{(x)}$ for all $x \in \mathsf{M}$. If we know an isometric isomorphism $\mathbb{R}^d \to T_x\mathsf{M}$ for all $x \in \mathsf{M}$ this is an easy task.
			\item Evaluation of geodesics $\gamma_{(x,v)}(\theta)$ for all $x\in \mathsf{M}$, $v \in \mathbb{S}^{d-1}_x$ and $\theta \in \mathbb{R}$.
			Some cases where this is possible are provided in Example \ref{Ex:_Illustative_scenarios_for_GSS}.
		\end{itemize}
		\item The geodesic slice sampler takes two hyperparameters, namely $w \in (0,\infty)$ and $m \in \mathbb{N}$.
		They arise from the usage of Algorithm \ref{A: Stepping-out} (the stepping-out procedure), and their influence on the geodesic slice sampler is derived from their influence on the stepping-out procedure, see Remark \ref{R: Hyperparameters of stepping-out}.
		Roughly speaking, $mw$ determines the maximal possible size of the neighbourhood of the initial point $x \in \mathsf{M}$ that the geodesic slice sampler takes into account when performing its transition.
		This affects the reach of the algorithm as well as its ability to jump between modes of the target distribution $\pi$.
		Choosing $m$ larger and $w$ smaller can be seen as increasing the likelihood that a smaller neighbourhood of $x$ is considered for the transition, compared to the maximal possible reach of the algorithm.
		Depending on the shape of the unnormalized density $p$, this can hamper the ability of the geodesic slice sampler to jump between modes or allow the consideration of more ``relevant'' neighbourhoods.
		Observe that larger $m$ may lead to higher computational cost by increasing the maximal possible cost of Algorithm \ref{A: Stepping-out}.
	\end{enumerate}
\end{remark}
	Before we provide some illustrative scenarios, we point out two aspects related to the Riemannian setting:
	First, note that the direction used to parametrize the geodesic is an element of the tangent space at the current point $x$,
	i.e. its distribution depends on $x$.
	This is not visible in the Euclidean case from Section \ref{Sec: Slice sampling on Rd}, 
	because there the tangent space at each point and the underlying state space $\mathbb{R}^{d}$ coincide.
	
	Secondly, it is crucial to choose the hyperparameter $m$ finite,
	i.e. to bound the maximal number of stepping-out steps within one transition of the geodesic slice sampler.
	Otherwise the stepping-out procedure may not terminate in finite time with positive probability,
	because on a Riemannian manifold $\mathrm{Leb}_{1}(L(x,v,t)) < \infty$ is not guaranteed.
\begin{example}\label{Ex:_Illustative_scenarios_for_GSS}
	\begin{enumerate}
		\item \label{Ex:_Illustative_scenarios_for_GSS_i} The Hit-and-run slice sampler on $\mathbb{R}^d$ described in Section \ref{Sec: Slice sampling on Rd} fits into the framework of the geodesic slice sampler.
		\item \label{Ex:_Illustative_scenarios_for_GSS_ii} We consider $\mathbb{S}^{d} \subseteq\mathbb{R}^{d+1}$.
		For $x \in \mathbb{S}^{d}$ we have	$\mathbb{S}^{d-1}_x = \{v \in \mathbb{S}^{d} \mid x^\top v = 0\}$.
		The projection $\psi_x: \mathbb{S}^{d}  \to \mathbb{S}^{d-1}_x, v \mapsto \left(\mathrm{Id} - xx^\top\right) v$ onto the subspace orthogonal to $x\in \mathbb{S}^d$,	
		where $\mathrm{Id}$ denotes the identity on $\mathbb{S}^d$, yields the formula
		$\sigma_{d-1}^{(x)} = (\psi_x)_\sharp\left(\nu_{\widehat{\mathfrak{g}}}(\mathbb{S}^{d})^{-1}\nu_{\widehat{\mathfrak{g}}}\right).$
		The geodesics of $\mathbb{S}^{d}$ are the great circles given by the explicit formula
		$\gamma_{(x,v)}(\theta) = \cos(\theta) x + \sin(\theta)v, \theta \in \mathbb{R}$,
		for $x \in \mathbb{S}^{d}$ and $v\in \mathbb{S}^{d-1}_x$.
		Of course we may apply the geodesic slice sampler for arbitrary hyperparameters as described in Algorithm \ref{A: GSS}.
		However, since all geodesics are periodic of a known period length (namely $2 \pi$), this renders 
		the stepping-out procedure
		somehow superfluous.
		For wisely chosen hyperparameters ($w= 2\pi, m =1$), the geodesic segment sampled by Algorithm \ref{A: Stepping-out} always equals exactly one winding of the great circle, and we can simply replace line 3 in Algorithm \ref{A: GSS} by deterministically setting $(\ell, r) = (-\pi, \pi)$.
		The resulting algorithm is the geodesic shrinkage slice sampler on the sphere from \citet{habeck2023geodesic}.
		\item \label{Ex:_Illustative_scenarios_for_GSS_iii} For the Stiefel manifold defined in Example \ref{Ex:Manifolds}.3, an isometric isomorphism between $\mathbb{R}^{k(k-1)/2 + k(n-k)}$ and the tangent space $T_\Gamma\mathcal{V}(n,k)$ at a point $\Gamma \in \mathcal{V}(n,k)$ is given by using the first $k(k-1)/2$ components of $v \in \mathbb{R}^{k(k-1)/2 + k(n-k)}$ to determine a skew symmetric matrix $\Pi$ and the remaining ones to form a matrix $\Sigma \in \mathbb{R}^{(n-k) \times k}$. These two matrices determine an element of $T_\Gamma\mathcal{V}(n,k)$ (after fixing $\Gamma_\perp$) by $\Delta = \Gamma \Pi + \Gamma_\perp \Sigma$.
		We provide an explicit formula for the geodesic $\gamma_{(\Gamma, \Delta)}$.
		To this end let $QR = \Gamma_\perp \Sigma$ be the compact QR-decomposition of $\Gamma_\perp \Sigma = (\mathrm{Id}_n - \Gamma \Gamma^\top) \Delta$.
		For $\theta \in \mathbb{R}$ set $N_1(\theta) \in \mathbb{R}^{k\times k}$ and $N_2(\theta) \in \mathbb{R}^{k\times k}$ to be
		\[
		\begin{pmatrix}
			N_1(\theta) \\ N_2(\theta)
		\end{pmatrix}
		= \exp \left(\theta \begin{pmatrix}
			\Pi\ & -R^\top \\ R\ & \boldsymbol{0}_{k\times k}
		\end{pmatrix}\right)
		\begin{pmatrix}
			\mathrm{Id}_k \\ \boldsymbol{0}_{k\times k}
		\end{pmatrix},
		\]
		where $\exp$ denotes here the matrix exponential. 
		Then
		\[
		\gamma_{(\Gamma, \Delta)}(\theta) = \Gamma N_1(\theta) + Q N_2(\theta), \qquad \theta \in \mathbb{R}.
		\]
		The derivation of these results can be found in \citep{edelman1998geometry}.
	\end{enumerate}
\end{example}

Finally, we present the Markov transition kernel corresponding to the geodesic slice sampler.
To this end, we first give a rigorous specification of the unnormalized density $p$:
\begin{assumption}\label{Ass: Unnormalized density}  
	The unnormalized density $p: \mathsf{M} \to [0, \infty)$ is a lower semi-continuous\footnote{All level sets $L(t):= \{x \in \mathsf{M} \mid p(x) > t\}$, $t \in \mathbb{R}$, are open.} function such that $ \int_{\mathsf{M}} p(x)\ \nu_{\mathfrak{g}}(\d x) \in (0,\infty)$.
\end{assumption} 
We denote by $\|p\|_{\infty} := \sup_{x \in \mathsf{M}} |p(x)|$
the supremum norm of $p$.
Observe that Assumption \ref{Ass: Unnormalized density} gives $\|p\|_{\infty} \in (0,\infty]$.
\begin{remark}
	We impose lower semicontinuity of the unnormalized density $p$ in Assumption \ref{Ass: Unnormalized density} to ensure that Algorithm \ref{A: Shrinkage} (the shrinkage procedure) terminates almost surely. 
	This guarantees that its output has indeed a distribution on $\mathbb{R}$.
	For more details see \citep{ReversibilityEllipticalSliceSampler}.
\end{remark}
Next we fix $w \in (0, \infty)$ and $m \in \mathbb{N}$.
For simplicity we drop these two hyperparameters of the geodesic slice sampler in our subsequent notation.
Let $x \in \mathsf{M}$, $v \in \mathbb{S}_x^{d-1}$ and $t \in (0, p(x))$.
We denote by 
\begin{equation}\label{Eq: Definition stepping-out distribution within GSS}
	\xi_{L(x,v,t)}^{(0)}(\mathsf{A}) := \mathbb{P}(\texttt{Step-out}_{w,m}(x,v,t) \in \mathsf{A}), \qquad \mathsf{A} \in \mathcal{B}(\mathbb{R}^2),
\end{equation}
the distribution of the output of Algorithm \ref{A: Stepping-out} and by
\begin{equation}\label{Eq: Definition shrinkage kernel within GSS}
	Q_{L(x,v,t)}^{\ell, r}(0, \mathsf{A}) := \mathbb{P}(\texttt{Shrink}_{\ell, r}(x,v,t) \in \mathsf{A}), \qquad \mathsf{A} \in \mathcal{B}(\mathbb{R}),
\end{equation}
where $\ell < 0 < r$, the distribution of the output of Algorithm \ref{A: Shrinkage}.
Note that in \eqref{Eq: Definition stepping-out distribution within GSS} and \eqref{Eq: Definition shrinkage kernel within GSS} the right hand side only depends on $x \in \mathsf{M}$, $v \in \mathbb{S}_x^{d-1}$ and $t \in (0, p(x))$ through the set $L(x,v,t)$.
A formal definition of these distributions and some of their properties can be found in Supplementary material \ref{Sec: Stepping-out} and \ref{Sec: Shrinkage procedure}.
For $x \in \mathsf{M}$, $t \in (0,p(x))$ and $\mathsf{A} \in \mathcal{B}(\mathsf{M})$ we define the auxiliary Markov kernels
\[ 
K_t(x,\mathsf{A}):= \int_{\mathbb{S}_x^{d-1}}\int_{\mathbb{R}^2} \int_{[\ell, r)}  \mathbbm{1}_{\mathsf{A}}\left(\gamma_{(x,v)}(\theta)\right) 
\ Q_{L(x,v,t)}^{\ell, r}(0, \d \theta) \ \xi_{L(x,v,t)}^{(0)}\big(\d (\ell, r) \big)\ \sigma_{d-1}^{(x)}(\d v).
\]
Then the Markov kernel
\begin{equation}\label{Eq: Definition of geodesic slice sampling kernel}
		K: \mathsf{M} \times \mathcal{B}(\mathsf{M}) \to [0,1],\qquad
		(x,\mathsf{A})\mapsto \begin{dcases}
			\frac{1}{p(x)} \int_{(0, p(x))} K_t(x,\mathsf{A})\ \Leb_1(\d t),&  p(x) > 0,\\
			\delta_x(\mathsf{A}), & p(x)=0,
		\end{dcases}
\end{equation}
where the case $p(x)= 0$ is added to extend $K$ to a Markov kernel on $\mathsf{M}$,
corresponds to Algorithm \ref{A: GSS}.
As long as the underlying manifold satisfies Assumption \ref{Ass: Manifolds} and the target distribution $\pi$ satisfies Assumption \ref{Ass: Unnormalized density},
this is a well-defined object.
And although an implementation requires the operations listed in Remark  \ref{R: Hyperparameters of stepping-out}.1,
it may serve as a starting point for further algorithmic development or as an intermediate step in the analysis of other algorithms.
The following theorem verifies the correctness of $K$.
\begin{theorem}\label{Thm: Reversibility}
	Suppose Assumption \ref{Ass: Manifolds} and \ref{Ass: Unnormalized density} are satisfied,
	and let $\pi$ be defined as in \eqref{Eq: Definition target density}.
	Fix $w \in (0, \infty)$ and $m \in \mathbb{N}$. 
	Then the Markov kernel $K$ given in \eqref{Eq: Definition of geodesic slice sampling kernel} is reversible with respect to $\pi$.
	Moreover, if\, $\mathsf{W} = \{x \in \mathsf{M} \mid p(x) > 0\}$ is connected, then we have $\lim_{n \to \infty} d_{\mathrm{tv}}(K^n(x, \cdot), \pi)=0$
		for $\pi$-almost all $x \in \mathsf{M}$.
\end{theorem}

The proof of this statement can be found in Supplementary material \ref{Sec: Proof reversibility} and \ref{Sec: proof ergodicity}.

\begin{remark}\label{R: Harris recurence}
	If the injectivity radius $\mathrm{inj}(\mathsf{M})$ of the Riemannian manifold $\mathsf{M}$ is strictly positive,
	the convergence statement of Theorem \ref{Thm: Reversibility} can be extended to 
	hold for all starting points $x \in \mathsf{W}$ as long as the hyperparameters satisfies $mw \leq \mathrm{inj}(\mathsf{M})$.
	For example on Cartan-Hadamard manifolds this is always true.
	However, even when we go away from the setting $\mathrm{inj}(\mathsf{M}) = \infty$,
	the convergence for all starting points in $\mathsf{W}$ is easily preserved for all choices of $m$ and $w$ with a small alteration
	either of the hyperparameters of the first iteration of the geodesic slice sampler or a randomization of the hyperparameters at every iteration.
	For further details see Supplementary material \ref{Sec: proof ergodicity}.
\end{remark}

\subsection{Literature review of MCMC-methods on Riemannian manifolds}\label{Sec: Literature review}

In this section, we aim to provide an overview of existing MCMC-methods on Riemannian manifolds in the literature.
Roughly speaking they can be assigned to three different categories.

The first class of MCMC algorithms are defined on open sets of $\mathbb{R}^d$ but consider non canonical Riemannian metrics.
\citet{girolami2011riemann} generalize Hamiltonian Monte Carlo (HMC) to $\mathbb{R}^d$ equipped with an arbitrary metric tensor obtaining Riemannian manifold HMC (RMHMC), which assumes knowledge about the Riemannian metric of the underlying manifold.
They also introduce a MALA-type algorithm in this setting.

A second class of MCMC algorithms consists of methods defined on submanifolds of $\mathbb{R}^d$ associated with the canonical metric introduced by the embedding in $\mathbb{R}^d$ that are not necessarily open sets.
This includes Hamiltonian based MCMC-methods tailor-made for specific classes of manifolds that have been further developed upon RMHMC such as for implicitly defined manifolds \citep{brubaker2012family}, manifolds embedded in Euclidean space \citep{byrne2013geodesic} and the sphere \citep{lan2014spherical}.

Distributions on submanifolds of $\mathbb{R}^d$ can also be approximated by proposing a sample from the ambient $\mathbb{R}^d$ projected to the manifold and then running a Metropolis-Hastings acceptance rejection step, see e.g. \citep{mantoux2021understanding,zappa2018monte}. 
However, as usually the conditional distribution of the proposal is intractable and is therefore not taken into account in the acceptance rejection step, the resulting Metropolis-Hastings algorithm is biased.
However, for submanifolds of $\mathbb{R}^d$ defined by inequalities and equality constraints, \citet{zappa2018monte} propose a bias-free modification of this method.
In the special case when the underlying manifold is a hypersphere equipped with the angular Gaussian distribution as reference measure, 
other specialized bias-free reprojected MCMC algorithms have been proposed 
such as reprojected preconditioned Crank–Nicolson algorithm or reprojected Elliptical Slice sampling, see \citep{lieAccepteddimension}.

The third class of MCMC-methods on Riemannian manifolds employs geodesic flows.
\citet{mangoubi2018rapid} analyse a geodesic random walk on Riemannian manifolds with positive bounded curvature,
which is invariant with respect to the Riemannian measure. 
This analysis has been extended in \citep{goyal2019sampling} to arbitrary target probability measures on manifolds with bounded non-negative curvature by adding a Metropolis-Hastings-like acceptance step resulting in a geodesic Metropolis-Hastings algorithm.
Note that a similar approach has been taken before in \citep{lee2017geodesic} to target the uniform distribution on a polytope by equipping it with a Hessian structure.
In addition to this class of Metropolis-Hastings algorithms, there already exist MCMC-methods for specific manifolds that combine slice sampling with geodesics,
that is, for distributions on $\mathbb{R}^d$ there is Hit-and-Run slice sampling \citep{latuszynski2024convergence,Mackay}, and for distributions on hyperspheres there is geodesic slice sampling on the sphere which uses great circles instead of straight lines \citep{habeck2023geodesic}.
They can both be viewed as special cases of GSS.

\section{Application}\label{Sec: Applications}
We numerically assess the performance of GSS (Algorithm \ref{A: GSS}) in comparison with other Riemannian MCMC algorithms.
In our experiments, we consider the compact Stiefel manifold $\mathcal{V}(n,k)\subseteq \mathbb{R}^{n \times k}$ and the Grassmann manifold $\mathcal{G}(n,k)$ with, $k,n \in\mathbb{N}$, $ k<n$.
They find applications in shape analysis \citep{hong2017regression}, dimensionality reduction \citep{holbrook2016bayesian} and computer vision \citep{lui2012advances,nguyen2019neural}. 
We give a brief overview of the conducted experiments.
 All the code is available on GitHub\footnote{\texttt{https://github.com/samuelgruffaz/Geodesic_Slice_Sampling_on_Riemannian_Manifold.git}}
and for details on the samplers used in the numerical experiments for comparison, we refer to Supplementary material section \ref{Sec: MCMC samplers}.

\begin{itemize}
	\item We first consider the case where the target distribution belongs to the class of von Mises-Fisher distributions on $\mathcal{V}(n,k)$ and $\mathcal{G}(n,k)$.
		When examining the Stiefel manifold $\mathcal{V}(n,k)$, we compare GSS with a biased adaptive random walk Metropolis Hastings (RMH) sampler and an independent Metropolis Hastings sampler with a uniform proposal (UIMH).
		On the Grassmann manifold, we contrast the GSS with a gradient-informed sampler referred to as the geodesics Metropolis-adjusted Langevin algorithm (GeoMALA), inspired by the framework proposed by \citep{byrne2013geodesic}.

		These experiments on synthetic data highlight the advantage of using inexact geodesic computations (as in RMH) when the number of constraints is low ($p \ll n$), whereas exact computations are preferable otherwise.
For GSS, we observe that high values of $w$ and $m$ improve sampling performance, as measured by the Effective Sample Size (ESS).
 In highly concentrated or anisotropic target distributions, the step size must be small in some directions and larger in others.
  GSS's stepping-out and shrinkage mechanism adapts the step size to the geometry of the target distribution, a feature that eliminates the need for extensive hyperparameter tuning to achieve good performance with a fixed number of steps.
Notably, GSS can outperform GeoMALA and RMH on concentrated or anisotropic target distributions, even when the latter methods use tuned step sizes after a burn-in period. Although GSS performs worse in other scenarios, GSS is unbiased contrary to RMH (See \cite{zappa2018monte}) and GeoMALA (see Figure \ref{fig:bias_analysis_geomala} in the supplement).
 Furthermore, increasing $w$ and $m$ raises the number of target evaluations and thus the computational cost. Excessively large $w$ and $m$ offer diminishing returns, as U-turn displacements can occur on compact manifolds.
A practical heuristic is to set $mw$ around twice the diameter of the support of $\pi$, which corresponds to $mw \approx 2\pi$ for a distribution whose support is the whole Stiefel manifold. Developing a theoretical foundation for such heuristics remains an open research question, even in the Euclidean case \citep{latuszynski2024convergence, power2024weak}. 
		 For sake of conciseness, these experiements are postponed to Supplementary material 5.2.
              \item We compare GSS and RMH in inferring a latent variable model developed in \citep{mantoux2021understanding}
               and described in Example \ref{Ex: Brain network}.
			  RMH is used in \citep{mantoux2021understanding} because it compares favorably to its unbiased counterparts. 
			  This model
              is employed for the analysis of brain network structures. Specifically, it encodes principal directions of adjacency matrices obtained from magnetic resonance imaging (MRI) as points on a Stiefel manifold.
\item Finally, we introduce a Bayesian Matrix von Mises-Fisher clustering model for action recognition in videos. We approximate the posterior distribution using GSS and compare our resulting model with other existing approaches \citep{lin2017bayesian, sengupta2017bayesian}.
        \end{itemize}

We do not conduct comparisons with the Riemannian Hamiltonian Monte Carlo (RHMC) algorithm \citep{girolami2011riemann}, since it is not tailored for constrained manifolds (as the Stiefel or the Grassmann manifold).
The follow-up paper \citep{byrne2013geodesic} handles the case of constrained spaces using the expression of geodesic.
Therefore we have included GeoMALA \citep{byrne2013geodesic}  in our comparison on synthetic data in Section 5.2 in Supplementary material. 
Additional experiments using the ARMA model can be found in Supplementary material \ref{appendix:ARMA},
and for experiments on the Euclidean unit sphere see \cite{habeck2023geodesic}.

 For comments on the structures of the Stiefel and the Grassmann manifold needed to implement GSS see Examples \ref{Ex:Manifolds}.\ref{Ex:Manifolds_iii}, \ref{Ex:_Illustative_scenarios_for_GSS}.\ref{Ex:_Illustative_scenarios_for_GSS_iii} and Supplementary material \ref{Sec: MCMC samplers}.

\subsection{A practical case: Understanding the variability in graph data sets.}
\label{sub_sec:practical_case}
As already mentioned in Remark \ref{Ex:Manifolds_iii}, we consider in this section a model, see \eqref{eq:model_adjanc_matrices}, introduced in \citep{mantoux2021understanding} that aims to infer the  structure of adjacency matrices from functional connectivity networks of  brains.

The original paper estimates the parameters using the Markov chain Monte Carlo-stochastic approximation expectation maximization (MCMC-SAEM) procedure \citep{kuhn2004coupling}.
MCMC-SAEM is an extension of the expectation maximization (EM) algorithm, where the E-step involves approximating integrals using MCMC methods.
In this context, the E-step samples approximately the density $p_\Gamma=p(\Gamma^{(j)}|\Upphi^{(j)},\kappa^{(j)},\theta_n)$, where $\theta_n$ is the current estimate for the parameter $\theta$, within a Gibbs procedure using RMH (see \citep[Algorithm 3]{mantoux2021understanding}).

For the optimization procedure, we propose to replace RMH with GSS($w=1,m=5$) to compare  performances. Since our implementation of RMH is four times faster than GSS, we chose to multiply the number of MCMC iterations by four when using RMH for the E-step. 

Comparison of the log-complete likelihood during the estimation for different synthetic datasets can be found in Supplementary material \ref{appendix:estimation}.

\bigskip
\textbf{Missing links imputation.}
We follow the experimental setup proposed in \citep[Section
5.1.2]{mantoux2021understanding} for missing links imputation: A
synthetic data set of $N = 200$ adjacency matrices
$\{\Upphi^{(i)}\}_{i=1}^{N}$ with $n=20$ nodes and $k=5$ is generated from
the model specified by \eqref{eq:model_adjanc_matrices} with
parameters $\theta=(\sigma_\kappa^2,\sigma^2_\epsilon,\mu,F)$ given in Supplementary material \ref{appendix:numerical_details_missing_link_imputation}.  In a first stage, the
MCMC-SAEM algorithm with GSS is applied to perform the estimation of
the parameters resulting in an estimator $\hat{\mathbf{\theta}}$. Then,
from the same model, we generate another 200 samples $\{\bar{\Upphi}^{(i)}\}_{i=1}^{N}$. For each of
these samples, 16\% of the edge weights corresponding to the
interactions between the last eight nodes are masked.  Denote by $\{\tilde{\Upphi}^{(i)}\}_{i=1}^{N}$ the resulting adjacency matrices. We then aim to
reconstruct the missing links of $\{\bar{\Upphi}^{(i)}\}_{i=1}^{N}$ using the following three procedures:
\begin{itemize}
\item The missing links are found from the maximum a posteriori (MAP) approximated as $\textrm{MAP} = \Gamma \operatorname{diag}(\kappa) (\Gamma)^{\top}$, where $\Gamma,\kappa$ are the result of 4000 iterations of gradient ascent on the posterior density $p(\Gamma,\kappa  |\tilde{\Upphi}^{(i)},\hat{\mathbf{\theta}})$, the conditional density of $(\Gamma,\kappa)$ given the masked observation $\tilde{\Upphi}^{(i)}$ in model \eqref{eq:model_adjanc_matrices}.
\item The missing links are found from the posterior mean (PM) defined
  as \begin{equation}
	\label{eq:PM}
\textrm{PM}=\int_{\mathbb{R}^{d}}\int_{\mathcal{V}(n,k)} \Gamma \operatorname{diag}(\kappa) (\Gamma)^{\top} p(\Gamma,\kappa |\tilde{\Upphi}^{(i)},\hat{\mathbf{\theta}})\ \nu_{\mathfrak{g}}(\d\Gamma)\, \Leb_d(\d \kappa).
  \end{equation}
   This distribution is approximated using $\hat{\mathbf{\theta}}$ and GSS or RMH within a Gibbs sampler.
\item Finally, we consider a reconstruction using simply the link from the Mean Sample (MS) $N^{-1}\sum_{i=1}^N \bar{\Upphi}^{(i)}$.
\end{itemize}
Then the same experiment is repeated, but this time,
  40\% of the edges chosen uniformly from the edges of all nodes are masked,
  instead of constraining the mask to the edges of certain nodes.
Table \ref{tab:imputation1} and \ref{tab:imputation2}, reporting relative reconstruction errors, indicate that the Bayesian estimator (PM) offers better results than the MAP or the naive estimator (MS).
Moreover, computing the posterior mean using MCMC samples from GSS is better than using RMH. 

\begin{figure}[htbp]
    \centering
    \begin{minipage}{0.45\textwidth}
        \centering
        \begin{tabular}{|c|c|c|c|}
            \hline
        	MCMC & PM & MAP  & MS\\ \hline
            GSS  & \textbf{0.5} $\pm$ 0.2 & 0.52 $\pm$ 0.16 & 0.78 $\pm$ 0.1\\ \hline
            RMH & 0.55 $\pm$ 0.26& 0.51 $\pm$ 0.15& 0.82$\pm$ 0.1\\ \hline
        \end{tabular}
        \captionof{table}{Imputation with 16\% of the edges masked not uniformly. Relative reconstruction errors are given.}
		\label{tab:imputation1}
		\begin{tabular}{|c|c|c|c|}
            \hline
        	MCMC & PM & MAP  & MS\\ \hline
			GSS & \textbf{0.31} $\pm$ 0.1 & 0.35 $\pm$ 0.08 & 0.78 $\pm$ 0.07\\ \hline
            RMH& 0.34 $\pm$ 0.1& 0.35 $\pm$ 0.08& 0.78$\pm$ 0.07 \\\hline
        \end{tabular}
        \captionof{table}{Imputation with 40\% of the edges masked uniformly. Relative reconstruction errors are given.}
        \label{tab:imputation2}
    \end{minipage}
	\hfill
    \begin{minipage}{0.45\textwidth}
        \centering
		\includegraphics[width=\textwidth]{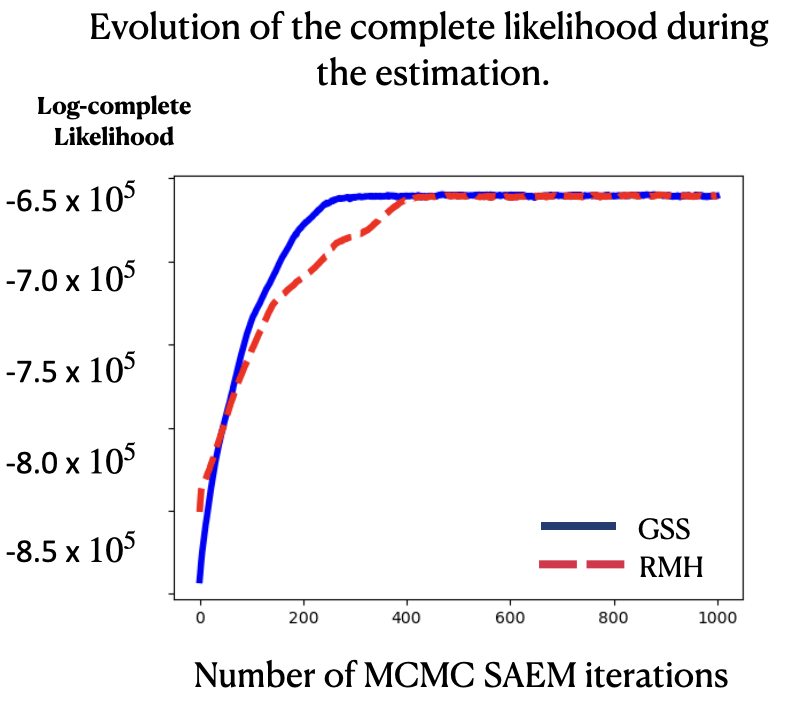}
	
\end{minipage}
\caption{(At right) Comparison of the MCMC-SAEM estimation with GSS and RMH.
We use four times more MCMC steps for RMH compared to GSS, since our implementation of RMH is four times faster.}
\label{fig:comparison_saem}
\end{figure}

\bigskip
\textbf{On a real dataset.}
We use real data to perform a comparison similar to the one proposed in \citep{mantoux2021understanding} where the authors  use connectivity matrices with $n=21$ nodes and $k=5$ in their model to analyse $N=1000$ subjects\footnote{We did not use the same data, since their dataset was unavailable to us.}.
Data were provided by the Human Connectome Project%
\footnote{\texttt{https://www.humanconnectome.org/study/hcp-young-adult/document/extensively-processed-fmri-data-documentation}}. The dataset is composed of brain connectivity matrices generated from resting-state functional MRI (rs-fMRI) by following the pipeline described in this documentation\footnote{\texttt{https://www.humanconnectome.org/storage/app/media/documentation/s1200/HCP1200-DenseConnectome+PTN+Appendix-July2017.pdf}}.
On a subject which receives no stimulation (at rest), the rs-fMRI records fluctuations in blood oxygenation levels throughout the brain.
 By maximizing the signal coherence in each region of the brain with a spatial independent component analysis (ICA)  \citep{beckmann2004probabilistic}, it yields a partition of the brain depending on its structure and variations from one individual to another. 
 Finally, the temporal correlations between the mean blood oxygenation levels in each region are assembled into a matrix.
 This matrix is called the brain's functional connectivity network. It should be noted that this network is not necessarily derived from physical reality, since it only represents correlations between brain regions. This is why  the term ``functional'' is coined.
 In this study, the connectivity matrices are defined on a parcellation of the brain into $n=25$ regions\footnote{The network modelling is related to ``netmats'' in the documentation.}. We chose the ``recon2'' group of the dataset leaving us 812 matrices of dimension $25\times 25$ to analyse.

 We choose $k=5$, and run 1000 MCMC-SAEM iterations with 20 MCMC steps per SAEM iteration when using GSS($w=1,m=10$), and  80 iterations when using RMH.
 The initialization procedure is the same as for the experiments with synthetic data described in Supplementary material \ref{appendix:estimation}.
 To evaluate the benefits of the model \eqref{eq:model_adjanc_matrices}, we calculate the rRMSE between PM and the observations, as in the previous paragraph, and compare it to the rRMSE obtained by projecting the data onto the subspace spanned by the first five principal components (PCA) of the full dataset, where each matrix $\Upphi^{(i)}$ is vectorized.
  The reconstruction using PM achieves a significantly lower average rRMSE (0.75) compared to PCA (0.95).
 Additionally, as shown in Figure \ref{fig:comparison_saem}, GSS demonstrates faster convergence than RMH when analyzing the evolution of the complete log-likelihood, consistent with observations on synthetic data. This supports the use of GSS for MCMC-SAEM estimation, where the target distribution evolves throughout the optimization process.
  The adaptive step size adjustments provided by the stepping-out and shrinkage mechanism are particularly advantageous in this context, as they eliminate the need for extensive manual tuning.  

\subsection{Bayesian clustering on the KTH video action dataset.}
In this section,
 GSS is used to estimate parameters of a Bayesian clustering model on the KTH video action data \citep{schuldt2004recognizing}
 following the pipeline used in \citep{chakraborty2019statistics}.
For four different scenarios (called ``d1'', ``d2'', ``d3'' and ``d4''), the dataset records 6 actions carried out by 25 humans, which yields 150 videos per scenario.
 From each video, a sequence of frames is extracted, and each frame is resized to $64\times 128$, before computing its histogram of oriented gradients (HOG) \citep{dalal2005histograms} features, which are $d=3780$-dimensional.
 Finally, an ARMA model is used to model the sequence of HOG features by estimating the parameter with the closed-form formula given in \citep{doretto2003dynamic}.
 A detailed description of this model and additional experiments can be found in Supplementary material \ref{appendix:ARMA}.
 For each video, let $T$ be the number of frames and $f_e\in \mathbb{R}^{d\times T}$ be the matrix formed by stacking the HOG feature vectors from each frame.
Let $f_e=\widehat{U}\mathrm{diag}(\lambda)\widehat{V}^\top$ be the SVD of $f_e$ by taking only the $d_l=50$ first components of the decomposition,  i.e. $\widehat{U} \in \mathbb{R}^{d \times d_l}$, $\lambda \in \mathbb{R}^{d_l}$ and $\widehat{V} \in \mathbb{R}^{T \times d_l}$.
 	Then the video is represented by $(\widehat{U},\lambda,\Sigma)\in \mathcal{V}(d,d_{l})\times \mathbb{R}^{d_{l}}\times \mathcal{V}(d_{l},d_{l})=:\mathcal{X}$ where 
	\begin{equation*}
		\Sigma=\mathrm{diag}(\lambda) \widehat{V}^\top M_1\widehat{V}(\widehat{V}^\top M_2 \widehat{V})^{-1}\mathrm{diag}(\lambda)^{-1},\quad M_1=\begin{pmatrix}
			0& 0\\
			 \mathrm{Id}_{T-1} & 0
			\end{pmatrix}, \,\, M_2=\begin{pmatrix}
				\mathrm{Id}_{T-1}& 0\\
				 0 & 0
				\end{pmatrix}\in \mathbb{R}^{T\times T}.
	\end{equation*}
 Clustering models using mixtures of matrix von Mises–Fisher distributions have already been proposed in \citep{lin2017bayesian,sengupta2017bayesian}, where the distribution on $\mathbb{R}^{d_l}$ is a mixture of multivariate Gaussian distributions with diagonal covariance matrix.
	The number of clusters is taken equal to the number of actions, that is, six.
	Each cluster is described by a parameter $\theta_k=(F_k^1,F_k^2,\mu_k,s_k) \in (\mathbb{R}^{d\times d_{l}} \times\mathbb{R}^{d_l\times d_{l}} \times\mathbb{R}^{d_{l}} \times (0, \infty)^{d_{l}})$ and a mixing weight $m_k\in[0,1]$.
	  The observations $y_i=(\widehat{U}_i,\lambda_i,\Sigma_i)\in\mathcal{X} $ and their cluster assignment $Z_i$, $i \in \{1, \ldots, 150\}$, are assumed to follow the generation process
	\begin{align*}
		&Z_i\overset{\text{i.i.d.}}{\sim} \mathrm{Mult}((m_k)_{k \in\{1, \ldots,6\}}),\quad y_i \overset{\text{i.i.d.}}{\sim} \sum_{k=1}^6 m_k p(\cdot |Z_i=k, \theta_k)\ ,\\
		& p(y_i|Z_i=k, \theta_k)= p_{\operatorname{vMF}(F_k^1)}(\widehat{U}_i) p_{\operatorname{vMF}(F_k^2)}(\Sigma_i)
		\exp\left(-\frac{1}{2}(\lambda  - \mu_k)^\top \mathrm{diag}(s_k)^{-1}(\lambda-\mu_k)\right),
	\end{align*}
	where $\mathrm{Mult}((m_k)_{k \in\{1, \ldots,6\}})$ is the multinomial distribution with parameter $(m_k)_{k \in\{1, \ldots,6\}}$.
	In this Bayesian setting, we impose a uniform non-informative prior on every parameter. 
	Our goal is to assess if the clustering separates the $6$ actions. 
 To this end, for each environment, the dataset is split into a training set (67\% of the data) and a test set (33\%). The clustering model is trained with the training set and evaluated on the test set.

 The parameter estimation is performed on the training set by adapting the EM algorithm described in \citep{sengupta2017bayesian} to the product space $\mathcal{X}$. We denote by $(\hat{\theta}_k)_{k\in\{1,\ldots,6\}}$ the estimated parameters.
 The cluster label for a point $y$ in the test set is given by the index $l_i=\argmax_{k\in\{1, \ldots,6\}} \{p(y|\hat{\theta}_k,Z_i=k)\} $  of the cluster with the highest model density at the observation. The related results are reported in Table \ref{table:clustering} as ``vMF clustering EM''.

 The result of the EM algorithm is then used as an initialization for the sampling of the posterior.
  However, the sampling is only done for some parameters for computational reasons and to demonstrate the benefit of GSS on the Stiefel manifold.
 The vMF distributions' parameters $(F_k^1,F_k^2)_{k\in \{1, \ldots,6\}} $ are parametrized using their SVD $F_k^1=\bar{U}_k\bar{D}_k\bar{V}_k^\top, F_k^2=\widetilde{U}_k\widetilde{D}_k\widetilde{V}_k^\top$ belonging to 
 $\mathcal{V}(d, d_l)\times\mathbb{R}^{d_l} \times \mathcal{V}(d_l, d_l) $ and $\mathcal{V}(d_l, d_l)\times\mathbb{R}^{d_l} \times \mathcal{V}(d_l, d_l)$ respectively. Then we approximately sample from the posterior associated with $(\bar{U}_k,\widetilde{U}_k)_{ k\in \{1, \ldots,6\}}$ using GSS, and the resulting posterior samples $(\bar{U}_k^j,\widetilde{U}_k^j)_{k\in\{1, \ldots,6\}}^{j\in\{1, \ldots,N\}}$ yield pseudo-posterior cluster parameter samples $(\tilde{\theta}_k^j)^{j\in \{1, \ldots,N\}}_{k\in\{1, \ldots,6\}}$ where $N=500$ is the number of samples.
 Then, the samples $(\tilde{\theta}_k^j)_{ k\in \{1, \ldots,6\}}^{j\in \{1, \ldots,N\} }$ are used to compute Bayesian assignment weights $w_{i,k}= \sum_{j} p(y_i|\tilde{\theta}_k^j,Z_i=k)/N $ for each observation $y_i\in \mathcal{X}$ of the test set and each cluster $k$.
 The label of each observation is finally assigned as the cluster $l_i=\argmax_{k\in\{1, \ldots,6\}} \{w_{i,k}\} $ giving the highest model density averaged over the sampled posterior parameters. This procedure is referred to as ``vMF clustering MCMC'' in Table \ref{table:clustering}.

 Exploiting the benefits coming with a Bayesian approach,
 for each observation $y_i$ of the test set, we also compute the empirical variance of the sample
\begin{equation}   L_i=(l_{i,j}=\argmax_{ k\in \{1, \ldots,6\}} p(y_i|\tilde{\theta}_k^j,Z_i=k))_{j \in \{1, \ldots, N\}}
\end{equation}
 and remove the point $i$ if it is positive.
 The procedure described above is then applied to this reduced test set.
 The scores resulting  from this procedure referred to as ``vMF clustering MCMC lower variance'', are given  in Table \ref{table:clustering}.

 \begin{table}
	\caption{Clustering results on the KTH action recognition dataset reported with f1 score weighted (\%).}
	\centering
	\begin{tabular}{lllll}
			  Scenario & d1 & d2& d3 &d4 \\
			  vMF clustering EM :  & 49.54  &56.61& 57.99 &49.98   \\
			  vMF clustering MCMC :  & \textbf{52.21}  &56.61& 57.99 &49.98   \\
			  vMF clustering MCMC lower variance  :  & 50.15  &\textbf{57.75}& \textbf{59.62} &\textbf{54.06}   \\
	
			\end{tabular}
		  \label{table:clustering}
	  \end{table}

      Conceptually, the Bayesian approach strengthens the estimation by averaging different plausible weights,
	   as in the ensemble methods \citep{dietterich2000ensemble}.
	   In Table \ref{table:clustering} we observe that this turns out to be always better than the simple maximum likelihood estimator (MLE).


\section*{Acknowledgement}
Data were kindly provided in part by the Human Connectome Project, WU-Minn Consortium
(Principal Investigators: David Van Essen and Kamil Ugurbil; 1U54MH091657) funded by the 16 NIH Institutes and Centers that support the NIH Blueprint for Neuroscience Research;
and by the McDonnell Center for Systems Neuroscience at Washington University.

We thank Rudrasis Chakraborty and Clément Mantoux for interesting discussions regarding the experiments on video actions and network data.
We thank Eric Moulines, Randal Douc and Philip Sch\"ar  for fruitful conversations regarding the theoretical part of this paper.
MH and DR gratefully acknowledge support of the DFG within project 432680300 – SFB 1456 subproject B02.
MH expresses her gratitude for the hospitality of \'Ecole Polytechnique.

\section*{Supplementary material}
The Supplementary Material includes a proof of Theorem \ref{Thm: Reversibility} alongside a formal description of the stepping-out and shrinkage procedure,
some more background on Riemannian geometry and slice sampling, and some further details on the numerical experiments, as well as some additional numerical experiments.



\bibliographystyle{biometrika}
\bibliography{GSSSources.bib}

\appendix

\section{Manifolds}\label{Sec: Manifolds}
We do not give a complete introduction to differential geometry in this section.
The aim is rather to give a better understanding on the key objects used by the geodesic slice sampler and provide references for what is outside the scope of this paper.

\subsection{Tangent space and Riemannian metric}

We revise some selected basic concepts on manifolds.
For a more thorough introduction to these objects see \citep[Sections I, III, IV, V]{Boothby}.

	A second countable Hausdorff space $\mathsf{M}$ with the property that
	for all $x \in \mathsf{M}$ there exists an open set $\mathsf{U} \subseteq \mathsf{M}$ containing $x$ and an open set $\mathsf{U}' \subseteq \mathbb{R}^d$ such that there is a homeomorphism $\varphi: \mathsf{U} \to \mathsf{U}'$
	is called a $d$-dimensional \emph{manifold}.
	Such a tuple $(\mathsf{U}, \varphi)$ of an open set and a homeomorphism is called a \emph{coordinate neighbourhood}.
	A manifold is additionally called \emph{smooth} if for all pairs of coordinate neighbourhoods $(\mathsf{U}, \varphi)$ and $(\mathsf{V}, \psi)$ with $\mathsf{U}\cap \mathsf{V} \neq  \emptyset$,
	the maps $\varphi \circ \psi^{-1}: \psi(\mathsf{U} \cap \mathsf{V}) \to \varphi(\mathsf{U} \cap \mathsf{V})$ and $\psi \circ \varphi^{-1}: \varphi(\mathsf{U} \cap \mathsf{V}) \to \psi(\mathsf{U} \cap \mathsf{V})$ are diffeomorphisms.
	Observe that $\psi(\mathsf{U} \cap \mathsf{V})$ and $\varphi(\mathsf{U} \cap \mathsf{V})$ are subsets of $\mathbb{R}^{d}$.
	Therefore the standard definition for diffeomorphisms applies.
	This strategy facilitated by the local Euclideanness of $\mathsf{M}$,
	allows also to define the notion of smooth functions on a smooth manifold $\mathsf{M}$.
	Namely, $f: \mathsf{U} \to \mathbb{R}$, where $\mathsf{U} \subseteq \mathsf{M}$ is open, is called a \emph{$C^\infty$-function}, 
	if for all coordinate neighbourhoods $(\mathsf{V}, \psi)$ with $\mathsf{V} \cap \mathsf{U}\neq \emptyset$ the function $f \circ \psi^{-1} : \psi(\mathsf{V} \cap \mathsf{U}) \to \mathbb{R}$ is infinitely often differentiable.
	This notion of smooth functions is used to define a \emph{tangent space}  $T_x\mathsf{M}$ to $\mathsf{M}$ for every point $x \in \mathsf{M}$.
	Given $x \in \mathsf{M}$, 
	it contains all maps $v$ from $C^\infty_x(\mathsf{M}) := \{ f: \mathsf{U} \to \mathbb{R} \mid x \in \mathsf{U},  f\ C^\infty\text{-function}\}$ to $\mathbb{R}$ that are linear and satisfy $v(fg) = f(x) v(g) + g(x) v(f)$ for all $f,g \in C^\infty_x(\mathsf{M})$.
	Here, two functions in $C^\infty_x(\mathsf{M})$ are identified with each other if they agree on an open neighbourhood of $x$.
	Intuitively, we may think about $T_x\mathsf{M}$ as the set of directional derivatives evaluated at $x$ applicable to $C^\infty$-functions on $\mathsf{M}$.
	Importantly, all tangent spaces of a manifold $\mathsf{M}$ are vector spaces of the same dimension as $\mathsf{M}$.
	Given a coordinate neighbourhood $(\mathsf{U}, \varphi)$, the coordinate frames $E_{1,x}^{\varphi}, \ldots, E_{d,x}^{\varphi} \in T_x\mathsf{M}$,
	defined as
	\[
	E_{i,x}^{\varphi}(f) = \partial_i\vert_{\varphi (x)} (f \circ \varphi^{-1}), \qquad f \in C^\infty(\mathsf{M}),
	\]
	where $\partial_i$ denotes the partial derivative with respect to the $i$th coordinate,
	form a basis of the vector space $T_x\mathsf{M}$ for all $x \in \mathsf{U}$.
	If the tangent spaces can be turned into inner product spaces in a way that is compatible with the smoothness structure on $\mathsf{M}$,
	we call the manifold Riemannian.
	Namely, a mapping $\mathfrak{g}$ that associates to each point $x \in \mathsf{M}$ a symmetric, positive definite bilinear form $\mathfrak{g}_x$ on $T_x\mathsf{M}$
	and satisfies that for all  coordinate neighbourhoods $(\mathsf{U}, \varphi)$ the entries of the Gram matrix $\mathsf{U} \to \mathbb{R}, x \mapsto \mathfrak{g}_x(E_{i,x}^{\varphi}, E_{j,x}^{\varphi})$, $i,j \in \{1, \ldots,d\}$, are $C^\infty$-functions,
	is called a \emph{Riemannian metric}.
	Each tangent space of a Riemannian manifold $(\mathsf{M}, \mathfrak{g})$ is then an inner product space where the inner product is given by $\mathfrak{g}_x$.
	The second condition concerning the entries of the Gram matrix is to be understood in the way that the inner product structures on the individual tangent spaces have to be compatible in a smooth way.

\subsection{Geodesics}

We call a map from an interval in $\mathbb{R}$ to the manifold $\mathsf{M}$ a \emph{curve}.
The Riemannian structure on $\mathsf{M}$ induces a special class of curves on $\mathsf{M}$, namely the \emph{geodesics}.
Intuitively, a geodesic can be thought of as a curve of constant velocity.
For a formal definition of geodesics consult \citep[Section VII.5]{Boothby}.
We say the manifold $\mathsf{M}$ is\emph{ geodesically complete} if all geodesics can be extended such that their domain is $\mathbb{R}$.
In this case, by virtue of \citep[Corollary 4.28]{Lee}, for all $x \in \mathsf{M}$ and all $v \in T_x\mathsf{M}$ there exists a unique geodesic
\begin{equation}
	\gamma_{(x,v)}: \mathbb{R} \to \mathsf{M}
\end{equation}
satisfying $\gamma_{(x,v)}(0) = x$ and $\frac{\d \gamma_{(x,v)}}{\d t}\vert_0 = v$ for the velocity vector field at zero.
Intuitively, $\gamma_{(x,v)}$ can be thought of as the geodesic  through $x$ in direction $v$.
It can also be written in terms of the \emph{exponential map} $\mathrm{Exp}$ (see \citep[Section VII.6]{Boothby}) as
\[  
\gamma_{(x,v)}(\theta) = \mathrm{Exp}_x(\theta v), \qquad x\in \mathsf{M}, v \in T_x\mathsf{M}, \theta \in \mathbb{R}.
\]
By the Hopf-Rinow Theorem (see e.g \citep[Theorem 6.19]{Lee}) geodesic completeness is equivalent to metric completeness for connected manifolds.

\subsection{The Riemannian measure}\label{Sec: Riemannian measure}
We now introduce a measure on $\mathsf{M}$ which can be viewed as an extension of the Lebesgue measure to Riemannian manifolds, see also \citep[Section II.5]{Sakai}.
The topology on $\mathsf{M}$ induces a Borel-$\sigma$-algebra, which we denote as $\mathcal{B}(\mathsf{M})$.
We need the following family of functions:
Given a coordinate neighbourhood $(\mathsf{U}, \varphi)$, we use the Gram matrix of the coordinate frames $E_{1,x}^{\varphi}, \ldots, E_{d,x}^\varphi$ evaluated at $x \in \mathsf{M}$ to introduce the function
\begin{align*}
	\sqrt{\det(\mathfrak{g}, \varphi)}: \mathsf{U} &\to [0, \infty)\\
	x &\mapsto \sqrt{\det\left[\left(\mathfrak{g}_x(E_{j, x}^{\varphi}, E_{k,x}^{\varphi})\right)_{\{1\leq j,k\leq d\}}\right]}.
\end{align*}
Intuitively, we may think about $\sqrt{\det(\mathfrak{g}, \varphi)}$ as measuring the deformation (of volume) when passing from $\mathbb{R}^{d}$ to $\mathsf{M}$ with $\varphi^{-1}$.
	These functions form the basis for the definition of the Riemannian measure restricted to a coordinate neighbourhood.
	However, we still require a way to extend beyond a single coordinate neighbourhood:
As $\mathsf{M}$ is by definition second countable, there exists a countable collection $\{(\mathsf{U}_i, \varphi_i)\}_{i \in \mathbb{N}}$ of coordinate neighbourhoods such that $\bigcup_{i \in \mathbb{N}}\mathsf{U}_i = \mathsf{M}$.\footnote{A space where every cover of open sets contains a countable subcover is called Lindelöf space. Second countable spaces are Lindelöf.}
Subordinate to $\{\mathsf{U}_i\}_{i \in \mathbb{N}}$, the manifold $\mathsf{M}$ admits a \emph{partition of unity} $\{\rho_i\}_{i \in \mathbb{N}}$, see \citep[Section I.2.1]{Sakai}, i.e.
\begin{itemize}
	\item $\rho_i: \mathsf{M} \to [0,\infty)$ is a $C^\infty$-function with support $\mathrm{supp}\,\rho_i \subseteq \mathsf{U}_i$ for all $i \in \mathbb{N}$,
	\item $\{\mathrm{supp}\,\rho_i\}_{i \in \mathbb{N}}$ is locally finite\footnote{For all $x \in \mathsf{M}$ there exists a neighbourhood $\mathsf{U}$ of $x$ such that $\mathsf{U} \cap \mathrm{supp}\,\rho_i\neq \emptyset$ holds only for finitely many $i \in \mathbb{N}$. },
	\item $\sum_{i=1}^{\infty}\rho_i(x) = 1$ for all $x \in \mathsf{M}$. 
\end{itemize}
With the help of such a partition of unity we can now define the \emph{Riemannian measure} $\nu_{\mathfrak{g}}$ induced by the Riemannian metric $\mathfrak{g}$  by
\begin{equation}\label{Eq: Definition of Riemannian measure - appendix}
	\nu_{\mathfrak{g}}(\mathsf{A}) := \sum_{i=1}^{\infty} \int_{\varphi_i(\mathsf{U}_i)} \left(\rho_i \cdot \mathbbm{1}_{\mathsf{A}} \cdot \sqrt{\det(\mathfrak{g}, \varphi_i)}\right) \circ \varphi_i^{-1} (z)\ \Leb_d(\d z), \qquad \mathsf{A} \in \mathcal{B}(\mathsf{M}),
\end{equation}
which is a measure on the measure space $(M, \mathcal{B}(\mathsf{M}))$.
We provide some brief arguments that the Riemannian measure is well-defined and indeed a measure.
Observe that $\mathbbm{1}_{\mathsf{A}} \circ \varphi_i^{-1}$ is Borel measurable, and $\rho_i \circ \varphi_i^{-1}$ and $\sqrt{\det(\mathfrak{g},\varphi_i)}\circ \varphi_i^{-1}$ are smooth for all $A \in \mathcal{B}(\mathsf{M})$ and all $i \in \mathbb{N}$.
Hence the appearing Lebesgue integrals are defined. 
For independence of the construction from the choice of open covering and partition of unity see \citep[page 62]{Sakai}. 
The $\sigma$-additivity of $\nu_{\mathfrak{g}}$ is inherited from the $\sigma$-additivity of the Lebesgue integral.
Applying standard extension arguments using the additivity of the Lebesgue integral and monotone convergence theorem, 
we can extend \eqref{Eq: Definition of Riemannian measure - appendix} to measurable functions $f: \mathsf{M} \to [0, \infty)$ yielding
\[  
\int_{\mathsf{M}} f(x)\ \nu_{\mathfrak{g}}(\d x) = \sum_{i=1}^{\infty} \int_{\varphi_i(\mathsf{U}_i)} \left(\rho_i \cdot f \cdot \sqrt{\det(\mathfrak{g}, \varphi_i)}\right) \circ \varphi_i^{-1} (z)\ \Leb_d(\d z).
\]

\subsection{The uniform distribution on the unit tangent spheres}

Throughout this section we fix $x \in M$.
We denote by 
\[
\mathbb{S}_x^{d-1} := \{ v \in T_x\mathsf{M} \mid \mathfrak{g}_x(v,v) = 1\}
\]
the \emph{unit tangent sphere} to $\mathsf{M}$ at $x$.
It is the unit sphere in the inner product space $(T_x\mathsf{M}, \mathfrak{g}_x)$ and
is immersed into $T_x\mathsf{M}$ via the identity map $\mathrm{Id}: \mathbb{S}_x^{d-1} \to T_x\mathsf{M}$.
For all $v \in \mathbb{S}_x^{d-1}$, this induces a map $\mathrm{Id}_*: T_v\mathbb{S}_x^{d-1} \to T_vT_x\mathsf{M}$ on the tangent spaces, see \citep[Theorem IV.1.2]{Boothby}.
As $(T_x\mathsf{M}, \mathfrak{g}_x)$ is a $d$-dimensional inner product space, the tangent space $T_vT_x\mathsf{M}$ to $T_x\mathsf{M}$ at $v \in T_vM$ is again $T_x\mathsf{M}$, see \citep[Section II.3]{Boothby} for more details on this construction.
Because $(\mathfrak{g}_x)_v$ then becomes $\mathfrak{g}_x$, we can exploit this to define the Riemannian metric $\widehat{\mathfrak{g}}_x$ on $\mathbb{S}_x^{d-1}$ given by
\[
\left(\widehat{\mathfrak{g}}_x\right)_v(\xi_1, \xi_2) := \mathfrak{g}_x\big(\mathrm{Id}_*(\xi_1), \mathrm{Id}_*(\xi_2)\big), \qquad v \in\mathbb{S}_x^{d-1}, \xi_1, \xi_2 \in T_v\mathbb{S}_x^{d-1},
\]
see \citep[Corollary V.2.5]{Boothby}.
As described in Supplementary material section \ref{Sec: Riemannian measure}, $\widehat{\mathfrak{g}}_x$ induces the (Riemannian) measure $\nu_{\widehat{\mathfrak{g}}_x}$ on $\mathbb{S}_x^{d-1}$.
Since $\mathbb{S}_x^{d-1}$ is compact, $\nu_{\widehat{\mathfrak{g}}_x}\left(\mathbb{S}_x^{d-1}\right)$ is finite, and we may define the probability measure
\[
\sigma_{d-1}^{(x)} := \frac{1}{\nu_{\widehat{\mathfrak{g}}_x}(\mathbb{S}^{d-1}_x)} \nu_{\widehat{\mathfrak{g}}_x}.
\]
Note that, up to a push forward under an isometric isomorphism, $\sigma_{d-1}^{(x)}$ is the uniform distribution on the Euclidean unit sphere $\mathbb{S}^{d-1}$.

\section{Validity}\label{Sec: Validity}

The aim of this section is to prove the reversibility and ergodicity of the geodesic slice sampler.
To this end, we introduce a stepping-out distribution to describe Algorithm \ref{A: Stepping-out} and a shrinkage kernel to describe Algorithm \ref{A: Shrinkage}.
Their properties collected in Lemma \ref{L: Key identity for reversiblity - measurable functions}, Lemma \ref{L: stepping out under rotation}, Lemma \ref{L: Reversibility of the pushforward kernel} and Lemma \ref{L: Shrinkage under rotation} intuitively ensure ``reversibility on every parametrised intersection of geodesic and levelset''.
Together with the invariance of the Liouville measure under a certain map on the tangent bundle, both introduced in Supplementary material \ref{Sec: Proof reversibility}, they are the essential ingredients for the proof of the reversibility of the geodesic slice sampler.
For verifying ergodicity we establish irreducibility and aperiodicity of the kernel restricted to the set where the target density $p$ is strictly positive. 
A key tool here is a suitable transformation from polar to Cartesian coordinates.
\subsection{Stepping-out procedure}\label{Sec: Stepping-out}

Throughout this section fix $w \in (0, \infty)$ and $m \in \mathbb{N}$.
We consider a generalization of the stepping-out procedure described in Section \ref{Sec: Slice sampling on Rd} that targets an arbitrary set $\mathsf{S} \in \mathcal{B}(\mathbb{R})$, see  Algorithm \ref{A: Stepping-out general set}.

\begin{algorithm}
	\caption{Stepping-out procedure targeting $\mathsf{S} \in \mathcal{B}(\mathbb{R})$.}
	\label{A: Stepping-out general set}
	\KwIn{point $\theta \in \mathbb{R}$, hyperparameters $w \in (0, \infty)$ and $m \in \mathbb{N}$} 
	\KwOut{two points $\ell, r \in \mathbb{R} $ such that $\ell < \theta < r$}
	Draw $\Upsilon \sim \mathrm{Unif}((0,w))$, call the result $u$.\;
	Set $\ell : =\theta -u$ and $r := \ell + w$.\;
	Draw $J \sim \mathrm{Unif}(\{1, \ldots, m\})$, call the result $\upiota$.\;
	Set $i = 2$ and $j = 2$.\;
	\While {$i \leq \upiota$ and $\ell \in \mathsf{S}$}{Set $\ell = \ell - w$.\;
	Update $ i  = i+1$.}
	\While{$j \leq m + 1 -\upiota$ and $r \in \mathsf{S}$}{	Set $r = r + w$.\;
	Update $j = j+1$.}
	\Return $(\ell, r)$
\end{algorithm}

To formally describe the resulting distribution on $\mathbb{R}^2$ we use stopped random variables.
Let $ \Upsilon \sim \mathrm{Unif}([0,w])$.	
For every $ \theta \in \mathbb{R} $ let
\begin{equation}\label{Eq: Definition L_i, R_j}
	\begin{split}
		L_i^{(\theta)}&:= \theta -\Upsilon - (i-1)w, \qquad i \in \mathbb{N},\\
		R_j^{(\theta)}& :=\theta  + (w- \Upsilon ) + (j-1) w =  \theta -\Upsilon + jw, \qquad j \in \mathbb{N},
	\end{split}
\end{equation}
be two sequences of random variables.
Observe that setting  
\begin{equation}\label{Eq: Convention for L_0,R_0}
	\begin{split}
		L_0^{(\theta)} &:= R_1^{(\theta)} =  \theta - \Upsilon - (-1)w\\
		R_0^{(\theta)} &:= L_1^{(\theta)} =  \theta - \Upsilon + 0 \cdot w
	\end{split}
\end{equation}
is consistent with this definition.
By construction both sequences are strictly monotone, i.e.
\begin{equation}\label{Eq: Monotonicity L, R}
	\begin{split}
		L_{i+1}^{(\theta)} = \theta - \Upsilon - iw &< \theta -\Upsilon - (i-1)w = L_i^{(\theta)}, \qquad i \in \mathbb{N}_0, \theta \in \mathbb{R},\\
		R_j^{(\theta)} = \theta - \Upsilon + jw &< \theta - \Upsilon + (j+1) w = R_{j+1}^{(\theta)}, \qquad j \in \mathbb{N}_0, \theta \in \mathbb{R}.
	\end{split}
\end{equation}
We now define appropriate stopping times depending on the target set $\mathsf{S} \in \mathcal{B}(\mathbb{R})$.
To this end, let $J$ be a random variable that is independent of all previous objects  satisfying
$J \sim \mathrm{Unif}(\{1,\ldots, m\})$,
and set
\begin{equation}\label{Eq: Definition stopping times}
	\begin{split}
		\tau_{\mathsf{S}}^{(\theta)} &:= \inf\{ i \in \mathbb{N} \mid L_i^{(\theta)} \notin \mathsf{S} \} \land J , \\
		\mathfrak{T}_{\mathsf{S}}^{(\theta)} &:= \inf\{ j \in \mathbb{N} \mid R_j^{(\theta)} \notin \mathsf{S} \} \land (m + 1 -J) , \qquad \theta \in \mathbb{R}, \mathsf{S} \in \mathcal{B}(\mathbb{R}).
	\end{split}
\end{equation}
Note that $\tau_{\mathsf{S}}^{(\theta)}$ and $\mathfrak{T}_{\mathsf{S}}^{(\theta)}$ are finite stopping times with respect 
to the filtration generated by the sequence $( J, L_1^{(\theta)}, \ldots, L_n^{(\theta)}, R_1^{(\theta)}, \ldots, R_n^{(\theta)} )_{n \in \mathbb{N}} $.
More precisely, we have the bounds
\begin{equation}\label{Eq: Bound tau, T}
	\begin{split}
		1 \leq \tau_{\mathsf{S}}^{(\theta)} \leq J \leq m,\\
		1 \leq \mathfrak{T}_{\mathsf{S}}^{(\theta)} \leq m + 1 - J \leq m + 1 - 1 = m,\\
		2 \leq \tau_{\mathsf{S}}^{(\theta)} + \mathfrak{T}_{\mathsf{S}}^{(\theta)} \leq J + m + 1 - J \leq m + 1.
	\end{split}
\end{equation}
As $\tau_{\mathsf{S}}^{(\theta)}$ and $\mathfrak{T}_{\mathsf{S}}^{(\theta)}$ are finite stopping times, 
\begin{equation}\label{Eq: Defintion stepping-out random variables}
	\left( \boldsymbol{L}_{\mathsf{S}}^{(\theta)},\boldsymbol{R}_{\mathsf{S}}^{(\theta)}  \right) := \left( L_{\tau_{\mathsf{S}}^{(\theta)}}^{(\theta)},R_{\mathfrak{T}_{\mathsf{S}}^{(\theta)}}^{(\theta)}  \right)
\end{equation}
are random variables for all $\theta \in \mathbb{R}$ and $\mathsf{S} \in \mathcal{B}(\mathbb{R})$, see \citep[Lemma 9.23]{Klenke}.
We define the \emph{stepping-out distributions}
\[  
\xi_{\mathsf{S}}^{(\theta)} := \mathbb{P}^{\big( \boldsymbol{L}_{\mathsf{S}}^{(\theta)},\boldsymbol{R}_{\mathsf{S}}^{(\theta)}  \big)}, \qquad \theta\in \mathbb{R}, \mathsf{S} \in \mathcal{B}(\mathbb{R}),
\]
on $(\mathbb{R}^2,\mathcal{B}(\mathbb{R}^2))$.
Observe that the output of Algorithm \ref{A: Stepping-out general set} has distribution $\xi_{\mathsf{S}}^{(\theta)}$.
Since Algorithm \ref{A: Stepping-out} is a special case of Algorithm \ref{A: Stepping-out general set} with $\mathsf{S} = L(x,v,t)$ and $\theta = 0$, this definition is coherent with \eqref{Eq: Definition stepping-out distribution within GSS}.
Moreover, note that $(\theta, C) \mapsto \xi_{\mathsf{S}}^{(\theta)}(C)$ is a transition kernel. More details on this can be found in Remark \ref{R: Stepping-out distribution is a kernel}.

We collect two properties of the stepping-out distribution that are useful to show reversibility of the geodesic slice sampler.
Their proof is postponed to Supplementary material \ref{Sec: Properties of stepping-out}.

\begin{lemma}\label{L: Key identity for reversiblity - measurable functions}
	Let $\mathsf{S} \in \mathcal{B}(\mathbb{R}),$ and let $f: \mathbb{R}^2 \to [0,\infty) $ be a measurable function.
	We have for all $\theta,\alpha \in \mathbb{R}$
	\begin{align*}
		\int_{\mathbb{R}^2} f(\ell, r) \mathbbm{1}_{(\ell, r)}(\alpha) \ \xi_{\mathsf{S}}^{(\theta)}\big(\d (\ell, r) \big)
		= \int_{\mathbb{R}^2} f(\ell, r) \mathbbm{1}_{(\ell, r)}(\theta) \ \xi_{\mathsf{S}}^{(\alpha)}\big(\d (\ell, r) \big).
	\end{align*}
\end{lemma}

The previous lemma essentially states that, conditioned on the event that the initial point lies inside the resulting interval, the stepping-out distribution does not depend on the initial point. 
To formulate the second property we need the collection
\begin{equation}\label{Eq: Definition Lambda_alpha}
	\Lambda_\alpha : \mathbb{R}  \to \mathbb{R}, \qquad
	\theta  \mapsto \alpha - \theta
\end{equation}
of linear functions indexed by $\alpha \in \mathbb{R}$.
Intuitively, they express a U-turn at $\alpha \in \mathbb{R}$.
Moreover, for $\alpha \in \mathbb{R}$ we define
\begin{equation}\label{Eq: Definition uplambda_alpha}
	\uplambda_\alpha : \mathbb{R}^2  \to \mathbb{R}^2, \qquad
	(\ell, r)  \mapsto (\alpha - r, \alpha -\ell). 
\end{equation}
The next lemma describes the behaviour of the stepping-out distribution under U-turns.

\begin{lemma}\label{L: stepping out under rotation}
	Let $\mathsf{S} \in \mathcal{B}(\mathbb{R})$ and $\theta, \alpha \in \mathbb{R}$.
	We have
	\[ 
	\xi_{\Lambda_\alpha(\mathsf{S})}^{(\theta)} = (\uplambda_\alpha)_\sharp\xi_{\mathsf{S}}^{(\Lambda_\alpha(\theta))}.
	\]
\end{lemma}

\subsection{Shrinkage procedure}\label{Sec: Shrinkage procedure}

In this section, for every half open interval $[\ell,r)\subseteq \mathbb{R}$, we introduce an algorithm that generalises Algorithm \ref{A: Shrinkage}  approximating
the uniform distribution on $\mathsf{S} \cap (\ell,r)$ for an arbitrary open set $\mathsf{S} \in \mathcal{B}\big(\mathbb{R}\big)$.
To express this scheme as a kernel, we employ the kernel of the shrinkage procedure, essentially operating on $\mathbb{S}^1$ parametrised by $[0,2\pi)$, introduced in \citep{ReversibilityEllipticalSliceSampler} and push it forward to an arbitrary interval $[\ell,r)$ as described in \citep[Appendix A]{rudolf2022robust}.
For the convenience of the reader we quickly sketch the kernel $Q_{\textgoth{S}}$ from \citep{ReversibilityEllipticalSliceSampler}.
To this end set
\[
	\mathbb{J}( \alpha, \theta) := \begin{dcases}
									[\alpha, \theta), & \alpha < \theta,\\
									[0, \theta) \cup [\alpha, 2\pi), & \alpha \geq \theta,
							\end{dcases}
	\qquad
	\widehat{\mathbb{J}}(\alpha, \theta):= \begin{dcases}
									[\alpha, \theta), & \alpha \leq \theta,\\
									[0, \theta) \cup [\alpha, 2\pi), & \alpha > \theta.
							\end{dcases}
\]
We denote by $\delta_z$ the Dirac measure at $z \in \mathbb{R}$.
Let $\Uptheta$ be a random variable on $[0, 2\pi)$, and let $(\Upgamma_n, \Upgamma_n^{\min}, \Upgamma_n^{\max})_{n \in \mathbb{N}}$ be a sequence of random variables with conditional distributions
\begin{align*}
	&\mathbb{P}\big((\Upgamma_{n+1}, \Upgamma_{n+1}^{\min}, \Upgamma_{n+1}^{\max}) \in \mathsf{C}\mid  (\Upgamma_n, \Upgamma_n^{\min}, \Upgamma_n^{\max}) = (z, z^{\min}, z^{\max}), \Uptheta = \theta\big)\\
	&\qquad= \mathbbm{1}_{\widehat{\mathbb{J}}(z^{\min}, z)}(\theta) \cdot\left(\mathrm{Unif}\big(\mathbb{J}(z^{\min}, z)\big) \otimes \delta_{z^{\min}} \otimes \delta_{z}\right)(\mathsf{C}) \\  
	&\qquad \qquad+   \mathbbm{1}_{\mathbb{J}(z, z^{\max})}(\theta) \cdot\left(\mathrm{Unif}\big(\mathbb{J}(z, z^{\max})\big) \otimes \delta_z \otimes \delta_{z^{\max}}\right)(\mathsf{C})
\end{align*}
for $\theta, z, z^{\min}, z^{\max} \in [0,2\pi)$ with $\theta, z \in \mathbb{J}(z^{\min}, z^{\max})$, $\mathsf{C} \in \mathcal{B}([0, 2\pi)^3), \theta \neq z$ and $n \in \mathbb{N}$.
Moreover, for every set $\textgoth{S} \in \mathcal{B}([0,2\pi))$ which is \emph{open in $\mathbb{S}^{1}$}, i.e. satisfies that for all $\theta \in \textgoth{S}$ there exists $\varepsilon > 0$ such that $\mathbb{J}(\theta - \varepsilon \mod 2 \pi, \theta + \varepsilon \mod 2\pi) \subseteq \textgoth{S}$, define the stopping time
\[
	\mathcal{T}_{\textgoth{S}} := \inf\{n \in \mathbb{N} \mid \Upgamma_n \in \textgoth{S}\}.
\]
Then the kernel of the shrinkage procedure targeting $\mathrm{Unif}(\textgoth{S})$ is given by
\begin{equation*}
	Q_{\textgoth{S}}: \textgoth{S} \times \mathcal{B}(\textgoth{S} )\to [0,1], \quad (\theta, \mathsf{A}) \mapsto \mathbb{P}\left(\Upgamma_{\mathcal{T}_{\textgoth{S}}} \in A, \mathcal{T}_{\textgoth{S}} < \infty \mid \Uptheta = \theta\right).
\end{equation*}

To ``bend the interval $[\ell,r)$ onto $\mathbb{S}^1$'', where $\ell, r \in \mathbb{R}$ such that $\ell < r$, we use the family of maps
\begin{align*}
	h_{\ell,r}: [\ell,r) \to [0,2\pi), \qquad
	\theta \mapsto \frac{2\pi}{r-\ell}\theta \ \mod 2\pi.
\end{align*}
Note that due to the restriction of the domain, these functions are bijective and therefore have inverses $ h_{\ell,r}^{-1} $. 
Let $\ell, r \in \mathbb{R}$ such that $\ell < r$, and let $\mathsf{S} \in \mathcal{B}(\mathbb{R})$ be an open set such that $(\ell,r) \cap \mathsf{S} \neq \emptyset$.
For such numbers $\ell, r$ and sets $\mathsf{S}$ we define the shrinkage kernel as the push forward of the kernel $Q_{\textgoth{S}}$ for $\textgoth{S} = h_{\ell,r}(\mathsf{S} \cap (\ell, r))$\footnote{Observe that $h_{\ell,r}(\mathsf{S} \cap (\ell, r))$ is open in $\mathbb{S}^{1}$, since $\mathsf{S} \cap (\ell, r)$ is open as a set in $\mathbb{R}$, and non-empty by assumption.} under $h_{\ell,r}^{-1}$, i.e.
\begin{align*}
	Q_{\mathsf{S}}^{\ell,r}(\theta, \mathsf{A}) := Q_{h_{\ell,r}(\mathsf{S} \cap (\ell,r))}\big(h_{\ell,r}(\theta),h_{\ell,r}(\mathsf{A})\big), \qquad \theta \in \mathsf{S} \cap (\ell,r), \mathsf{A} \in \mathcal{B}\big(\mathsf{S}\cap(\ell,r)\big).
\end{align*}
Observe that this agrees with the definition made in \eqref{Eq: Definition shrinkage kernel within GSS}, where we have $\theta = 0$ and $\mathsf{S} = L(x,v,t)$ for $x \in \mathsf{M}$, $v \in \mathbb{S}_x^{d-1}$ and $t \in (0, p(x))$.

We briefly discuss the measureability of the shrinkage kernel in the arguments $(\theta, \ell, r)$.
\begin{remark}
	Let $L$ and $R$ be two random variables independent of $\Theta$ and $(\Upgamma_n, \Upgamma_n^{\min}, \Upgamma_n^{\max})_{n \in \mathbb{N}}$ satisfying $L < R$ almost surely, and let $ \mathsf{S} \in \mathcal{B}(\mathbb{R}) $, $\mathsf{A} \in \mathcal{B}(\mathsf{S})$.
	By a disintegration argument, we have for all $\theta \in [0, 2\pi)$, $\ell,r \in \mathbb{R}$ that
	\begin{align*}
		f_k(\theta, \ell, r) &:= \mathbb{E}\left( \mathbbm{1}_{h_{\ell,r}(\mathsf{A}\cap \mathsf{S} \cap (\ell,r))}(\Gamma_k)\prod_{i =1}^{k-1}\mathbbm{1}_{[0,2\pi)\setminus h_{\ell,r}(\mathsf{S} \cap (\ell,r))}(\Gamma_i) \mid \Theta = \theta\right)\\
		&= \mathbb{E}\left( \mathbbm{1}_{h_{L,R}(\mathsf{A}\cap \mathsf{S} \cap (L,R))}(\Gamma_k)\prod_{i =1}^{k-1}\mathbbm{1}_{[0,2\pi)\setminus h_{L,R}(\mathsf{S} \cap (L,R))}(\Gamma_i) \mid \Theta = \theta, L = \ell, R = r\right).
	\end{align*}
	Therefore
	\[
	Q_{\mathsf{S}}^{\ell,r}(\theta, \mathsf{A}\cap (\ell,r)) = \sum_{k = 1}^{\infty} f_k\left(h_{\ell,r}(\theta), \ell, r\right), \qquad \theta \in \mathsf{S}\cap (\ell,r), \ell < r,
	\]
	is measurable in $(\theta, \ell,r)$.
	The equality above holds by definition of the shrinkage kernel, see also \citep[Proof of Theorem 2.10]{ReversibilityEllipticalSliceSampler}.
	Consequently 
	\[
	\{(\theta, \ell, r )\in \mathbb{R}^3 \mid \ell < r, \theta \in \mathsf{S} \cap (\ell, r)\} \times \mathcal{B}(\mathsf{S}) \to [0,1], 
	\qquad \big((\theta, \ell, r), B\big) \mapsto Q_{\mathsf{S}}^{\ell,r}(\theta, \mathsf{B}\cap (\ell,r))
	\]
	is a transition kernel.
\end{remark}

\begin{remark}\label{R: Shrinkage and stepping out are compatible}
	Note that the combination of stepping-out distribution and shrinkage kernel as in \eqref{Eq: Definition of geodesic slice sampling kernel} is valid, i.e. the random interval generated by the stepping-out procedure can be used as an input for the shrinkage procedure.
	To see this, fix $x \in \mathsf{M}$, $v \in \mathbb{S}_x^{d-1}$ and $t \in (0,p(x))$.
	Let $\boldsymbol{L}_{L(x,v,t)}^{(0)}$ and $ \boldsymbol{R}_{L(x,v,t)}^{(0)}$ be as in \eqref{Eq: Defintion stepping-out random variables}.
	We need to verify that
	\begin{itemize}
		\item $L(x,v,t)$ is open,
		\item $0 \in L(x,v,t) \cap \big(\boldsymbol{L}_{L(x,v,t)}^{(0)}, \boldsymbol{R}_{L(x,v,t)}^{(0)}\big)$ almost surely. In particular this implies that
		this intersection is almost surely non-empty.
	\end{itemize}
	The lower semi-continuity of $p$ yields that $L(x,v,t)$ is open.
	Since $t \in (0,p(x))$, we have $0 \in L(x,v,t)$.
	Moreover, we have $0 \in \big(\boldsymbol{L}_{L(x,v,t)}^{(0)}, \boldsymbol{R}_{L(x,v,t)}^{(0)}\big)$ almost surely by construction.
	Therefore 0 is also almost surely contained in the intersection of these two sets.
\end{remark}

We provide two properties of the shrinkage kernel $Q_{\mathsf{S}}^{\ell,r}$ that are useful to derive the reversibility of the geodesic slice sampler.
Both are essentially extensions of corresponding properties of $Q_{\textgoth{S}}$.

\begin{lemma}\label{L: Reversibility of the pushforward kernel}
	Let $\ell, r \in \mathbb{R}$ and  $\mathsf{S} \in \mathcal{B}(\mathbb{R})$ an open set such that $(\ell,r) \cap \mathsf{S} \neq \emptyset$.
	Then the kernel $Q_{\mathsf{S}}^{\ell,r}$ is reversible with respect to $ \mathrm{Unif}\big(\mathsf{S}\cap(\ell,r)\big) $.
\end{lemma}

To obtain this result, we push the reversibility statement for $Q_{\textgoth{S}}$ formulated in \citep{ReversibilityEllipticalSliceSampler} forward to the shrinkage kernel on arbitrary half open intervals.

\begin{proof}
	By \citep[Theorem 2.10]{ReversibilityEllipticalSliceSampler} we know that $ Q_{h_{\ell,r}(\mathsf{S} \cap (\ell, r))} $ is reversible with respect to $ \mathrm{Unif}\big(h_{\ell,r}(\mathsf{S} \cap (\ell, r))\big)$.
	Observe that by \citep[Proposition 19]{rudolf2022robust} this implies that $ Q_{\mathsf{S}}^{\ell,r} $ is reversible with respect to $ (h_{\ell,r}^{-1})_\sharp\mathrm{Unif}\big(h_{\ell,r}(\mathsf{S} \cap (\ell, r))\big) =  \mathrm{Unif}\big(\mathsf{S} \cap (\ell, r)\big)$.

\end{proof}

The next lemma can be seen as describing the behaviour of the shrinkage kernel under U-turns.

\begin{lemma}\label{L: Shrinkage under rotation}                                                                
	Let $\ell, r \in \mathbb{R}$ such that $\ell < r$ and let $\mathsf{S} \in \mathcal{B}(\mathbb{R})$ be an open set such that $(\ell,r) \cap \mathsf{S} \neq \emptyset$.
	For all $\alpha \in \mathsf{S} \cap (\ell,r)$ and $\mathsf{A} \in \mathcal{B}(\Lambda_\alpha(\mathsf{S} \cap (\ell,r)))$ we have
	\[  
	Q_{\Lambda_\alpha(\mathsf{S})}^{\uplambda_\alpha(\ell,r)}(0, \mathsf{A})=Q_{\mathsf{S}}^{\ell,r}\big(\alpha, \Lambda_\alpha(\mathsf{A})\big),
	\]
	where $\Lambda_\alpha$ and $\uplambda_\alpha$ are defined as in \eqref{Eq: Definition Lambda_alpha} and \eqref{Eq: Definition uplambda_alpha} respectively.
\end{lemma}

In order to leverage a similar property of the kernel $Q_{\textgoth{S}}$ for showing the above statement, we need the collection of maps
\begin{align*}
	\widetilde{\Lambda}_\alpha : [0,2\pi) \to [0,2\pi), \qquad
	\theta \mapsto \alpha - \theta \mod 2 \pi
\end{align*}
indexed by $\alpha \in [0,2\pi)$.

\begin{proof}
	We aim to apply \citep[Lemma 2.12]{ReversibilityEllipticalSliceSampler}.
	Let $\ell, r \in \mathbb{R}$ such that $\ell < r$.
	Moreover, let $\mathsf{S} \in \mathcal{B}(\mathbb{R})$ be an open set such that $(\ell,r) \cap \mathsf{S} \neq \emptyset$,  and $\alpha \in (\ell,r)\cap \mathsf{S}$.
	Observe that for $ \theta \in \mathsf{S} \cap (\ell, r) $ we have
	\begin{align*}
		&\widetilde{\Lambda}_{h_{\ell,r}(\alpha)}\left(h_{\ell,r}(\theta)\right)
		=  \left( \frac{2\pi}{r-\ell}\alpha \mod 2\pi - \frac{2\pi}{r-\ell}\theta \mod 2 \pi \right) \mod 2\pi\\
		&\qquad= \left(\frac{2\pi}{r-\ell} (\alpha-\theta)\right) \mod 2\pi
		= \left(\frac{2\pi}{(\alpha-\ell)- (\alpha-r)} (\alpha-\theta)\right) \mod 2\pi\\
		&\qquad = h_{\uplambda_\alpha(\ell,r)}\left(\Lambda_\alpha(\theta)\right).
	\end{align*}
	Therefore
	\begin{equation*}
		\widetilde{\Lambda}_{h_{\ell,r}(\alpha)}\Big(h_{\ell,r}\big(\mathsf{S} \cap (\ell, r)\big)\Big) 
		= h_{\uplambda_\alpha(\ell,r)}\Big(\Lambda_\alpha\big(\mathsf{S} \cap (\ell,r)\big) \Big)
		= h_{\uplambda_\alpha(\ell,r)}\Big(\Lambda_\alpha(\mathsf{S}) \cap  \big(\Lambda_\alpha(r), \Lambda_\alpha(\ell) \big)\Big),
	\end{equation*}
	and
	\begin{equation*}
		\widetilde{\Lambda}_{h_{\ell,r}(\alpha)}\left(h_{\ell,r}\big(\Lambda_\alpha(\mathsf{A})\big)\right) = h_{\uplambda_\alpha(\ell,r)} (\mathsf{A})
	\end{equation*}
	for all $\mathsf{A} \in \mathcal{B}(\Lambda_\alpha(\mathsf{S} \cap (\ell,r)))$, as $\Lambda_\alpha^{-1} = \Lambda_\alpha$.
	Since $ \widetilde{\Lambda}_{h_{\ell,r}(\alpha)}\left(h_{\ell,r}(\alpha) \right) =h_{\ell,r}(\alpha) -h_{\ell,r}(\alpha)\mod 2\pi = 0$ and $\widetilde{\Lambda}_\alpha^{-1}= \widetilde{\Lambda}_\alpha  $, we get by \citep[Lemma 2.12, note that $g_\theta = \widetilde{\Lambda}_\alpha$]{ReversibilityEllipticalSliceSampler} that for $\mathsf{A}\in \mathcal{B}(\Lambda_\alpha(\mathsf{S} \cap (\ell,r)))$
	\begin{align*}
		Q_{\Lambda_\alpha(\mathsf{S})}^{\uplambda_\alpha(\ell,r)}(0, \mathsf{A})
		&= Q_{h_{\uplambda_\alpha(\ell,r)}\big(\Lambda_\alpha(\mathsf{S}) \cap (\Lambda_\alpha(r),\Lambda_\alpha(\ell))\big)} \left(0, h_{\uplambda_\alpha(\ell,r)}(\mathsf{A})\right)\\
		&= Q_{\widetilde{\Lambda}_{h_{\ell,r}(\alpha)}\big(h_{\ell,r}(\mathsf{S}\cap (\ell,r))\big)}\left( \widetilde{\Lambda}_{h_{\ell,r}(\alpha)}\left(h_{\ell,r}(\alpha) \right), \widetilde{\Lambda}_{h_{\ell,r}(\alpha)}\big(h_{\ell,r}(\Lambda_\alpha(\mathsf{A}))\big)\right)\\
		&=Q_{h_{\ell,r}(\mathsf{S} \cap (\ell,r))}\big( h_{\ell,r}(\alpha) , h_{\ell,r}(\Lambda_\alpha(\mathsf{A}))\big)\\
		&=Q_{\mathsf{S}}^{\ell,r}(\alpha, \Lambda_\alpha(\mathsf{A})). 
	\end{align*}
\end{proof}

\subsection{Reversibility of the geodesic slice sampler}\label{Sec: Proof reversibility}

Note that the \emph{tangent bundle }
\[ 
T\mathsf{M} := \bigcup_{x \in \mathsf{M}} \{x\} \times T_x\mathsf{M}
\]
of $\mathsf{M}$ is a smooth, $2d$-dimensional, connected manifold, see e.g. \citep[Section I.2.2]{Sakai}.
To prove Theorem \ref{Thm: Reversibility}, we employ the Riemannian structure of $T\mathsf{M}$.
We only briefly sketch how this Riemannian structure is introduced, for more details see \citep[Section II.4]{Sakai}.

We denote by
\[  
\mathrm{proj}_{\mathsf{M}} : T\mathsf{M} \to \mathsf{M}, \qquad (x,v) \mapsto x
\]
the projection map from the tangent bundle onto $\mathsf{M}$.
Observe that for $u=(x,v) \in T\mathsf{M}$ the tangent space $T_uT\mathsf{M}$ to $T\mathsf{M}$ at $u$ can be identified with the direct sum 
\[  
T_uT\mathsf{M} = T_x\mathsf{M} \oplus T_x\mathsf{M}.
\]
For $u=(x,v) \in \mathsf{M}$ this identification allows us to introduce a ``canonical'' metric on $T\mathsf{M}$ referenced to as the \emph{Sasaki metric} 
\[  
\mathfrak{G}_u(\eta, \widetilde{\eta}) := \mathfrak{g}_x\left(\eta_h, \widetilde{\eta}_h\right) + \mathfrak{g}_x(\eta_v, \widetilde{\eta}_v), \qquad \eta = (\eta_h, \eta_v), \widetilde{\eta} = (\widetilde{\eta}_h, \widetilde{\eta}_v) \in T_x\mathsf{M} \oplus T_x\mathsf{M}.
\]
Together with $\mathfrak{G}$, the tangent bundle is a Riemannian manifold. 
However, in fact we are more interested in a submanifold of $T\mathsf{M}$, that is, the \emph{unit tangent bundle}
\[ 
U\mathsf{M} := \bigcup_{x \in \mathsf{M}} U_x\mathsf{M} := \bigcup_{x \in \mathsf{M}} \{x\}\times \mathbb{S}^{d-1}_x.
\]
We call the Riemannian measure $\nu_{\mathfrak{G}}$, which is induced by the Sasaki metric $\mathfrak{G}$ onto the unit tangent bundle $U\mathsf{M}$, the \emph{Liouville measure}.
Observe that the restriction $\mathrm{proj}_{\mathsf{M}}\vert_{U\mathsf{M}}$ is a Riemannian submersion, and for  $x \in \mathsf{M}$ the fibre $\mathrm{proj}_{\mathsf{M}}\vert_{U\mathsf{M}}^{-1}(x)$ equipped with the metric induced by $\mathfrak{g}$ is $\left(\mathbb{S}^{d-1}_x, \widehat{\mathfrak{g}}_x\right)$. Note that additionally $\nu_{\widehat{\mathfrak{g}}_x}(\mathbb{S}^{d-1}_x) = \nu_{\widehat{\mathfrak{g}}}(\mathbb{S}^{d-1})$ for all $x \in \mathsf{M}$.
Applying Fubini's theorem for manifolds (see \citep[Theorem II.5.6]{Sakai}), this yields a nice expression for the Liouville measure, namely we have 
\begin{equation}\label{Eq: Nice representation for Liouville measure}
	\int_{U\mathsf{M}} f(x,v) \ \nu_{\mathfrak{G}}\big(\d (x,v) \big) = \int_{\mathsf{M}} \int_{\mathbb{S}_x^{d-1}} f(x,v)\ \nu_{\widehat{\mathfrak{g}}_x}(\d v)\, \nu_{\mathfrak{g}}(\d x)
\end{equation}
for all measurable functions $f: U\mathsf{M} \to [0,\infty)$.

In the following, we combine the Liouville measure with a family of maps $T^{(\theta)}: U\mathsf{M} \to U\mathsf{M}$ indexed by $\theta \in \mathbb{R}$ which can be interpreted as ``walking along the geodesic specified by $u \in U\mathsf{M}$ with step length $\theta$ and then doing a U-turn''.
To this end,  for $ \theta \in \mathbb{R} $ we denote the \emph{geodesic flow} by
\begin{align*}
	\phi_\theta: U\mathsf{M} \to U\mathsf{M}, \qquad
	(x,v) \mapsto \left(\left.\gamma_{(x,v)}(\theta), \frac{\d \gamma_{(x,v)}}{\d t}\right\vert_\theta\right).
\end{align*}
For more details on the geodesic flow see \citep[Section II.4.II]{Sakai}.
Moreover, we define a flip on the unit tangent bundle
\begin{align*}
	\mathfrak{I}: U\mathsf{M}  \to U\mathsf{M}, \qquad
	(x,v) \mapsto (x,-v).
\end{align*}
Then we set
\[  
T^{(\theta)} :=\mathfrak{I} \circ \phi_\theta, \qquad \theta \in \mathbb{R}.
\]
Observe that the Liouville measure is invariant under the geodesic flow (see \citep[Exercise II.16]{Sakai}) and under the flip $\mathfrak{I}$ (see \citep[Lemma 1.34]{Paternain}).
Therefore we have due to the representation in \eqref{Eq: Nice representation for Liouville measure} for all measurable functions $f: U\mathsf{M} \to [0,\infty)$ and all $\theta \in \mathbb{R}$ that
\begin{equation}\label{Eq: Invaraince under Ttheta}
	\begin{split}
		&\int_{\mathsf{M}} \int_{\mathbb{S}_x^{d-1}} f\left(T^{(\theta)}(x,v)\right)\ \sigma_{d-1}^{(x)}(\d v)\, \nu_{\mathfrak{g}}(\d x)
		= \frac{1}{\nu_{\widehat{\mathfrak{g}}}(\mathbb{S}^{d-1})}\int_{U\mathsf{M}} f\left(T^{(\theta)}(x,v)\right) \ \nu_{\mathfrak{G}}\big(\d (x,v) \big) \\
		&\qquad = \frac{1}{\nu_{\widehat{\mathfrak{g}}}(\mathbb{S}^{d-1})} \int_{U\mathsf{M}} f(x,v) \ \nu_{\mathfrak{G}}\big(\d (x,v) \big) 
		= \int_{\mathsf{M}} \int_{\mathbb{S}_x^{d-1}} f(x,v)\ \sigma_{d-1}^{(x)}(\d v)\, \nu_{\mathfrak{g}}(\d x),
	\end{split}
\end{equation}
i.e. the Liouville measure is invariant under $T^{(\theta)}$.

We shed some further light on the interaction of $T^{(\theta)}$ and the geodesics.
\begin{remark}\label{R: Interaction Ttheta geodesic}
	Let  $ x \in \mathsf{M} $, $ v \in \mathbb{S}_x^{d-1} $ and $ \theta, \alpha \in \mathbb{R} $.
	Using the rescaling property of geodesics (see \citep[Lemma 5.18]{Lee}) and the chain rule, we have
	\begin{align*}
		(\mathfrak{I} \circ \phi_\theta )(x,v) &= \left(\gamma_{(x,v)}(\theta), -\frac{\d \gamma_{(x,v)}}{\d t}\vert_\theta\right)
		= \left(\gamma_{(x,-v)}(-\theta), \frac{\d \gamma_{(x,-v)}}{\d t}\vert_{-\theta}\right) 
		= (\phi_{-\theta} \circ \mathfrak{I})(x,v).
	\end{align*}
	Since\footnote{This is a basic property of a flow, see e.g. \citep[Chapter 9]{LeeSmooth}. Observe that the geodesic flow is a flow on the (unit) tangent bundle.} $\phi_{\theta} \circ \phi_{-\alpha} =\phi_{\theta-\alpha}$, this yields
	\begin{align*}
		\gamma_{T^{(\alpha)}(x,v)}(\theta) &= \mathrm{proj}_{\mathsf{M}}\big(\phi_\theta(T^{(\alpha)}(x,v))\big)
		=  \mathrm{proj}_{\mathsf{M}}\big((\phi_\theta \circ \mathfrak{I} \circ \phi_\alpha)(x,v)\big)\\
		&=  \mathrm{proj}_{\mathsf{M}}\big((\phi_\theta  \circ \phi_{-\alpha}\circ \mathfrak{I})(x,v)\big)
		=  \mathrm{proj}_{\mathsf{M}}\big(( \phi_{\theta-\alpha}\circ \mathfrak{I})(x,v)\big)\\
		&=  \mathrm{proj}_{\mathsf{M}}\left((\mathfrak{I} \circ  \phi_{\Lambda_\alpha(\theta)})(x,v)\right)
		= \mathrm{proj}_{\mathsf{M}}\left( \phi_{\Lambda_\alpha(\theta)}(x,v)\right)
		=\gamma_{(x,v)}\big(\Lambda_\alpha(\theta)\big).
	\end{align*}
	In particular this implies
	\begin{align*}
		L\big(T^{(\alpha)}(x,v), t\big) &= \{\theta \mid p\big(\gamma_{T^{(\alpha)}(x,v)}(\theta)\big) > t\}
		=  \{\theta \mid p\left( \gamma_{(x,v)}\big(\Lambda_\alpha(\theta)\big) \right) > t\}
		=\Lambda_\alpha\left(L(x,v,t)\right),
	\end{align*}
	as $ \Lambda_\alpha^{-1} = \Lambda_\alpha $.
\end{remark}

Now we can prove the reversibility of the geodesic slice sampler.

\begin{proof}[of reversibility in Theorem \ref{Thm: Reversibility}]
	Set $\mathsf{W}:=\{x \in \mathsf{M} \mid p(x) > 0\}$. 
	Since $\pi(\mathsf{W}^\complement)= 0$ and $K(x,\mathsf{W}^\complement)= 0$ for all $x \in \mathsf{W}$,
	for all $\mathsf{A}, \mathsf{B} \in \mathcal{B}(\mathsf{M})$ holds
	\[
		\int_{\mathsf{B}} K(x,\mathsf{A})\ \pi(\d x) = \int_{\mathsf{W}\cap \mathsf{B}} K(x,\mathsf{A})\ \pi(\d x)
		=\int_{\mathsf{W}\cap \mathsf{B}} K(x,\mathsf{A} \cap \mathsf{W})\ \pi(\d x).
	\]
	Therefore, it suffices to show that 
	\[
		\widetilde{K}: \mathsf{W} \times \mathcal{B}(\mathsf{W}),
		\qquad (x,\mathsf{A}) \mapsto \frac{1}{p(x)} \int_{(0, p(x))} K_t(x,\mathsf{A})\ \mathrm{Leb}_{1}(\d t)
	\]
	is reversible with respect to $\pi$.
	By virtue of \cite[Lemma 1]{latuszynski2024convergence} this follows if $K_t$ is reversible with respect to $\nu_{\mathfrak{g}}\vert_{\mathsf{L}}(t)$ for all $t \in (0, \|p\|_{\infty})$.
	
	Let $ t \in   (0, \|p\|_{\infty}) $ and $ \mathsf{A},\mathsf{B} \in \mathcal{B}(\mathsf{W}) $.
	After introducing the uniform distribution on $ L(x,v,t) $, we get by Lemma \ref{L: Key identity for reversiblity - measurable functions} that
	\begin{align*}
		&\int_{L(t)\cap \mathsf{A}} K_t(x,\mathsf{B})\  \nu_{\mathfrak{g}}(\d x) \\
		&\quad=\int_{\mathsf{M}}\int_{\mathbb{S}_x^{d-1}}\int_{\mathbb{R}^2} \int_{(\ell,r)} \mathbbm{1}_{L(t)\cap \mathsf{A}}(x) \mathbbm{1}_{\mathsf{B}}\left(\gamma_{(x,v)}(\theta) \right) 
		\ Q_{L(x,v,t)}^{\ell,r}(0, \d \theta) \, \xi_{L(x,v,t)}^{(0)}\big(\d (\ell,r) \big)\, \sigma_{d-1}^{(x)}(\d v)\, \nu_{\mathfrak{g}}(\d x)\\
		& \quad=\int_{\mathbb{R}} \int_{\mathsf{M}}  \int_{\mathbb{S}_x^{d-1}}  \int_{\mathbb{R}^2} \int_{(\ell,r)} \mathbbm{1}_{L(t)\cap \mathsf{A}}(x) \mathbbm{1}_{\mathsf{B}} \left(\gamma_{(x,v)}(\theta) \right) \frac{1}{\mathrm{Leb}_1(L(x,v,t) \cap (\ell,r))} \mathbbm{1}_{L(x,v,t) \cap (\ell,r)}(\alpha)\\
		&\hspace{5cm} \times Q_{L(x,v,t)}^{\ell,r}(0, \d \theta) \ \xi_{L(x,v,t)}^{(0)}\big(\d (\ell,r) \big) \ \sigma_{d-1}^{(x)}(\d v) \  \nu_{\mathfrak{g}}(\d x) \ \mathrm{Leb}_1(\d \alpha)\\
		& \quad=\int_{\mathbb{R}} \int_{\mathsf{M}}  \int_{\mathbb{S}_x^{d-1}}  \int_{\mathbb{R}^2} \int_{(\ell,r)} \mathbbm{1}_{L(t)\cap \mathsf{A}}(x)  \mathbbm{1}_{\mathsf{B}}\left(\gamma_{(x,v)}(\theta) \right) \frac{1}{\mathrm{Leb}_1(L(x,v,t) \cap (\ell,r))} \mathbbm{1}_{L(x,v,t)}(\alpha) \mathbbm{1}_{(\ell,r)}(0)\\
		&\hspace{5cm} \times Q_{L(x,v,t)}^{\ell,r}(0, \d \theta) \ \xi_{L(x,v,t)}^{(\alpha)}\big(\d (\ell,r) \big) \ \sigma_{d-1}^{(x)}(\d v) \  \nu_{\mathfrak{g}}(\d x) \ \mathrm{Leb}_1(\d \alpha).
	\end{align*}
	Using \eqref{Eq: Invaraince under Ttheta}, we obtain 
	\begin{align*}
		&\int_{L(t)\cap \mathsf{A}} K_t(x,\mathsf{B})\  \nu_{\mathfrak{g}}(\d x) \\
		& \quad=\int_{\mathbb{R}} \int_{\mathsf{M}}  \int_{\mathbb{S}_x^{d-1}}  \int_{\mathbb{R}^2} \int_{(\ell,r)}  \frac{\mathbbm{1}_{L(t)\cap \mathsf{A}}\left(\gamma_{(x,v)}(\alpha)\right) \mathbbm{1}_{\mathsf{B}}\left(\gamma_{T^{(\alpha)}(x,v)}(\theta) \right)\mathbbm{1}_{L\left(T^{(\alpha)}(x,v),t\right)}(\alpha) \mathbbm{1}_{(\ell,r)}(0)}{\mathrm{Leb}_1\big(L(T^{(\alpha)}(x,v),t) \cap (\ell,r)\big)}\\
		&\hspace{4cm}
		\times Q_{L\left(T^{(\alpha)}(x,v),t\right)}^{\ell,r}(0, \d \theta) \, \xi_{L\left(T^{(\alpha)}(x,v),t\right)}^{(\alpha)}\big(\d (\ell,r) \big) \, \sigma_{d-1}^{(x)}(\d v) \,  \nu_{\mathfrak{g}}(\d x) \, 	\mathrm{Leb}_1(\d \alpha).
	\end{align*}
	For all $\alpha \in \mathbb{R}$, we have
	\[  
	\gamma_{(x,v)}(\alpha) \in L(t)\quad \Leftrightarrow \quad p\left(\gamma_{(x,v)}(\alpha) \right) > t\quad \Leftrightarrow \quad \alpha \in L(x,v,t),
	\]
	and by Remark \ref{R: Interaction Ttheta geodesic}
	\[  
	\alpha \in L\big(T^{(\alpha)}(x,v), t\big)\quad \Leftrightarrow\quad p(x) = p\left(\gamma_{(x,v)}(\alpha-\alpha)\right) = p\left(\gamma_{T^{(\alpha)}(x,v)}(\alpha)\right) > t\quad \Leftrightarrow \quad x \in L(t).
	\]
	If we also apply Remark \ref{R: Interaction Ttheta geodesic} to $\mathbbm{1}_{\mathsf{\mathsf{B}}}$ and to the set in the 1-dimensional Lebesgue measure in the numerator, we overall obtain
	\begin{align*}
		&\int_{L(t)\cap \mathsf{A}} K_t(x,\mathsf{B})\  \nu_{\mathfrak{g}}(\d x) \\
		&\quad =\int_{\mathbb{R}} \int_{\mathsf{M}}  \int_{\mathbb{S}_x^{d-1}}  \int_{\mathbb{R}^2} \int_{(\ell,r)}  \frac{\mathbbm{1}_{L(x,v,t)}\left(\alpha\right) \mathbbm{1}_{\mathsf{A}}\left(\gamma_{(x,v)}(\alpha)\right) \mathbbm{1}_{\mathsf{B}}\left(\gamma_{(x,v)}\big(\Lambda_\alpha(\theta)\big) \right)\mathbbm{1}_{L(t)}(x) \mathbbm{1}_{(\ell,r)}(0)}{\mathrm{Leb}_1\left(\Lambda_\alpha(L(x,v,t)) \cap (\ell,r)\right)}\\
		&\hspace{4cm}
		\times Q_{L(T^{(\alpha)}(x,v),t)}^{\ell,r}(0, \d \theta) \ \xi_{L(T^{(\alpha)}(x,v),t)}^{(\alpha)}\big(\d (\ell,r) \big) \ \sigma_{d-1}^{(x)}(\d v) \  \nu_{\mathfrak{g}}(\d x) \ 	\mathrm{Leb}_1(\d \alpha).
	\end{align*}
	Then Lemma \ref{L: stepping out under rotation} together with Remark \ref{R: Interaction Ttheta geodesic} yields
	\begin{align*}
		&\int_{L(t)\cap \mathsf{A}} K_t(x,\mathsf{B})\  \nu_{\mathfrak{g}}(\d x) \\
		&\quad=\int_{\mathbb{R}} \int_{\mathsf{M}}  \int_{\mathbb{S}_x^{d-1}}  \int_{\mathbb{R}^2} \int_{\Lambda_\alpha((\ell,r))}  \frac{\mathbbm{1}_{L(x,v,t)}\left(\alpha\right) \mathbbm{1}_{\mathsf{A}}\left(\gamma_{(x,v)}(\alpha)\right) \mathbbm{1}_{\mathsf{B}}\left(\gamma_{(x,v)}\big(\Lambda_\alpha(\theta)\big) \right)\mathbbm{1}_{L(t)}(x) \mathbbm{1}_{\Lambda_\alpha((\ell,r))}(0)}{\mathrm{Leb}_1\big(\Lambda_\alpha(L(x,v,t)\cap(\ell,r))\big)}\\
		&\hspace{4cm}
		\times Q_{L(T^{(\alpha)}(x,v),t)}^{\uplambda_\alpha(\ell,r)}(0, \d \theta) \ \xi_{L(x,v,t)}^{(0)}\big(\d (\ell,r) \big) \ \sigma_{d-1}^{(x)}(\d v) \  \nu_{\mathfrak{g}}(\d x) \ 	\mathrm{Leb}_1(\d \alpha).
	\end{align*}
	Observe that
	\begin{align*}
		\mathrm{Leb}_1\big(\Lambda_\alpha(L(x,v,t)\cap(\ell,r))\big) = \mathrm{Leb}_1\big(L(x,v,t)\cap(\ell,r)\big)
	\end{align*}
	for all $x \in M$, $v \in \mathbb{S}_x^{d-1}$ and  $\alpha, \ell, r \in \mathbb{R}$.
	Hence
	\begin{align*}
		&\int_{L(t)\cap \mathsf{A}} K_t(x,\mathsf{B})\  \nu_{\mathfrak{g}}(\d x)\\
		& \quad=\int_{\mathbb{R}} \int_{\mathsf{M}}  \int_{\mathbb{S}_x^{d-1}}  \int_{\mathbb{R}^2} \int_{\Lambda_\alpha((\ell,r))}  \frac{\mathbbm{1}_{L(x,v,t)}\left(\alpha\right) \mathbbm{1}_{\mathsf{A}}\left(\gamma_{(x,v)}(\alpha)\right) \mathbbm{1}_{\mathsf{B}}\left(\gamma_{(x,v)}\big(\Lambda_\alpha(\theta)\big) \right)\mathbbm{1}_{L(t)}(x) \mathbbm{1}_{\Lambda_\alpha((\ell,r))}(0)}{ \mathrm{Leb}_1(L(x,v,t)\cap(\ell,r))}\\
		&\hspace{4cm}
		\times Q_{L(T^{(\alpha)}(x,v),t)}^{\uplambda_\alpha(\ell,r)}(0, \d \theta) \ \xi_{L(x,v,t)}^{(0)}\big(\d (\ell,r) \big) \ \sigma_{d-1}^{(x)}(\d v) \  \nu_{\mathfrak{g}}(\d x) \ 	\mathrm{Leb}_1(\d \alpha).
	\end{align*}
	Note that due to Remark \ref{R: Shrinkage and stepping out are compatible} and Remark \ref{R: Interaction Ttheta geodesic}, we may apply Lemma \ref{L: Shrinkage under rotation}, such that we obtain
	\begin{align*}
		&\int_{L(t)\cap \mathsf{A}} K_t(x,\mathsf{B})\  \nu_{\mathfrak{g}}(\d x) \\
		&\quad=\int_{\mathbb{R}} \int_{\mathsf{M}}  \int_{\mathbb{S}_x^{d-1}}  \int_{\mathbb{R}^2} \int_{(\ell,r)} \mathbbm{1}_{L(x,v,t)}\left(\alpha\right) \mathbbm{1}_{\mathsf{A}}\left(\gamma_{(x,v)}(\alpha)\right) \mathbbm{1}_{\mathsf{B}}\left(\gamma_{(x,v)}(\theta) \right) \frac{1}{\mathrm{Leb}_1(L(x,v,t)\cap(\ell,r))}\\
		&\hspace{2cm}\mathbbm{1}_{L(t)}(x)  \mathbbm{1}_{\Lambda_\alpha((\ell,r))}(0)
		\ Q_{L(x,v,t)}^{\ell,r}(\alpha, \d \theta) \ \xi_{L(x,v,t)}^{(0)}\big(\d (\ell,r) \big) \ \sigma_{d-1}^{(x)}(\d v) \  \nu_{\mathfrak{g}}(\d x) \ 	\mathrm{Leb}_1(\d \alpha)\\
		& \quad=\int_{\mathbb{R}} \int_{\mathsf{M}}  \int_{\mathbb{S}_x^{d-1}}  \int_{\mathbb{R}^2} \int_{(\ell,r)} \mathbbm{1}_{L(x,v,t)}\left(\alpha\right) \mathbbm{1}_{\mathsf{A}}\left(\gamma_{(x,v)}(\alpha)\right) \mathbbm{1}_{\mathsf{B}}\left(\gamma_{(x,v)}(\theta) \right) \frac{1}{\mathrm{Leb}_1(L(x,v,t)\cap(\ell,r))}\\
		&\hspace{2cm}\mathbbm{1}_{L(t)}(x)  \mathbbm{1}_{(\ell,r)}(\alpha)
		\ Q_{L(x,v,t)}^{\ell,r}(\alpha, \d \theta) \ \xi_{L(x,v,t)}^{(0)}\big(\d (\ell,r) \big) \ \sigma_{d-1}^{(x)}(\d v) \  \nu_{\mathfrak{g}}(\d x) \ 	\mathrm{Leb}_1(\d \alpha)\\
		& \quad= \int_{\mathsf{M}}  \int_{\mathbb{S}_x^{d-1}} \int_{\mathbb{R}^2}\frac{1}{\mathrm{Leb}_1(L(x,v,t)\cap(\ell,r))} \int_{L(x,v,t)\cap (\ell,r)}  \int_{(\ell,r)}  \mathbbm{1}_{\mathsf{A}}\left(\gamma_{(x,v)}(\alpha)\right) \mathbbm{1}_{\mathsf{B}}\left(\gamma_{(x,v)}(\theta) \right) \\
		&\hspace{2cm}\mathbbm{1}_{L(t)}(x) 
		\ Q_{L(x,v,t)}^{\ell,r}(\alpha, \d \theta)\ 	\mathrm{Leb}_1(\d \alpha) \ \xi_{L(x,v,t)}^{(0)}\big(\d (\ell,r) \big) \ \sigma_{d-1}^{(x)}(\d v) \  \nu_{\mathfrak{g}}(\d x) .
	\end{align*}
	This expression is symmetric in $\mathsf{A}$ and $\mathsf{B}$, because of the reversibility of the shrinkage kernel. Namely Lemma \ref{L: Reversibility of the pushforward kernel} yields
	\begin{align*}
		&\int_{L(t)\cap \mathsf{A}} K_t(x,\mathsf{B})\  \nu_{\mathfrak{g}}(\d x)\\
		& \quad= \int_{\mathsf{M}}  \int_{\mathbb{S}_x^{d-1}} \int_{\mathbb{R}^2}\frac{1}{\mathrm{Leb}_1(L(x,v,t)\cap(\ell,r))} \int_{L(x,v,t)\cap (\ell,r)}  \int_{(\ell,r)}  \mathbbm{1}_{\mathsf{A}}\left(\gamma_{(x,v)}(\theta)\right) \mathbbm{1}_{\mathsf{B}}\left(\gamma_{(x,v)}(\alpha) \right) \\
		&\hspace{2cm}\mathbbm{1}_{L(t)}(x) 
		\ Q_{L(x,v,t)}^{\ell,r}(\alpha, \d \theta)\ 	\mathrm{Leb}_1(\d \alpha) \ \xi_{L(x,v,t)}^{(0)}\big(\d (\ell,r) \big) \ \sigma_{d-1}^{(x)}(\d v) \  \nu_{\mathfrak{g}}(\d x).
	\end{align*}
	Consequently,
	\begin{align*}
		\int_{L(t)\cap \mathsf{A}} K_t(x,\mathsf{B})\  \nu_{\mathfrak{g}}(\d x) =\int_{L(t)\cap \mathsf{B}} K_t(x,\mathsf{A})\  \nu_{\mathfrak{g}}(\d x).
	\end{align*}
\end{proof}

\subsection{Ergodicity of the geodesic slice sampler}\label{Sec: proof ergodicity}
As in the previous section, let $\mathsf{W}:= \{x \in \mathsf{M} \mid p(x) > 0\}$ and
$\widetilde{K}: \mathsf{W} \times \mathcal{B}(\mathsf{W}),(x,\mathsf{A}) \mapsto K(x,\mathsf{A})$ be the restriction of $K$ to $\mathsf{W}$.
By \cite[Theorem 1]{tierney1994markov}, the following lemma implies that $\lim_{n \to \infty} d_{\mathrm{tv}}(\widetilde{K}^n(x, \cdot), \pi)=0$ holds for $\pi$-almost all $x \in \mathsf{M}$. 
Since $\widetilde{K}^n(x, \cdot)= K^n(x, \cdot)$ for all $x \in \mathsf{W}$, $n \in \mathbb{N}$, we get the ergodicity statement of Theorem \ref{Thm: Reversibility}.
\begin{lemma}\label{L: Irreducibility and aperiodicity}
	Let the assumptions of Theorem \ref{Thm: Reversibility} be satisfied.
	Then the kernel $\widetilde{K}$ is $\pi$-irreducible and aperiodic.
\end{lemma}
The proof of this lemma requires some preparation.
Let $x \in \mathsf{M}$.
For all $v \in \mathbb{S}_x^{d-1}$, we denote by $t_{\mathrm{cut}}(x,v)$ the cut time of $(x,v)$.
Then $\mathrm{Cut}(x)= \{\gamma_{(x,v)}(t_{\mathrm{cut}}(x,v)) \mid v \in \mathbb{S}_x^{d-1}, t_{\mathrm{cut}}(x,v)< \infty\}$ is called the \emph{cut locus} at $x$, and $D_{\mathrm{inj}}= \{\theta v \in T_x\mathsf{M} \mid v \in \mathbb{S}^{d-1}_x, \theta \in [0, t_{\mathrm{cut}}(x,v)) \}$ is called the \emph{injectivity domain} at $x$.
Observe that the exponential map $\mathrm{Exp}_x:T_x\mathsf{M} \to \mathsf{M}$ at $x$ restricted to the injectivity domain is a diffeomorphism onto $\mathsf{M} \setminus \mathrm{Cut}(x)$. For more details related to cut points see \cite[Chapter 10]{Lee}.
Additionally fix an isometry $\Phi_x : \mathbb{R}^{d} \to T_x\mathsf{M}$, 
where $\mathbb{R}^{d}$ is equipped with the standard inner product and $T_x\mathsf{M}$ with $\mathfrak{g}_x$, respectively.
Set $\varphi_{x}= \Phi_x^{-1} \circ \mathrm{Exp}_x^{-1}: \mathsf{M}\setminus \mathrm{Cut}(x) \to \Phi_x^{-1}(D_{\mathrm{inj}}(x))$.
Then $(\mathsf{M}\setminus\mathrm{Cut}(x), \varphi_{x})$ is a coordinate neighbourhood of $\mathsf{M}$, called \emph{normal coordinate neighbourhood}, 
which is specially adapted to the geodesic emanating at $x$.
\begin{remark}\label{R: Relation Riemannian and Lebesgue mass}
	Let $x \in \mathsf{M}$ and $\mathsf{A} \in \mathcal{B}(\mathsf{M})$.
	Since the cut locus at $x$ is a $\nu_{\mathfrak{g}}$-null set, we have by construction of the Riemannian measure that
	\[
		\nu_{\mathfrak{g}}(\mathsf{A}) = \nu_{\mathfrak{g}}(\mathsf{A}\setminus\mathrm{Cut}(x)) = \int_{\mathbb{R}^{d}} \mathbbm{1}_{\varphi_x(\mathsf{A}\setminus\mathrm{Cut}(x))}(z) \cdot \left(\sqrt{\det(\mathfrak{g}, \varphi_x)} \circ \varphi_x^{-1}\right)(z)\ \mathrm{Leb}_{d}(\d z).
	\]
	Therefore $\nu_{\mathfrak{g}}(\mathsf{A}) > 0$ implies $\mathrm{Leb}_{d}(\varphi_x(\mathsf{A}\setminus\mathrm{Cut}(x)))>0$.
\end{remark}
\begin{lemma}\label{L: reach a given set}
	Assume the same setting as in Theorem \ref{Thm: Reversibility}.
	Let $\mathsf{A} \in \mathcal{B}(\mathsf{W})$ with $\pi(\mathsf{A}) > 0$.
	Moreover, let $x \in \mathsf{W}$ such that for all $y \in \varphi_{x}(\mathsf{A}\setminus \mathrm{Cut}(x))$ holds $\|y\| < w/2$.
	Then we have $\widetilde{K}(x,\mathsf{A})> 0$.
\end{lemma}
\begin{proof}
	Set $\mathsf{R}=\{(\ell, r) \in \mathbb{R}^2 \mid \ell < 0, r > w/2\}$.
	Observe that
	\[
		\xi_{\mathsf{L}(x,v,t)}^{(0)}(\mathsf{R}) \geq \mathbb{P}\left (R_1^{(0)} > \frac{w}{2}\right ) = \mathbb{P}\left (\Upsilon \in \left [0, \frac{w}{2}\right ]\right ) = \frac{1}{2},
		\qquad v \in \mathbb{S}_x^{d-1}, t \in (0, \infty)
	\]
	by \eqref{Eq: Monotonicity L, R} and \eqref{Eq: Definition L_i, R_j}.
	Moreover \eqref{Eq: Bound tau, T} and \cite[proof of Theorem 2.10]{ReversibilityEllipticalSliceSampler} imply 
	\[
		Q^{\ell,r}_{\mathsf{L}(x,v,t)}(0, \cdot) \geq (mw)^{-1} \mathrm{Leb}_{1}\big(\cdot \cap\mathsf{L}(x,v,t)\cap (\ell,r)\big), 
	\]
	for all $v \in \mathbb{S}_x^{d-1}$, $ t \in (0, \infty)$ and $\xi_{\mathsf{L}(x,v,t)}^{(0)}$-almost all $(\ell,r)\in \mathbb{R}^2$.
	We obtain
	\begin{align*}
		&\widetilde{K}(x,\mathsf{A})
		\geq \begin{multlined}[0.85\textwidth][t]
			\frac{1}{p(x)\ mw} \int_{0}^{p(x)}\int_{\mathbb{S}_x^{d-1}} \int_{\mathsf{R}} \int_{\mathsf{L}(x,v,t)\cap (\ell,r) } \mathbbm{1}_{\mathsf{A}}(\gamma_{(x,v)}(\theta))\\
			\times \mathrm{Leb}_{1}(\d \theta)\, \xi_{\mathsf{L}(x,v,t)}^{(0)}\big(\d(\ell,r)\big)\, \sigma_{d-1}^{(x)}(\d v)\,\mathrm{Leb}_{1}(\d t)
		\end{multlined}\\
		&\quad\geq \begin{multlined}[0.85\textwidth][t]
			\frac{1}{p(x)\ mw}  \int_{0}^{p(x)}\int_{\mathbb{S}_x^{d-1}} \xi_{\mathsf{L}(x,v,t)}^{(0)}(\mathsf{R})  \int_{\mathsf{L}(x,v,t)\cap [0, \frac{w}{2}) } \mathbbm{1}_{\mathsf{A}}(\gamma_{(x,v)}(\theta))\ \mathrm{Leb}_{1}(\d \theta)\\
			 \times \sigma_{d-1}^{(x)}(\d v)\,\mathrm{Leb}_{1}(\d t)
		\end{multlined}\\
		&\quad\geq
			\frac{1}{2 p(x)\ mw}  \int_{0}^{p(x)}\int_{\mathbb{S}_x^{d-1}}  \int_{\mathsf{L}(x,v,t)\cap [0, \frac{w}{2}) } \mathbbm{1}_{\mathsf{A}\setminus\mathrm{Cut}(x)}(\gamma_{(x,v)}(\theta))
			\ \mathrm{Leb}_{1}(\d \theta)\, \sigma_{d-1}^{(x)}(\d v)\,\mathrm{Leb}_{1}(\d t)\\
		&\quad=	\frac{1}{2mw}  \int_{\mathbb{S}_x^{d-1}}  \int_{[0, \frac{w}{2}) } \left(1 \wedge \frac{p(\gamma_{(x,v)}(\theta))}{p(x)}\right) \mathbbm{1}_{\mathsf{A}\setminus\mathrm{Cut}(x)}(\gamma_{(x,v)}(\theta))\ 
		\mathrm{Leb}_{1}(\d \theta)\, \sigma_{d-1}^{(x)}(\d v),
	\end{align*}
	where for the last equality we used that $\mathbbm{1}_{\mathsf{L}(x,v,t)}(\theta) = \mathbbm{1}_{(-\infty, p(\gamma_{(x,v)}(\theta)))}(t)$.
	Exploiting $\sigma^{(x)}_{d-1} = \nu_{\widehat{\mathfrak{g}}}(\mathbb{S}^{d-1})^{-1} (\Phi_x)_\sharp\nu_{\widehat{\mathfrak{g}}}$ and $\gamma_{(x,v)}(\theta) = \mathrm{Exp}_x(\theta v)$,
	we get with a change from polar to Cartesian coordinates, see \cite[Theorem 16.22]{schilling2017measures}, that
	\begin{align*}
			&\widetilde{K}(x,\mathsf{A})
			\geq \begin{multlined}[0.65\textwidth][t]
					\frac{1}{2mw a}  \int_{\mathbb{S}^{d-1}} \int_{ [0, \infty) } 
				\left(1 \wedge \frac{p(\mathrm{Exp}_x\circ\Phi_x (\theta v))}{p(x)}\right) \mathbbm{1}_{[0, \frac{w}{2})}(\theta)
				\mathbbm{1}_{\mathsf{A}\setminus\mathrm{Cut}(x)}\big(\mathrm{Exp}_x \circ \Phi_x (\theta v)\big)\\
				\times \mathrm{Leb}_{1}(\d \theta)\, \nu_{\widehat{\mathfrak{g}}}(\d v)
			\end{multlined}\\
			&\quad \geq
				\begin{multlined}[t]
					\frac{1}{2mw a}  \int_{\mathbb{S}^{d-1}} \int_{ [0, \infty) } 
					\left(1 \wedge \frac{p(\varphi_{x}^{-1}(\theta v))}{p(x)}\right) \mathbbm{1}_{[0, \frac{w}{2})}(\theta) \mathbbm{1}_{\mathsf{A}\setminus\mathrm{Cut}(x)}\big(\varphi_{x}^{-1}(\theta v)\big) \mathbbm{1}_{\Phi_x^{-1}(D_{\mathrm{inj}}(x))}(\theta v)\\
					\times \mathrm{Leb}_{1}(\d \theta)\, \nu_{\widehat{\mathfrak{g}}}(\d v)
				\end{multlined}\\
			&\quad =
				\frac{1}{2mw a}  \int_{\mathbb{S}^{d-1}} \int_{ [0, \infty) } 
				\left(1 \wedge \frac{p(\varphi_{x}^{-1}(\theta v))}{p(x)}\right) \mathbbm{1}_{[0, \frac{w}{2})}(\theta) \mathbbm{1}_{\varphi_{x}(\mathsf{A}\setminus\mathrm{Cut}(x))}(\theta v)\ 
				 \mathrm{Leb}_{1}(\d \theta)\, \nu_{\widehat{\mathfrak{g}}}(\d v)\\
			&\quad =\frac{1}{2mw a}  \int_{\mathbb{R}^{d}} \|y\|^{1-d}
			\left(1 \wedge \frac{p(\varphi_{x}^{-1}(y))}{p(x)}\right) \mathbbm{1}_{[0, \frac{w}{2})}(\|y\|) \mathbbm{1}_{\varphi_x(\mathsf{A}\setminus\mathrm{Cut}(x))}(y)\ \mathrm{Leb}_{d}(\d y),
	\end{align*}
	where $ a := \nu_{\widehat{\mathfrak{g}}}(\mathbb{S}^{d-1})$.
	Hence by assumption
	\[
		\widetilde{K}(x,\mathsf{A})
		\geq \frac{1}{2mw}  \int_{\varphi_x(\mathsf{A}\setminus\mathrm{Cut}(x))} \|y\|^{1-d}
		\left(1 \wedge \frac{p(\varphi_{x}^{-1}(y))}{p(x)}\right) \ \mathrm{Leb}_{d}(\d y).
	\]
	Since $\pi$ is absolutely continuous with respect to the Riemannian measure, $\pi(\mathsf{A}) > 0$ implies $\nu_{\mathfrak{g}}(\mathsf{A})> 0$.
	Moreover, we have $p(\varphi_{x}^{-1}(y))> 0$ for all $y \in \varphi_{x}(\mathsf{A} \setminus \mathrm{Cut}(x))$ as $\mathsf{A} \subseteq \mathsf{W}$.
	Therefore by Remark \ref{R: Relation Riemannian and Lebesgue mass}, 
	the right hand side is an integral of a strictly positive function over a set of positive Lebesgue mass,
	which implies $\widetilde{K}(x, \mathsf{A}) > 0$.
\end{proof}
The Riemannian metric $\mathfrak{g}$ induces a metric on $\mathsf{M}$, which we denote as $\mathrm{dist}$, see \cite[Section V.3]{Boothby}.
For any set $\mathsf{A} \subseteq \mathsf{M}$ and $x \in \mathsf{M}$, we denote the distance from $x$ to $\mathsf{A}$ as $\mathrm{dist}(x,\mathsf{A})= \inf\{\mathrm{dist}(x,y)\mid y \in \mathsf{A}\}$.
\begin{proof}[of Lemma \ref{L: Irreducibility and aperiodicity}]
	Let $\mathsf{A} \in \mathcal{B}(\mathsf{W})$ with $\pi(\mathsf{A})> 0$.
	Because $\mathsf{M}$ is Lindelöff, the $\sigma$-subadditivity of $\pi$ yields, that there exists some $x_0 \in \mathsf{A}$ such that $\pi(\mathsf{A} \cap B_{w/4}(x_0)) > 0$,
	where $B_{w/4}(x_0)= \{y \in \mathsf{M} \mid\mathrm{dist}(x_0,y) < w/4\}$ denotes the metric ball with radius $w/4$ centred at $x_0$.
	Set $\mathsf{B}_1=\mathsf{A} \cap B_{w/4}(x_0)$ and define $\mathsf{B}_{n+1}=\{x \in \mathsf{W} \mid \mathrm{dist}(x, \mathsf{B}_n) < w/2\}$ for all $n \in \mathbb{N}$ inductively.
	Let $x \in \mathsf{B}_1$ and $y \in \varphi_{x}(\mathsf{B}_1\setminus\mathrm{Cut}(x))$.
	Then there exists some $z \in \mathsf{B}_1\setminus\mathrm{Cut}(x)$ such that $\varphi_{x}(z)=y$.
	By triangle inequality holds $\mathrm{dist}(x,z) \leq \mathrm{dist}(x,x_0) + \mathrm{dist}(x_0,z) < w/2$,
	such that there exist $v \in \mathbb{S}^{d-1}_x$ and $\theta \in [0, w/2)$ with $\mathrm{Exp}_x(\theta v) = z$.
	Since $\Phi_x$ is an isometry, this implies $\|y\| = \theta < w/2$.
	Therefore Lemma \ref{L: reach a given set} implies $\widetilde{K}(x, \mathsf{A}) \geq\widetilde{K}(x, \mathsf{B}_1) > 0$ for all $x \in \mathsf{B}_1$.
	We show next that for $x \in \mathsf{B}_{n+1}$ holds $\widetilde{K}(x, \mathsf{B}_n) > 0$ for all $n \in \mathbb{N}$.
	To this end, let $n \in \mathbb{N}$ and fix $x \in \mathsf{B}_{n+1}$.
	By construction of $\mathsf{B}_{n+1}$, the intersection $\mathsf{C}_x =B_{w/2}(x) \cap \mathsf{B}_{n}$ is non-empty and open as the intersection of two open sets.
	Therefore the definition of the Riemannian measure implies $\nu_{\mathfrak{g}}(\mathsf{C}_x)> 0$.
	Since also $\mathsf{C}_x \subseteq \mathsf{W}$, we have $\pi(\mathsf{C}_x) > 0$.
	Moreover, by similar arguments as above, for all  $y \in \varphi_{x}(\mathsf{C}_x\setminus\mathrm{Cut}(x))$ holds $\|y\| = \theta < w/2$.
	As $\mathsf{C}_x \subseteq \mathsf{B}_n$, we get $\widetilde{K}(x, \mathsf{B}_n) \geq \widetilde{K}(x, \mathsf{C}_x) > 0$ by Lemma \ref{L: reach a given set}.
	
	We now show $\widetilde{K}^n(x, \mathsf{A}) > 0$ for all $x \in \mathsf{B}_n$ and all $n \in \mathbb{N}$ by induction:
	The case $n=1$ is already treated. Assuming that the statement holds for some $n \in \mathbb{N}$,
	we get for all $x \in \mathsf{B}_{n+1}$ from $\widetilde{K}(x, \mathsf{B}_{n})> 0$ that
	\[
		\widetilde{K}^{n+1}(x,\mathsf{A}) = \int_{\mathsf{W}} \widetilde{K}^n(y, \mathsf{A}) \ \widetilde{K}(x,\d y)
		\geq \int_{\mathsf{B}_n} \widetilde{K}^n(y, \mathsf{A}) \ \widetilde{K}(x,\d y) > 0.
	\]
	Since $\mathsf{W}$ is by assumption path connected \citep[see][Propositions 4.23 and 4.26]{LeeTopological}, we have
	$\bigcup_{n \in \mathbb{N}} \mathsf{B}_n = \mathsf{W}$.
	This implies that $\mathsf{A}$ is accessible, such that $\widetilde{K}$ is indeed $\pi$-irreducible.
	
	To show aperiodicity, let now $\mathsf{A} \in \mathcal{B}(\mathsf{W})$ be accessible.
	Note that for an irreducible kernel the invariant probability measure is always a maximal irreducibility measure, see \cite[Theorem 9.2.15]{douc2018markov}.
	Hence $\pi(\mathsf{A}) > 0$.
	As $\mathsf{B}_{n} \subseteq \mathsf{B}_{n+1}$ for all $n \in \mathbb{N}$,
	the previous considerations imply that for all $x \in \mathsf{W}$ there exists some $k \in \mathbb{N}$ with $\widetilde{K}^n(x, \mathsf{A}) > 0$ for all $n \geq k$.
	Therefore, $\widetilde{K}$ is also aperiodic, see \cite[Theorem 9.3.10]{douc2018markov}.
\end{proof}

	We now turn to proving that, as noted in Remark \ref{R: Harris recurence}, the convergence in total variation distance in Theorem \ref{Thm: Reversibility} holds for all starting points $x \in \mathsf{W}$ if the hyperparameters of the geodesic slice sampler satisfy $mw \leq \mathrm{inj}(\mathsf{M})$.
	But first we start with elaborating on the alterations to the geodesic slice sampler that preserve this property for any choice of hyperparameters $m \in \mathbb{N}$ and $w \in (0, \infty)$, also mentioned in Remark \ref{R: Harris recurence}.
\begin{remark}\label{R: Modification of geodesic slice sampling}
	Assume in the setting of Theorem \ref{Thm: Reversibility} that the injectivity radius $\mathrm{inj}(\mathsf{M})$ of the Riemannian manifold $\mathsf{M}$ is strictly positive.
	We present two ways to modify the geodesic slice sampler that result in a sampler that converges to the target distribution $\pi$ in total variation distance for all starting points $x \in \mathsf{W}$
	and essentially allows a free choice of the hyperparameters $m$ and $w$.
	Firstly, for the first iteration of the geodesic slice sampler started at an arbitrary starting point in $\mathsf{W}$ we can choose the hyperparameters $m$ and $w$ such that $mw \leq \mathrm{inj}(\mathsf{M})$.
		This way, we simulate an initial distribution that is absolutely continuous with respect to the target distribution $\pi$,
		and after this first iteration we may fix the hyperparameters $m$ and $w$ regardless of the injectivity radius.
	
		Alternatively, we can fix some (small) probability $r \in (0, 1)$ and two sets of hyperparameters $(m_1,w_1)\in \mathbb{N} \times (0, \infty)$ and $(m_2, w_2)\in \mathbb{N} \times (0, \infty)$,
		where $(m_1, w_1)$ satisfies $m_1w_1 \leq \mathrm{inj}(\mathsf{M})$ and $(m_2, w_2)$ does not need to obey any further restrictions.
		We now determine for each iteration at random which set of hyperparameters to use:
		Denoting with $K_1$ the geodesic slice sampler that uses the hyperparameters $m_1$ and $w_1$, and with $K_2$ the geodesic slice sampler with hyperparameters $m_2$ and $w_2$,
		the kernel $K_0 := r K_1 + (1-r) K_2$ satisfies $\lim_{n \to \infty} d_{tv}(K_0^n(x, \cdot), \pi) = 0$ for all starting points $x \in \mathsf{W}$.
		A proof of these claims can be found at the end of this section and Lemma \ref{L: absolute continuity of the geodesic slice sampler}.
\end{remark}

\begin{lemma}\label{L: absolute continuity of the geodesic slice sampler}
	Impose Assumptions \ref{Ass: Manifolds} and \ref{Ass: Unnormalized density}, and let $\pi$ be defined as in \eqref{Eq: Definition target density}.
	If the hyperparameters $m \in \mathbb{N}$ and $w \in (0, \infty)$ satisfy $mw \leq  \mathrm{inj}(\mathsf{M})$,
	then for all $x \in \mathsf{W}$  the measure $\widetilde{K}(x, \cdot)$ is absolutely continuous with respect to $\pi$. 
\end{lemma}

\begin{proof}
	Fix $x \in \mathsf{W}$, and let $\mathsf{A} \in \mathcal{B}(\mathsf{W})$ with $\pi(\mathsf{A}) = 0$.
	Then we have $\nu_{\mathfrak{g}}(\mathsf{A}) = 0$, as $p > 0$ on $\mathsf{W}$.
	By Remark \ref{R: Relation Riemannian and Lebesgue mass} this implies $\mathrm{Leb}_{d}(\varphi_x(\mathsf{A}\setminus \mathrm{Cut}(x))) = 0$.
	With a similar change of variables argument as in the proof of Lemma \ref{L: reach a given set}, we get
	\begin{align*}
		&\mathrm{Leb}_{d}\big(\varphi_x(\mathsf{A}\setminus \mathrm{Cut}(x))\big)
		= \int_{\mathbb{S}^{d-1}} \int_{[0, \infty)} \mathbbm{1}_{\varphi_x(\mathsf{A} \setminus \mathrm{Cut}(x))} (\theta v) \theta^{d-1}\ \mathrm{Leb}_{1}(\d \theta)\, \nu_{\widehat{\mathfrak{g}}}(\d v)\\
		&\qquad= \int_{\mathbb{S}_x^{d-1}} \int_{[0, t_{\mathrm{cut}}(x,v))} \mathbbm{1}_{\mathsf{A}}\left(\gamma_{(x,v)}(\theta)\right)\theta^{d-1}\ \mathrm{Leb}_{1}(\d \theta)\, \nu_{\widehat{\mathfrak{g}}_x}(\d v)\\
		&\qquad= \frac{1}{2} \int_{\mathbb{S}_x^{d-1}} \int_{(-t_{\mathrm{cut}}(x,-v), t_{\mathrm{cut}}(x,v))} \mathbbm{1}_{\mathsf{A}}\left(\gamma_{(x,v)}(\theta)\right)\theta^{d-1}\ \mathrm{Leb}_{1}(\d \theta)\, \nu_{\widehat{\mathfrak{g}}_x}(\d v),
	\end{align*}
	where for the last equality we used that $\nu_{\widehat{\mathfrak{g}}_x}$ inherits rotation invariance from $\nu_{\widehat{\mathfrak{g}}}$.
	Define $\widetilde{\mathsf{A}} := \{(\theta, v) \in \mathbb{R}\times \mathbb{S}_x^{d-1} \mid \gamma_{(x,v)}(\theta v) \in \mathsf{A}\}$.
	Since by assumption $mw \leq \mathrm{inj}(\mathsf{M}) \leq t_{\mathrm{cut}}(x,v)$ for all $v \in \mathbb{S}_x^{d-1}$,
	we obtain that $\widetilde{\mathsf{A}}$ is a null set for the product measure $\mathrm{Leb}_{1}\vert_{[-mw, mw]} \otimes \nu_{\widehat{\mathfrak{g}}_x}$.
	Since $Q_{\mathsf{L}(x,v,t)}^{\ell,r}(0, \cdot)$ is absolutely continuous with respect to $\mathrm{Leb}_{1}\vert_{(\ell,r)}$ by virtue of \cite[Proof of Theorem 2.10]{ReversibilityEllipticalSliceSampler}, and $(\ell,r) \subseteq [-mw, mw]$ holds $\xi_{\mathsf{L}(x,v,t)}^{(0)}$-almost surely by definition of the stepping-out distributions,
	this implies $\widetilde{K}(x, \mathsf{A}) = 0$.
\end{proof}

\begin{lemma}\label{L: criterion for Harris recurrence}
	Let $K_1$ and $K_2$ be two Markov kernels on a measurable space $(\mathsf{X}, \mathcal{X})$, and let $r \in (0, 1)$.
	Assume that the kernel $K := r K_1 + (1-r)K_2$ on $(\mathsf{X}, \mathcal{X})$ has an invariant probability measure $\pi$ and is $\pi$-irreducible.
	If $K_1(x, \cdot)$ is absolutely continuous with respect to $\pi$ for all $x \in \mathsf{X}$, then $K$ is Harris recurrent.
\end{lemma}

\begin{proof}
	Let $h : \mathsf{X} \to [0, \infty)$ be a bounded harmonic function of $K$, i.e. $h$ is a bounded measurable function that satisfies $Kh = h$,
	where $Kh$ denotes the function $x \mapsto \int_{\mathsf{X}} h(y)\ K(x, \d y)$.
	We show that $h$ is constant.
	Since $K$ is recurrent (see e.g. \cite[Theorem 10.1.6]{douc2018markov}), this then implies by  \cite[Theorem 2]{tierney1994markov} that $K$ is Harris recurrent.
	
	As $K$ is recurrent, by \cite[Proposition 3.13]{nummelin1984general} the function $h$ is $\pi$-almost surely constant with value $\pi(h)$,
	where $\pi(h)$ denotes the integral of $h$ with respect to $\pi$.
	Define $g: \mathsf{X} \to \mathbb{R}, x \mapsto h(x) - \pi(h)$.
	The function $g$ is $\pi$-almost surely equal to zero and satisfies $Kg = g$.
	Since $K_1(x, \cdot)$ is absolutely continuous with respect to $\pi$, we get 
	\[
	g(x) = Kg(x) = r K_1g(x) + (1-r) K_2g(x) = (1-r) K_2g(x),
	\]
	for all $x \in \mathsf{X}$.
	Induction over $n$, yields $K_2^ng = 1/(1-r)^n g$ for all $n \in \mathbb{N}$.
	Since $g$ is bounded, this implies that $g(x) = 0$ for all $x \in \mathsf{X}$, which is equivalent to $h(x) = \pi(h)$ for all $x \in \mathsf{X}$.
\end{proof}

\begin{proof}[of Remark \ref{R: Harris recurence} and Remark \ref{R: Modification of geodesic slice sampling}]
	If $mw \leq \mathrm{inj}(\mathsf{M})$, then $\widetilde{K}(x, \cdot)$ is absolutely continuous with respect to $\pi$ for all $x \in \mathsf{W}$ by Lemma \ref{L: absolute continuity of the geodesic slice sampler}.
	In conjecture with the reversibility of $\widetilde{K}$ with respect to $\pi$ and Lemma 5, this yields that $\widetilde{K}$ is Harris recurrent, 
	see \cite[Corollary 1]{tierney1994markov}.
	Hence, we get $\lim_{n \to \infty} d_{tv}(\widetilde{K}^n(x,\cdot), \pi) = 0$ for all $x \in \mathsf{W}$, see \cite[Theorem 1]{tierney1994markov}.
	Moreover, combining Lemma \ref{L: absolute continuity of the geodesic slice sampler} and Lemma \ref{L: criterion for Harris recurrence}, 
	we obtain $\lim_{n \to \infty} d_{tv}(K_0^n(x,\cdot), \pi) = 0$ for all $x \in \mathsf{W}$ again by \cite[Theorem 1]{tierney1994markov},
	where $K_0$ is defined as in Remark \ref{R: Modification of geodesic slice sampling}.
\end{proof}

\section{Properties of the stepping-out procedure}\label{Sec: Properties of stepping-out}

In this section we provide the proofs of Lemma \ref{L: Key identity for reversiblity - measurable functions} and Lemma \ref{L: stepping out under rotation}.
To this end, we suppose that we are in the setting of Supplementary material \ref{Sec: Stepping-out} summarised as follows:
\begin{assumption}\label{Set: Stepping-out}
	Fix $w \in (0, \infty)$ and $m \in \mathbb{N}$, and let $\mathsf{S} \in \mathcal{B}(\mathbb{R})$.
	For all $\theta \in \mathbb{R}$ define the sequences $(L_i^{(\theta)})_{i \in \mathbb{N}}$ and $( R_j^{(\theta)})_{j \in \mathbb{N}}$ as in \eqref{Eq: Definition L_i, R_j}, and the stopping times $\tau_{\mathsf{S}}^{(\theta)}$ and $\mathfrak{T}_{\mathsf{S}}^{(\theta)}$ as in \eqref{Eq: Definition stopping times}.
\end{assumption}
In the proof of both lemmas we push a corresponding property of the sequences $(L_i^{(\theta)})_{i \in \mathbb{N}}$ and $( R_j^{(\theta)})_{j \in \mathbb{N}}$ to the stopped sequences.
Observe that it is relatively simple to establish the required properties for $(L_i^{(\theta)})_{i \in \mathbb{N}}$ and $( R_j^{(\theta)})_{j \in \mathbb{N}}$, because their joint distributions are available explicitly.

\begin{lemma}\label{L: joint distribution}
	Suppose Assumption \ref{Set: Stepping-out} is satisfied.
	Let $i,j \in \mathbb{N}$ and $\theta \in \mathbb{R}$.
	We have for all $\mathsf{A}_1, \ldots, \mathsf{A}_i, \mathsf{B}_1, \ldots, \mathsf{B}_j \in \mathcal{B}(\mathbb{R})$ that
	\begin{align*}
		&\mathbb{P}\Big(L_1^{( \theta)}\in \mathsf{A}_1, \ldots, L_i^{( \theta)} \in \mathsf{A}_i, R_1^{(\theta)} \in \mathsf{B}_1, \ldots, R_j^{(\theta)}\in \mathsf{B}_j\Big)\\
		&\qquad= \frac{1}{w} \int_{0}^{w} \prod_{k = 1}^i \mathbbm{1}_{\mathsf{A}_k}\big(\theta-u- (k-1)w\big) \cdot \prod_{l=1}^j \mathbbm{1}_{\mathsf{B}_l}(\theta-u+lw) \ \mathrm{Leb}_1(\d u).
	\end{align*} 
\end{lemma}

\begin{proof}
	Let $i,j \in \mathbb{N}$, $\theta \in \mathbb{R}$ and $\mathsf{A}_1, \ldots, \mathsf{A}_i, \mathsf{B}_1, \ldots, \mathsf{B}_j \in \mathcal{B}(\mathbb{R})$.
	Then, as $\Upsilon \sim \mathrm{Unif}([0,w])$, we have
	\begin{align*}
		&\mathbb{P}\Big(L_1^{( \theta)}\in \mathsf{A}_1, \ldots, L_i^{( \theta)} \in \mathsf{A}_i, R_1^{(\theta)} \in \mathsf{B}_1, \ldots, R_j^{(\theta)}\in \mathsf{B}_j\Big)\\
		&\qquad=\mathbb{P}\Big(\theta-\Upsilon-(k-1)w \in \mathsf{A}_k,\, k \in \{1,\ldots, i\},\ \theta-\Upsilon + lw \in \mathsf{B}_l,\, l \in \{1,\ldots, j\}\Big)\\
		&\qquad= \frac{1}{w} \int_{0}^{w} \prod_{k = 1}^i \mathbbm{1}_{\mathsf{A}_k}(\theta-u- (k-1)w) \cdot \prod_{l=1}^j \mathbbm{1}_{\mathsf{B}_l}(\theta-u+lw) \ \mathrm{Leb}_1(\d u).	
	\end{align*}
\end{proof}

We present partitions of certain events that facilitate the extension of properties of the sequences $(L_i^{(\theta)})_{i \in \mathbb{N}}$ and $( R_j^{(\theta)})_{j \in \mathbb{N}}$ to the stepping-out distribution.

\begin{lemma}\label{L: Partitions}
	Impose Assumption \ref{Set: Stepping-out}. 
	Then for all $\theta, \alpha \in \mathbb{R}$ we have
	\begin{enumerate}
		\item $	\begin{aligned}[t]\left\{ L_i^{(\theta)} < \alpha < R_j^{(\theta)} \right\} 
			&= \bigsqcup_{k = 0}^{i-1} \left\{ L_{k+1}^{(\theta)} < \alpha < L_k^{(\theta)} \right\} \ \sqcup \ \bigsqcup_{l = 1}^{j-1} \left\{ R_l^{(\theta)} < \alpha < R_{l+1}^{(\theta)} \right\} \\
			&\qquad \sqcup\ \bigsqcup_{k=1}^{i-1}\left\{L_k^{(\theta)} = \alpha \right\} \ \sqcup\ \bigsqcup_{l=1}^{i-1}\left\{R_l^{(\theta)} = \alpha \right\} \text{ for all } i,j \in \mathbb{N},\end{aligned}$ 
		\label{L: Representation of Li <= y <= Rj}
		\item $\begin{aligned}[t]\{\boldsymbol{L}_{\mathsf{S}}^{(\theta)} < \alpha < L_0^{(\theta)}\} &= \bigsqcup_{i = 1}^m \{ L_{\tau_{\mathsf{S}}^{(\theta)}-i+1}^{(\theta)} < \alpha < 	 L_{\tau_{\mathsf{S}}^{(\theta)}-i}^{(\theta)}\, ,\ \tau_{\mathsf{S}}^{(\theta)} \geq i\} \\
			&\qquad \sqcup\ \bigsqcup_{i' = 1}^{m-1}\{L_{i'}^{(\theta)} = \alpha,\ \tau_{\mathsf{S}}^{(\theta)}\geq i'+1\},\end{aligned}$
		\label{L: Partition of Ltau <= < <0 L0}
		\item $\begin{aligned}[t] \{R_1^{(\theta)} < \alpha < \boldsymbol{R}_{\mathsf{S}}^{(\theta)}\} &= \bigsqcup_{j = 1}^{m-1} \{ R_{\mathfrak{T}_{\mathsf{S}}^{(\theta)}-j}^{(\theta)} < 	\alpha < 	 R_{\mathfrak{T}_{\mathsf{S}}^{(\theta)}-j+1}^{(\theta)}\,,\ \mathfrak{T}_{\mathsf{S}}^{(\theta)} \geq j+1\}\\
			&\qquad  \sqcup \ \bigsqcup_{j'=2}^{m-1} \left\{R_{j'}^{(\theta)} = \alpha, \mathfrak{T}_{\mathsf{S}}^{(\theta)} \geq j'+1\right\},\end{aligned}$
		\label{L: Partition of R1 < 0< RT}
		\item $\mathbb{P}\left( L_n^{(\theta)} = \alpha \right)	= \mathbb{P}\left( R_n^{(\theta)} = \alpha \right) = 0$ for all $n \in \mathbb{N}$.
		\label{L: Overlap is nullset}
	\end{enumerate}
	
\end{lemma}

\begin{proof}
	\emph{To \ref{L: Representation of Li <= y <= Rj}:}
	Observe that the monotonicity property \eqref{Eq: Monotonicity L, R} of $(L_k^{(\theta)})_{k \in \mathbb{N}}$ and $( R_l^{(\theta)})_{l \in \mathbb{N}}$ implies the statement.

	\emph{To \ref{L: Partition of Ltau <= < <0 L0}:}
	Let $\upomega \in \{\boldsymbol{L}_{\mathsf{S}}^{(\theta)} < \alpha < L_0^{(\theta)}\}$.
	Then there exists 
	\begin{align*}
		\widehat{k} &\in \{0, \ldots, \tau_{\mathsf{S}}^{(\theta)}(\upomega) - 1\} \text{ such that } L_{\widehat{k}+1}^{(\theta)}(\upomega) < \alpha < L_{\widehat{k}}^{(\theta)}(\upomega), \text{ or}\\
		k' &\in \{1,\ldots, \tau_{\mathsf{S}}^{(\theta)}(\upomega) - 1\} \text{ such that } L_{k'}^{(\theta)}(\upomega)= \alpha.
	\end{align*}
	In the first case, choosing $k = \tau_{\mathsf{S}}^{(\theta)}(\upomega) - \widehat{k}$ and taking \eqref{Eq: Bound tau, T} into account, 
	we get $k \in \{1, \ldots, m\}$, $k \leq \tau_{\mathsf{S}}^{(\theta)}(\upomega)$ and $L_{\tau_{\mathsf{S}}^{(\theta)}(\upomega)-k+1}^{(\theta)}(\upomega) < \alpha < 	 L_{\tau_{\mathsf{S}}^{(\theta)}(\upomega)-k}^{(\theta)}(\upomega)$.
	Or, in the second case, $k' \in \{1, \ldots, m-1\}$, $k' + 1 \leq \tau_{\mathsf{S}}^{(\theta)}(\upomega)$ and  $L_{k'}^{(\theta)}(\upomega)= \alpha$.
	Hence 
	\[ 
	\upomega \in \bigsqcup_{i = 1}^m \{ L_{\tau_{\mathsf{S}}^{(\theta)}-i+1}^{(\theta)} < \alpha < 	 L_{\tau_{\mathsf{S}}^{(\theta)}-i}^{(\theta)}\,,\ \tau_{\mathsf{S}}^{(\theta)} \geq i\} \ \sqcup\ \bigsqcup_{i' = 1}^{m-1}\{L_{i'}^{(\theta)} = \alpha, \tau_{\mathsf{S}}^{(\theta)}\geq i'+1\}. 
	\]
	Conversely, let there exists 
	\begin{align*}
		i &\in \{1, \ldots, m\} \text{ such that } \upomega \in  \{ L_{\tau_{\mathsf{S}}^{(\theta)}-i+1}^{(\theta)} < \alpha < 	 L_{\tau_{\mathsf{S}}^{(\theta)}-i}^{(\theta)}\,,\ \tau_{\mathsf{S}}^{(\theta)} \geq i\}, \text{ or}\\
		i' &\in \{1, \ldots, m-1\} \text{ such that } \{L_{i'}^{(\theta)} = \alpha, \tau_{\mathsf{S}}^{(\theta)}\geq i'+1\}.
	\end{align*}
	Then clearly $\upomega \in \{\boldsymbol{L}_{\mathsf{S}}^{(\theta)} < \alpha < L_0^{(\theta)}\}$.
	
	\emph{To \ref{L: Partition of R1 < 0< RT}:}
	We apply similar arguments as for the statement \ref{L: Partition of Ltau <= < <0 L0}.
	Let $\upomega \in \{R_1^{(\theta)} < \alpha < \boldsymbol{R}_{\mathsf{S}}^{(\theta)}\}$.
	Then there exists 
	\begin{align*}
		\widehat{k} &\in \{1, \ldots, \mathfrak{T}_{\mathsf{S}}^{(\theta)}(\upomega)-1\} \text{ such that } R_{\widehat{k}}^{(\theta)}(\upomega) < \alpha < R_{\widehat{k}+1}^{(\theta)}(\upomega), \text{ or}\\
		k'& \in \{2, \ldots, \mathfrak{T}_{\mathsf{S}}^{(\theta)}(\upomega) - 1\} \text{ such that } R_{k'}^{(\theta)}(\upomega) = \alpha.
	\end{align*}
	In the first case, again choosing $k = \mathfrak{T}_{\mathsf{S}}^{(\theta)}(\upomega) - \widehat{k}$ and observing \eqref{Eq: Bound tau, T},
	we get $k \in \{1, \ldots, m-1\}$, $k+1 \leq \mathfrak{T}_{\mathsf{S}}^{(\theta)}(\upomega)$ and $R_{\mathfrak{T}_{\mathsf{S}}^{(\theta)}(\upomega)-k}^{(\theta)}(\upomega) < \alpha < 	 R_{\mathfrak{T}_{\mathsf{S}}^{(\theta)}(\upomega)-k+1}^{(\theta)}(\upomega)$.
	Or, in the second case, we get $k' \in \{2, \ldots, m-1\}$, $k'+1 \leq \mathfrak{T}_{\mathsf{S}}^{(\theta)}(\upomega)$ and $R_{k'}^{(\theta)}(\upomega)= \alpha$.
	Thus 
	\[
	\upomega \in  \bigsqcup_{j = 1}^{m-1} \{ R_{\mathfrak{T}_{\mathsf{S}}^{(\theta)}-j}^{(\theta)} < \alpha < 	 R_{\mathfrak{T}_{\mathsf{S}}^{(\theta)}-j+1}^{(\theta)}\,,\ \mathfrak{T}_{\mathsf{S}}^{(\theta)} \geq j+1\}\ \sqcup \ \bigsqcup_{j'=2}^{m-1} \left\{R_{j'}^{(\theta)} = \alpha, \mathfrak{T}_{\mathsf{S}}^{(\theta)} \geq j'+1\right\}.
	\]
	Conversely, let 
	\begin{align*}
		\upomega& \in  \{ R_{\mathfrak{T}_{\mathsf{S}}^{(\theta)}-j}^{(\theta)} < \alpha < 	 R_{\mathfrak{T}_{\mathsf{S}}^{(\theta)}-j+1}^{(\theta)}, \mathfrak{T}_{\mathsf{S}}^{(\theta)} \geq j+1\} \text{ for some } j \in \{1, \ldots, m-1\}, \text{ or} \\
		\upomega &\in \left\{R_{j'}^{(\theta)} = \alpha, \mathfrak{T}_{\mathsf{S}}^{(\theta)} \geq j'+1\right\} \text{ for some } j' \in \{2, \ldots, m-1\}.
	\end{align*}
	Then we have $\upomega \in \{R_1^{(\theta)} < \alpha < \boldsymbol{R}_{\mathsf{S}}^{(\theta)}\}$, 
	as $\mathfrak{T}_{\mathsf{S}}^{(\theta)}(\upomega) - j \geq j+1-j = 1$ and $\mathfrak{T}_{\mathsf{S}}^{(\theta)}(\upomega) - j + 1 \leq  \mathfrak{T}_{\mathsf{S}}^{(\theta)}(\upomega) - 1 +1 = \mathfrak{T}_{\mathsf{S}}^{(\theta)}(\upomega)$.
	
	\emph{To \ref{L: Overlap is nullset}:}
	Let $n \in \mathbb{N}$.
	We have by definition of $ L_n^{(\theta)} $ that
	\begin{align*}
		\mathbb{P}\left( L_n^{(\theta)} = \alpha \right) 
		&= \frac{1}{w} \int_0^w \mathbbm{1}_{\{\alpha\}} \big(\theta-u - (n-1)w\big)\ \mathrm{Leb}_1(\d u)\\
		&\leq \frac{1}{w}\ \mathrm{Leb}_1\big(\{\theta-\alpha - (n-1)w\} \big)= 0. 
	\end{align*}
	The statement for $R_n^{(\theta)}$ follows analogously.
\end{proof}

The following lemma collects the properties of $(L_i^{(\theta)})_{i \in \mathbb{N}}$ and $( R_j^{(\theta)})_{j \in \mathbb{N}}$ and the stopped sequences which lead to the proof of Lemma \ref{L: Key identity for reversiblity - measurable functions}.

\begin{lemma}\label{L: Tools for switching lemma}
	Assume Assumption \ref{Set: Stepping-out} and let $\theta, \alpha \in \mathbb{R}$.
	\begin{enumerate}
		\item Let $i,j \in \mathbb{N}$ and $\mathsf{A}_1, \ldots, \mathsf{A}_i, \mathsf{B}_1, \ldots, \mathsf{B}_j \in \mathcal{B}(\mathbb{R})$.
		We have for all $q \in \{0, \ldots, i-1\}$ that
		\begin{align*}
			& \mathbb{P} \left( L_1^{(\theta)} \in \mathsf{A}_1, \ldots, L_i^{(\theta)} \in \mathsf{A}_i, R_1^{(\theta)} \in \mathsf{B}_1, \ldots, R_j^{(\theta)} \in \mathsf{B}_j, L_{q+1}^{(\theta)} < \alpha < L_{q}^{(\theta)} \right)\\
			& \quad = \mathbb{P} \left( R_q^{(\alpha)} \in \mathsf{A}_1, \ldots, R_1^{(\alpha)} \in \mathsf{A}_q, L_1^{(\alpha)} \in \mathsf{A}_{q+1}, \ldots, L_{i-q}^{(\alpha)} \in \mathsf{A}_i,\right. \\
			&\quad \qquad  \left. R_{q+1}^{(\alpha)} \in \mathsf{B}_1, \ldots, R_{j+q}^{(\alpha)} \in \mathsf{B}_j, R_q^{(\alpha)} < \theta < R_{q+1}^{(\alpha)} \right).
		\end{align*} 
		Observe that for $ q= 0 $ we use convention \eqref{Eq: Convention for L_0,R_0}.
		\label{L: Relation Y^x to Y^y}
		\item Let $i, j \in \mathbb{N}$.
		For all $\mathsf{A}, \mathsf{B} \in \mathcal{B}(\mathbb{R})$  and $q \in \{0, \ldots, i-1\}$ we have
		\begin{align*}
			& \mathbb{P}\left( \boldsymbol{L}_{\mathsf{S}}^{(\theta)} \in \mathsf{A},  \boldsymbol{R}_{\mathsf{S}}^{(\theta)}  \in \mathsf{B} , L_{q+1}^{(\theta)} < \alpha < L_q^{(\theta)}, 
			\tau_{\mathsf{S}}^{(\theta)} = i, \mathfrak{T}_{\mathsf{S}}^{(\theta)} = j \right)\\
			& \qquad = \mathbb{P}\left( \boldsymbol{L}_{\mathsf{S}}^{(\alpha)} \in \mathsf{A},  \boldsymbol{R}_{\mathsf{S}}^{(\alpha)}  \in \mathsf{B} , R_q^{(\alpha)} < \theta < R_{q+1}^{(\alpha)},\tau_{\mathsf{S}}^{(\alpha)} = i-q, \mathfrak{T}_{\mathsf{S}}^{(\alpha)} = j+q \right).
		\end{align*}
		\label{L: Relation stopped chains}
		\item For all $\mathsf{A},\mathsf{B} \in \mathcal{B}(\mathbb{R})$ holds
		\begin{align*}
			&\mathbb{P}\left( \boldsymbol{L}_{\mathsf{S}}^{(\theta)} \in \mathsf{A},  \boldsymbol{R}_{\mathsf{S}}^{(\theta)}  \in \mathsf{B} , \boldsymbol{L}_{\mathsf{S}}^{(\theta)} < \alpha < \boldsymbol{R}_{\mathsf{S}}^{(\theta)} \right)= \mathbb{P}\left( \boldsymbol{L}_{\mathsf{S}}^{(\alpha)} \in \mathsf{A},  \boldsymbol{R}_{\mathsf{S}}^{(\alpha)}  \in \mathsf{B} , \boldsymbol{L}_{\mathsf{S}}^{(\alpha)}< \theta < \boldsymbol{R}_{\mathsf{S}}^{(\alpha)} \right).
		\end{align*}
		\label{L: Key identity for indicator functions}
	\end{enumerate}
\end{lemma}

\begin{proof}
	\emph{To \ref{L: Relation Y^x to Y^y}:}
	Let $i,j \in \mathbb{N}$, $q \in \{0, \ldots, i-1\}$  and $\mathsf{A}_1, \ldots, \mathsf{A}_i, \mathsf{B}_1, \ldots, \mathsf{B}_j \in \mathcal{B}(\mathbb{R})$.
	By Lemma \ref{L: joint distribution} we have
	\begin{align*}
		& \mathbb{P} \left( L_1^{(\theta)} \in \mathsf{A}_1, \ldots, L_i^{(\theta)} \in \mathsf{A}_i, R_1^{(\theta)} \in \mathsf{B}_1, \ldots, R_j^{(\theta)} \in \mathsf{B}_j, L_{q+1}^{(\theta)} < \alpha < L_{q}^{(\theta)} \right)\\
		&\qquad = \frac{1}{w} \int_0^w \prod_{k = 1}^i \mathbbm{1}_{\mathsf{A}_k}\big(\theta-u - (k-1)w\big)\cdot \prod_{l=1}^j \mathbbm{1}_{\mathsf{B}_l}\big(\theta-u+ l w\big) \\
		&\qquad\qquad\qquad\qquad\cdot \mathbbm{1}_{(-\infty, \alpha)}(\theta-u - qw) \mathbbm{1}_{(\alpha, \infty)}\big(\theta-u - (q-1)w\big)\ \mathrm{Leb}_1(\d u).
	\end{align*}
	Substituting $\widetilde{u} = u - \theta + \alpha + q w$, we obtain
	\begin{align*}
		& \mathbb{P} \left( L_1^{(\theta)} \in \mathsf{A}_1, \ldots, L_i^{(\theta)} \in \mathsf{A}_i, R_1^{(\theta)} \in \mathsf{B}_1, \ldots, R_j^{(\theta)} \in \mathsf{B}_j, L_{q+1}^{(\theta)} < \alpha < L_{q}^{(\theta)} \right)\\
		&\qquad = \frac{1}{w} \int_{\mathbb{R}}\mathbbm{1}_{(0,w)}(\widetilde{u}+\theta-\alpha-qw) \prod_{k = 1}^i \mathbbm{1}_{\mathsf{A}_k}\big(\alpha-\widetilde{u} - (k-q-1)w\big)\\
		&\qquad\qquad \cdot \prod_{l=1}^j \mathbbm{1}_{\mathsf{B}_l}\big(\alpha-\widetilde{u}+(l+q) w\big) 
		\cdot \mathbbm{1}_{(-\infty, \alpha)}(\alpha-\widetilde{u}) \mathbbm{1}_{(\alpha, \infty)}\big(\alpha-\widetilde{u}+w\big)\ \mathrm{Leb}_1(\d \widetilde{u})\\
		&\qquad = \frac{1}{w} \int_{\mathbb{R}}\mathbbm{1}_{(-\infty, \theta)}(\alpha-\widetilde{u}+qw) \mathbbm{1}_{(\theta, \infty)}\big(\alpha-\widetilde{u}+(q+1)w\big) \prod_{k' = 1}^{i-q}
		\mathbbm{1}_{\mathsf{A}_{k'+q}}\big(\alpha-\widetilde{u} - (k'-1)w\big)\\
		&\qquad\qquad\prod_{l'=1}^{q} \mathbbm{1}_{\mathsf{A}_{q+1-l'}}\big(\alpha-\widetilde{u}+l' w\big)  \cdot \prod_{l'=q+1}^{j+q} \mathbbm{1}_{\mathsf{B}_{l'-q}}\big(\alpha-\widetilde{u}+l' w\big) 
		\cdot \mathbbm{1}_{(0,w)}(\widetilde{u})\ \mathrm{Leb}_1(\d \widetilde{u}).
	\end{align*}
	Applying again Lemma \ref{L: joint distribution}, we get
	\begin{align*}
		& \mathbb{P} \left( L_1^{(\theta)} \in \mathsf{A}_1, \ldots, L_i^{(\theta)} \in \mathsf{A}_i, R_1^{(\theta)} \in \mathsf{B}_1, \ldots, R_j^{(\theta)} \in \mathsf{B}_j, L_{q+1}^{(\theta)} < \alpha < L_{q}^{(\theta)} \right)\\
		&\qquad = \mathbb{P} \left( R_q^{(\alpha)} \in \mathsf{A}_1, \ldots, R_1^{(\alpha)} \in \mathsf{A}_q, L_1^{(\alpha)} \in \mathsf{A}_{q+1}, \ldots, L_{i-q}^{(\alpha)} \in \mathsf{A}_i,\right. \\
		&\qquad \qquad\qquad  \left. R_{q+1}^{(\alpha)} \in \mathsf{B}_1, \ldots, R_{j+q}^{(\alpha)} \in \mathsf{B}_j, R_q^{(\alpha)} < \theta < R_{q+1}^{(\alpha)} \right).
	\end{align*}

	\emph{To \ref{L: Relation stopped chains}:}
	Let $i, j \in \mathbb{N}$, $q \in\{0, \ldots, i-1\}$ and $\mathsf{A},\mathsf{B} \in \mathcal{B}(\mathbb{R})$.
	If $ i+j > m+1 $, the statement is true by \eqref{Eq: Bound tau, T}, which implies that both probabilities are zero.
	We now consider $ i+j \leq m +1$.
	To this end, for all $n, \widetilde{n} \in \mathbb{N}$ with $n +\widetilde{n} \leq m+1$ and $\upiota \in \{n, \ldots, m+1-\widetilde{n}\}$ set
	\begin{align*}
		\mathsf{A}_k^{\upiota, n,\widetilde{n}} =  \mathsf{S},\quad k \in \{1, \ldots, n-1\} \qquad &\text{and} \qquad \mathsf{A}_n^{\upiota, n, \widetilde{n}} = \begin{dcases} \mathsf{A}, & \upiota = n\\ \mathsf{A} \cap (\mathbb{R} \setminus \mathsf{S}), & \upiota > n,\end{dcases}\\
		\mathsf{B}_l^{\upiota,n, \widetilde{n}} =  \mathsf{S},\quad  l \in \{1, \ldots, \widetilde{n}-1\} \qquad &\text{and} \qquad 
		\mathsf{B}_{\widetilde{n}}^{\upiota,  n, \widetilde{n}} =  \begin{dcases} \mathsf{B}, & \upiota = m+1-\widetilde{n}\\ \mathsf{B} \cap( \mathbb{R} \setminus \mathsf{S}), & \upiota < m+1-\widetilde{n}.\end{dcases}
	\end{align*}
	Then for $q \in \{0, \ldots , n-1\}$ holds
	\begin{equation}\label{Eq: Identity for sets in switching lemma for stepping-out}
		\begin{split}
			\mathsf{A}_{k}^{\upiota - q, n - q,\widetilde{n} + q} &= \mathsf{A}_{k+q}^{\upiota, n,\widetilde{n}}, \qquad \forall \ k \in \{1,\ldots, n-q\},\\
			\mathsf{B}_l^{\upiota - q, n-q, \widetilde{n}+q} & = \begin{dcases}
				\mathsf{A}_{q-l+1}^{\upiota, n,\widetilde{n}}, & l \leq q,\\
				\mathsf{B}_{l-q}^{\upiota, n,\widetilde{n}}, & l > q,\qquad \forall\ l \in \{1, \ldots, \widetilde{n}+q\}.
			\end{dcases} 
		\end{split}
	\end{equation}
	We can use these sets to express events of the form $ \{\boldsymbol{L}_{\mathsf{S}}^{(\theta)} \in \mathsf{A},  \boldsymbol{R}_{\mathsf{S}}^{(\theta)}  \in \mathsf{B} , \tau_{\mathsf{S}}^{(\theta)} = n, \mathfrak{T}_{\mathsf{S}}^{(\theta)} = \widetilde{n}\}$.
	Namely, observe that by \eqref{Eq: Bound tau, T} we have $\tau_{\mathsf{S}}^{(\theta)} \leq J \leq m+1-\mathfrak{T}_{\mathsf{S}}^{(\theta)} $.
	Together with the independence of $J$ from $(L_k^{(\theta)})_{k \in \mathbb{N}}$ and $(R_l^{(\theta)})_{l \in \mathbb{N}}$, we get
	\begin{equation}\label{Eq: Partition of stopping event}
		\begin{split}
			& \mathbb{P}\left( \boldsymbol{L}_{\mathsf{S}}^{(\theta)} \in \mathsf{A},  \boldsymbol{R}_{\mathsf{S}}^{(\theta)}  \in \mathsf{B} , 
			\tau_{\mathsf{S}}^{(\theta)} = n, \mathfrak{T}_{\mathsf{S}}^{(\theta)} = \widetilde{n}\right)\\
			&\quad = \sum_{\upiota = n}^{m+1-\widetilde{n}} \mathbb{P}\left( \boldsymbol{L}_{\mathsf{S}}^{(\theta)} \in \mathsf{A},  \boldsymbol{R}_{\mathsf{S}}^{(\theta)}  \in \mathsf{B} ,\tau_{\mathsf{S}}^{(\theta)} = n, \mathfrak{T}_{\mathsf{S}}^{(\theta)} = \widetilde{n}, J = \upiota \right)\\
			& \quad = \sum_{\upiota = n}^{m+1-\widetilde{n}} \mathbb{P}\left( L_1^{(\theta)}\in \mathsf{A}_1^{\upiota,n, \widetilde{n}}, \ldots,  L_n^{(\theta)}\in \mathsf{A}_n^{\upiota,n, \widetilde{n}}, R_1^{(\theta)} \in \mathsf{B}_1^{\upiota,n, \widetilde{n}}, \ldots, R_{\widetilde{n}}^{(\theta)} \in \mathsf{B}_{\widetilde{n}}^{\upiota,n, \widetilde{n}}, J = \upiota\right)\\
			& \quad = \frac{1}{m} \sum_{\upiota = n}^{m+1-\widetilde{n}} \mathbb{P}\left( L_1^{(\theta)}\in \mathsf{A}_1^{\upiota,n, \widetilde{n}}, \ldots, L_n^{(\theta)}\in \mathsf{A}_n^{\upiota,n, \widetilde{n}},
			R_1^{(\theta)} \in \mathsf{B}_1^{\upiota,n, \widetilde{n}}, \ldots, R_{\widetilde{n}}^{(\theta)} \in \mathsf{B}_{\widetilde{n}}^{\upiota,n, \widetilde{n}}\right).
		\end{split}
	\end{equation}
	If we use this for $n = i$ and $\widetilde{n} = j$, we get
	\begin{align*}
		& \mathbb{P}\left( \boldsymbol{L}_{\mathsf{S}}^{(\theta)} \in \mathsf{A},  \boldsymbol{R}_{\mathsf{S}}^{(\theta)}  \in \mathsf{B} , L_{q+1}^{(\theta)} < \alpha < L_q^{(\theta)}, 
		\tau_{\mathsf{S}}^{(\theta)} = i, \mathfrak{T}_{\mathsf{S}}^{(\theta)} = j\right)\\
		& \quad = \frac{1}{m} \sum_{\upiota = i}^{m+1-j} \mathbb{P}\left( L_1^{(\theta)}\in \mathsf{A}_1^{\upiota,i,j}, \ldots, L_i^{(\theta)}\in \mathsf{A}_i^{\upiota,i,j},
		R_1^{(\theta)} \in \mathsf{B}_1^{\upiota,i,j}, \ldots, R_j^{(\theta)} \in \mathsf{B}_j^{\upiota,i,j}, L_{q+1}^{(\theta)} < \alpha < L_q^{(\theta)}\right).
	\end{align*}
	Applying statement \ref{L: Relation Y^x to Y^y} and then \eqref{Eq: Identity for sets in switching lemma for stepping-out} yields
	\begin{align*}
		& \mathbb{P}\left( \boldsymbol{L}_{\mathsf{S}}^{(\theta)} \in \mathsf{A},  \boldsymbol{R}_{\mathsf{S}}^{(\theta)}  \in \mathsf{B} , L_{q+1}^{(\theta)} < \alpha < L_q^{(\theta)}, 
		\tau_{\mathsf{S}}^{(\theta)} = i, \mathfrak{T}_{\mathsf{S}}^{(\theta)} = j\right)\\
		& \quad = \frac{1}{m} \sum_{\upiota = i}^{m+1-j} \mathbb{P} \left( R_q^{(\alpha)} \in \mathsf{A}_1^{\upiota,i,j}, \ldots, R_1^{(\alpha)} \in \mathsf{A}_q^{\upiota,i,j}, L_1^{(\alpha)} \in \mathsf{A}_{q+1}^{\upiota,i,j}, \ldots, L_{i-q}^{(\alpha)} \in \mathsf{A}_i^{\upiota,i,j},\right.\\
		&\quad\qquad\qquad\qquad\qquad\left. R_{q+1}^{(\alpha)} \in \mathsf{B}_1^{\upiota,i,j}, \ldots, R_{j+q}^{(\alpha)} \in \mathsf{B}_j^{\upiota,i,j}, R_q^{(\alpha)} < \theta < R_{q+1}^{(\alpha)} \right)\\
		& \quad = \frac{1}{m} \sum_{\upiota = i}^{m+1-j} \mathbb{P} \left( L_1^{(\alpha)} \in \mathsf{A}_{1}^{\upiota-q,i-q,j+q}, \ldots, L_{i-q}^{(\alpha)} \in \mathsf{A}_{i-q}^{\upiota-q,i-q,j+q},\right. \\
		& \quad \qquad\qquad \qquad \qquad \left. R^{(\alpha)}_1 \in \mathsf{B}_1^{\upiota-q,i-q,j+q},\ldots, R_{j+q}^{(\alpha)} \in \mathsf{B}_j^{\upiota-q,i-q,j+q}, R_q^{(\alpha)} < \theta < R_{q+1}^{(\alpha)} \right).
	\end{align*}
	Then doing an index shift and using \eqref{Eq: Partition of stopping event} for $n = i-q$ and $\widetilde{n}= j+q$, we obtain
	\begin{align*}
		& \mathbb{P}\left( \boldsymbol{L}_{\mathsf{S}}^{(\theta)} \in \mathsf{A},  \boldsymbol{R}_{\mathsf{S}}^{(\theta)}  \in \mathsf{B} , L_{q+1}^{(\theta)} < \alpha < L_q^{(\theta)}, 
		\tau_{\mathsf{S}}^{(\theta)} = i, \mathfrak{T}_{\mathsf{S}}^{(\theta)} = j\right)\\
		& \quad = \frac{1}{m} \sum_{\upiota' = i -q }^{m+1-(j+q)} \mathbb{P} \left( L_1^{(\alpha)} \in \mathsf{A}_{1}^{\upiota',i-q,j+q}, \ldots, L_{i-q}^{(\alpha)} \in \mathsf{A}_{i-q}^{\upiota',i-q,j+q},\right. \\
		& \quad \qquad\qquad \qquad \qquad \left. R^{(\alpha)}_1 \in \mathsf{B}_1^{\upiota',i-q,j+q},\ldots, R_{j+q}^{(\alpha)} \in \mathsf{B}_j^{\upiota',i-q,j+q}, R_q^{(\alpha)} < \theta < R_{q+1}^{(\alpha)} \right)\\
		& \quad = \mathbb{P}\left( \boldsymbol{L}_{\mathsf{S}}^{(\alpha)} \in \mathsf{A},  \boldsymbol{R}_{\mathsf{S}}^{(\alpha)}  \in \mathsf{B} , R_q^{(\alpha)} < \theta < R_{q+1}^{(\alpha)},\tau_{\mathsf{S}}^{(\alpha)} = i-q, \mathfrak{T}_{\mathsf{S}}^{(\alpha)} = j+q \right).
	\end{align*}

	\emph{To \ref{L: Key identity for indicator functions}:}
	Let $ \mathsf{A},\mathsf{B} \in \mathcal{B}(\mathbb{R}) $.
	Using \eqref{Eq: Bound tau, T}, we obtain
	\begin{align*}
		&\mathbb{P}\left( \boldsymbol{L}_{\mathsf{S}}^{(\theta)} \in \mathsf{A},  \boldsymbol{R}_{\mathsf{S}}^{(\theta)} \in \mathsf{B} , \boldsymbol{L}_{\mathsf{S}}^{(\theta)} < \alpha < \boldsymbol{R}_{\mathsf{S}}^{(\theta)} \right)\\
		&\qquad =\sum_{i = 1}^m \sum_{j = 1}^{m+1-i} \mathbb{P}\left( L_i^{(\theta)} \in \mathsf{A},  R_j^{(\theta)}  \in \mathsf{B} , L_i^{(\theta)} < \alpha < R_j^{(\theta)}, 
		\tau_{\mathsf{S}}^{(\theta)} = i, \mathfrak{T}_{\mathsf{S}}^{(\theta)} = j \right).
	\end{align*}
	By virtue of Lemma \ref{L: Partitions}.\ref{L: Representation of Li <= y <= Rj} and  \ref{L: Partitions}.\ref{L: Overlap is nullset} this yields
	\begin{align*}
		&\mathbb{P}\left( \boldsymbol{L}_{\mathsf{S}}^{(\theta)} \in \mathsf{A},  \boldsymbol{R}_{\mathsf{S}}^{(\theta)}  \in \mathsf{B} , \boldsymbol{L}_{\mathsf{S}}^{(\theta)} < \alpha < \boldsymbol{R}_{\mathsf{S}}^{(\theta)} \right)\\
		&\quad =  \sum_{i = 1}^m \sum_{j = 1}^{m+1-i}\left[ \sum_{q = 0}^{i-1} \mathbb{P}\left(\boldsymbol{L}_{\mathsf{S}}^{(\theta)} \in \mathsf{A},  
		\boldsymbol{R}_{\mathsf{S}}^{(\theta)} \in \mathsf{B} , L_{q+1}^{(\theta)} < \alpha < L_q^{(\theta)}, 
		\tau_{\mathsf{S}}^{(\theta)} = i, \mathfrak{T}_{\mathsf{S}}^{(\theta)} = j \right) \right. \\
		&\qquad \qquad \quad + 	\left.  \sum_{q = 1}^{j-1} \mathbb{P}\left( \boldsymbol{L}_{\mathsf{S}}^{(\theta)} \in \mathsf{A},  \boldsymbol{R}_{\mathsf{S}}^{(\theta)}  \in \mathsf{B} , R_q^{(\theta)} < \alpha < R_{q+1}^{(\theta)}, 
		\tau_{\mathsf{S}}^{(\theta)} = i, \mathfrak{T}_{\mathsf{S}}^{(\theta)} = j \right) \right].
	\end{align*}
	We apply statement \ref{L: Relation stopped chains} in the first sum straight forwardly.
	In the second sum we apply  statement \ref{L: Relation stopped chains} with $i+q$ and $j-q$ and reversed roles of $\theta$ and $\alpha$.
	Relocating the summand for $q = 0$ from the first to the second sum, we obtain
	\begin{align*}
		&\mathbb{P}\left( \boldsymbol{L}_{\mathsf{S}}^{(\theta)} \in \mathsf{A},  \boldsymbol{R}_{\mathsf{S}}^{(\theta)}  \in \mathsf{B} , \boldsymbol{L}_{\mathsf{S}}^{(\theta)} < \alpha < \boldsymbol{R}_{\mathsf{S}}^{(\theta)} \right)\\
		&\quad =  \sum_{i = 1}^m \sum_{j = 1}^{m+1-i}\left[ \sum_{q = 1}^{i-1} \mathbb{P}\left(\boldsymbol{L}_{\mathsf{S}}^{(\alpha)} \in \mathsf{A},  
		\boldsymbol{R}_{\mathsf{S}}^{(\alpha)} \in \mathsf{B} , R_q^{(\alpha)} < \theta < R_{q+1}^{(\alpha)}, 
		\tau_{\mathsf{S}}^{(\alpha)} = i-q, \mathfrak{T}_{\mathsf{S}}^{(\alpha)} = j+q \right) \right. \\
		&\qquad \qquad \quad + 	\left.  \sum_{q = 0}^{j-1} \mathbb{P}\left( \boldsymbol{L}_{\mathsf{S}}^{(\alpha)} \in \mathsf{A},  \boldsymbol{R}_{\mathsf{S}}^{(\alpha)}  \in \mathsf{B} , L_{q+1}^{(\alpha)} < \theta < L_q^{(\alpha)}, 
		\tau_{\mathsf{S}}^{(\alpha)} = i+q, \mathfrak{T}_{\mathsf{S}}^{(\alpha)} = j-q \right) \right]\\
		&\quad =  \sum_{i = 1}^m \sum_{j = 1}^{m+1-i}\left[ \sum_{q = 1}^{i-1} \mathbb{P}\left(\boldsymbol{L}_{\mathsf{S}}^{(\alpha)} \in \mathsf{A},  
		\boldsymbol{R}_{\mathsf{S}}^{(\alpha)} \in \mathsf{B} , R_{\mathfrak{T}_{\mathsf{S}}^{(\alpha)} - j}^{(\alpha)} < \theta < R_{\mathfrak{T}_{\mathsf{S}}^{(\alpha)} - j + 1}^{(\alpha)}, \mathfrak{T}_{\mathsf{S}}^{(\alpha)} = j+q , \right. \right.\\
		&\quad \hspace{10cm} \tau_{\mathsf{S}}^{(\alpha)} + \mathfrak{T}_{\mathsf{S}}^{(\alpha)} = i+j\bigg)  \\
		&\quad \qquad \qquad  + 	\left.  \sum_{q = 0}^{j-1} \mathbb{P}\left( \boldsymbol{L}_{\mathsf{S}}^{(\alpha)} \in \mathsf{A},  \boldsymbol{R}_{\mathsf{S}}^{(\alpha)}  \in \mathsf{B} , 
		L_{\tau_{\mathsf{S}}^{(\alpha)} -i+1}^{(\alpha)} < \theta < L_{\tau_{\mathsf{S}}^{(\alpha)} -i}^{(\alpha)}, 
		\tau_{\mathsf{S}}^{(\alpha)} = i+q,\right. \right . \\
		&  \quad \hspace{10cm} \tau_{\mathsf{S}}^{(\alpha)} + \mathfrak{T}_{\mathsf{S}}^{(\alpha)} = i+j \bigg) \Bigg].
	\end{align*}
	Then the lower bounds in \eqref{Eq: Bound tau, T} imply
	\begin{align*}
		&\mathbb{P}\left( \boldsymbol{L}_{\mathsf{S}}^{(\theta)} \in \mathsf{A},  \boldsymbol{R}_{\mathsf{S}}^{(\theta)}  \in \mathsf{B} , \boldsymbol{L}_{\mathsf{S}}^{(\theta)} < \alpha < \boldsymbol{R}_{\mathsf{S}}^{(\theta)} \right)\\
		&\quad =  \sum_{i = 1}^m \sum_{j = 1}^{m+1-i}\left[  \mathbb{P}\left(\boldsymbol{L}_{\mathsf{S}}^{(\alpha)} \in \mathsf{A},  
		\boldsymbol{R}_{\mathsf{S}}^{(\alpha)} \in \mathsf{B} , R_{\mathfrak{T}_{\mathsf{S}}^{(\alpha)} - j}^{(\alpha)} < \theta < R_{\mathfrak{T}_{\mathsf{S}}^{(\alpha)} - j + 1}^{(\alpha)}, 
		\mathfrak{T}_{\mathsf{S}}^{(\alpha)} \geq j+1\right.\right. ,\\
		&\quad \hspace{10cm} \tau_{\mathsf{S}}^{(\alpha)} + \mathfrak{T}_{\mathsf{S}}^{(\alpha)} = i+j\bigg)  \\
		&\qquad \qquad \quad + 	\left.  \mathbb{P}\left( \boldsymbol{L}_{\mathsf{S}}^{(\alpha)} \in \mathsf{A},  \boldsymbol{R}_{\mathsf{S}}^{(\alpha)}  \in \mathsf{B} , 
		L_{\tau_{\mathsf{S}}^{(\alpha)} -i+1}^{(\alpha)} < \theta < L_{\tau_{\mathsf{S}}^{(\alpha)} -i}^{(\alpha)}, 
		\tau_{\mathsf{S}}^{(\alpha)} \geq i, \tau_{\mathsf{S}}^{(\alpha)} + \mathfrak{T}_{\mathsf{S}}^{(\alpha)} = i+j \right) \right].
	\end{align*}
	Changing the order of summation, we obtain
	\begin{align*}
		&\mathbb{P}\left( \boldsymbol{L}_{\mathsf{S}}^{(\theta)} \in \mathsf{A},  \boldsymbol{R}_{\mathsf{S}}^{(\theta)}  \in \mathsf{B} , \boldsymbol{L}_{\mathsf{S}}^{(\theta)} < \alpha < \boldsymbol{R}_{\mathsf{S}}^{(\theta)} \right)\\
		&\quad =  \sum_{j = 1}^m \sum_{i = 1}^{m+1-j}  \mathbb{P}\left(\boldsymbol{L}_{\mathsf{S}}^{(\alpha)} \in \mathsf{A},  
		\boldsymbol{R}_{\mathsf{S}}^{(\alpha)} \in \mathsf{B} , R_{\mathfrak{T}_{\mathsf{S}}^{(\alpha)} - j}^{(\alpha)} < \theta < R_{\mathfrak{T}_{\mathsf{S}}^{(\alpha)} - j + 1}^{(\alpha)}, 
		\mathfrak{T}_{\mathsf{S}}^{(\alpha)} \geq j+1 ,\right.\\
		&\quad \hspace{10cm} \tau_{\mathsf{S}}^{(\alpha)} + \mathfrak{T}_{\mathsf{S}}^{(\alpha)} = i+j\bigg)  \\
		&\quad \qquad  + 	\sum_{i = 1}^m \sum_{j = 1}^{m+1-i}  \mathbb{P}\left( \boldsymbol{L}_{\mathsf{S}}^{(\alpha)} \in \mathsf{A},  \boldsymbol{R}_{\mathsf{S}}^{(\alpha)}  \in \mathsf{B} , 
		L_{\tau_{\mathsf{S}}^{(\alpha)} -i+1}^{(\alpha)} < \theta < L_{\tau_{\mathsf{S}}^{(\alpha)} -i}^{(\alpha)}, 
		\tau_{\mathsf{S}}^{(\alpha)} \geq i,\right.\\
		&\quad \hspace{10cm} \tau_{\mathsf{S}}^{(\alpha)} + \mathfrak{T}_{\mathsf{S}}^{(\alpha)} = i+j \bigg) .
	\end{align*}
	Together with the upper bound in \eqref{Eq: Bound tau, T} this yields
	\begin{align*}
		&\mathbb{P}\left( \boldsymbol{L}_{\mathsf{S}}^{(\theta)} \in \mathsf{A},  \boldsymbol{R}_{\mathsf{S}}^{(\theta)}  \in \mathsf{B} , \boldsymbol{L}_{\mathsf{S}}^{(\theta)} < \alpha < \boldsymbol{R}_{\mathsf{S}}^{(\theta)} \right)\\
		&\quad =  \sum_{j = 1}^m   \mathbb{P}\left(\boldsymbol{L}_{\mathsf{S}}^{(\alpha)} \in \mathsf{A},  
		\boldsymbol{R}_{\mathsf{S}}^{(\alpha)} \in \mathsf{B} , R_{\mathfrak{T}_{\mathsf{S}}^{(\alpha)} - j}^{(\alpha)} < \theta < R_{\mathfrak{T}_{\mathsf{S}}^{(\alpha)} - j + 1}^{(\alpha)}, 
		\mathfrak{T}_{\mathsf{S}}^{(\alpha)} \geq j+1 \right)  \\
		&\quad \qquad  + 	\sum_{i = 1}^m  \mathbb{P}\left( \boldsymbol{L}_{\mathsf{S}}^{(\alpha)} \in \mathsf{A},  \boldsymbol{R}_{\mathsf{S}}^{(\alpha)}  \in \mathsf{B} , 
		L_{\tau_{\mathsf{S}}^{(\alpha)} -i+1}^{(\alpha)} < \theta < L_{\tau_{\mathsf{S}}^{(\alpha)} -i}^{(\alpha)}, 
		\tau_{\mathsf{S}}^{(\alpha)} \geq i\right) .
	\end{align*}
	By virtue of Lemma \ref{L: Partitions}.\ref{L: Partition of Ltau <= < <0 L0}  to \ref{L: Partitions}.\ref{L: Overlap is nullset} we get
	\begin{align*}
		&\mathbb{P}\left( \boldsymbol{L}_{\mathsf{S}}^{(\theta)} \in \mathsf{A},  \boldsymbol{R}_{\mathsf{S}}^{(\theta)}  \in \mathsf{B} , \boldsymbol{L}_{\mathsf{S}}^{(\theta)} < \alpha < \boldsymbol{R}_{\mathsf{S}}^{(\theta)} \right)\\
		&\qquad =  \mathbb{P}\left(\boldsymbol{L}_{\mathsf{S}}^{(\alpha)} \in \mathsf{A},  \boldsymbol{R}_{\mathsf{S}}^{(\alpha)} \in \mathsf{B} , R_1^{(\alpha)} < \theta < \boldsymbol{R}_{\mathsf{S}}^{(\alpha)}\right)  
		+\mathbb{P}\left( \boldsymbol{L}_{\mathsf{S}}^{(\alpha)} \in \mathsf{A},  \boldsymbol{R}_{\mathsf{S}}^{(\alpha)}  \in \mathsf{B} ,	\boldsymbol{L}_{\mathsf{S}}^{(\alpha)} < \theta < L_0^{(\alpha)}\right) \\
		&\qquad = \mathbb{P}\left(\boldsymbol{L}_{\mathsf{S}}^{(\alpha)} \in \mathsf{A},  \boldsymbol{R}_{\mathsf{S}}^{(\alpha)} \in \mathsf{B} , \boldsymbol{L}_{\mathsf{S}}^{(\alpha)}  < \theta < \boldsymbol{R}_{\mathsf{S}}^{(\alpha)}\right).
	\end{align*}
\end{proof}

Lemma \ref{L: Key identity for reversiblity - measurable functions} is now an easy corollary of the previous lemma.

\begin{proof}[of Lemma \ref{L: Key identity for reversiblity - measurable functions}]
	Lemma \ref{L: Tools for switching lemma}.\ref{L: Key identity for indicator functions} yields that the finite measures
	\[  
	\int_{\mathbb{R}^2} \mathbbm{1}_{\mathsf{C}}(\ell,r) \mathbbm{1}_{(\ell, r)}(\alpha) \ \xi_{\mathsf{S}}^{(\theta)}\big(\d (\ell,r) \big), \qquad \mathsf{C} \in \mathcal{B}(\mathbb{R}^2),
	\]
	and
	\[  
	\int_{\mathbb{R}^2}\mathbbm{1}_{\mathsf{C}}(\ell,r) \mathbbm{1}_{(\ell,r) }(\theta) \ \xi_{\mathsf{S}}^{(\alpha)}\big(\d (\ell,r) \big), \qquad \mathsf{C} \in \mathcal{B}(\mathbb{R}^2),
	\]
	agree on the intersection stable generator $\{\mathsf{A} \times \mathsf{B} \mid \mathsf{A}, \mathsf{B} \in \mathcal{B}(\mathbb{R}) \}$ of $\mathcal{B}(\mathbb{R}^2)$.
	By \citep[Lemma 1.42]{Klenke}, this implies that these two measures are equal.
\end{proof}

Before addressing Lemma \ref{L: stepping out under rotation}, we take a brief detour, because some of the notation of the proof of Lemma \ref{L: Tools for switching lemma} is convenient to show measureability of the stepping-out distribution in the argument $\theta$.
\begin{remark}\label{R: Stepping-out distribution is a kernel}
	Impose Assumption \ref{Set: Stepping-out}. 
	Let $A, B \in \mathcal{B}(\mathbb{R})$, 
	and for all $n, \widetilde{n} \in \mathbb{N}$ with $n +\widetilde{n} \leq m+1$ and $\upiota \in \{n, \ldots, m+1-\widetilde{n}\}$ 
	define $\mathsf{A}_1^{\upiota, n,\widetilde{n}}, \ldots,\mathsf{A}_n^{\upiota, n,\widetilde{n}}$, $\mathsf{B}_1^{\upiota, n,\widetilde{n}}, \ldots,\mathsf{B}_{\widetilde{n}}^{\upiota, n,\widetilde{n}}$ as in the proof of Lemma \ref{L: stepping out under rotation}.
	Combining Tonnelli's theorem and Lemma \ref{L: joint distribution}, we get that
	\[
	\theta \mapsto \mathbb{P}\left( L_1^{(\theta)}\in \mathsf{A}_1^{\upiota,n, \widetilde{n}}, \ldots, L_n^{(\theta)}\in \mathsf{A}_n^{\upiota,n, \widetilde{n}},
	R_1^{(\theta)} \in \mathsf{B}_1^{\upiota,n, \widetilde{n}}, \ldots, R_{\widetilde{n}}^{(\theta)} \in \mathsf{B}_{\widetilde{n}}^{\upiota,n, \widetilde{n}}\right)
	\]
	is measurable for all $n, \widetilde{n} \in \mathbb{N}$ with $n +\widetilde{n} \leq m+1$ and $\upiota \in \{n, \ldots, m+1-\widetilde{n}\}$ .
	Thus also
	\begin{align*}
		\theta \mapsto &\xi_{\mathsf{S}}^{(\theta)}(A \times B)\\
		&\ = \frac{1}{m}  \sum_{n=1 }^{m} \sum_{\widetilde{n} = 1}^{m+1 -n}\sum_{\upiota = n}^{m+1-\widetilde{n}} \mathbb{P}\left( L_1^{(\theta)}\in \mathsf{A}_1^{\upiota,n, \widetilde{n}}, \ldots, L_n^{(\theta)}\in \mathsf{A}_n^{\upiota,n, \widetilde{n}},
		R_1^{(\theta)} \in \mathsf{B}_1^{\upiota,n, \widetilde{n}}, \ldots, R_{\widetilde{n}}^{(\theta)} \in \mathsf{B}_{\widetilde{n}}^{\upiota,n, \widetilde{n}}\right)
	\end{align*}
	is measurable.
	It is now straight forward to verify that $(\theta, C) \mapsto \xi_{\mathsf{S}}^{(\theta)}(C)$ is a transition kernel.
\end{remark}

We turn to the proof of Lemma \ref{L: stepping out under rotation}.
Its essential strategy is to leverage the following property of $(L_i)_{i \in \mathbb{N}}$ and $(R_j)_{j \in \mathbb{N}}$ to the stepping-out distribution.

\begin{lemma}\label{L: unstopped stepping put chain under rotation}
	Let Assumption \ref{Set: Stepping-out} hold.
	\begin{enumerate}
		\item Let $ \alpha \in \mathbb{R} $, $ i,j \in \mathbb{N}$ and $ \mathsf{A}_1, \ldots, \mathsf{A}_i, \mathsf{B}_1,\ldots,\mathsf{B}_j \in \mathcal{B}(\mathbb{R}) $.
		Then we have
		\begin{align*}
			&\mathbb{P} \left(L_1^{(\theta)}\in \Lambda_\alpha(\mathsf{A}_1), \ldots, L_i^{(\theta)}\in \Lambda_\alpha(\mathsf{A}_i),R_1^{(\theta)}\in \Lambda_\alpha(\mathsf{B}_1), \ldots, R_j^{(\theta)}\in \Lambda_\alpha(\mathsf{B}_j) \right)\\
			&\quad= \mathbb{P} \left(L_1^{\left(\Lambda_\alpha(\theta)\right)}\in \mathsf{B}_1, \ldots, L_j^{\left(\Lambda_\alpha(\theta)\right)}\in \mathsf{B}_j,R_1^{\left(\Lambda_\alpha(\theta)\right)}\in \mathsf{A}_1, \ldots, R_i^{\left(\Lambda_\alpha(\theta)\right)}\in \mathsf{A}_i \right).
		\end{align*} \label{L: unstopped chains under rotation - no stopping times yet}
		\item For all $ i \in \{1, \ldots,m\} $, $ j \in \{1,\ldots, m+1-i\} $, $\theta,\alpha \in \mathbb{R}$ and $\mathsf{A}, \mathsf{B} \in \mathcal{B}(\mathbb{R}) $ holds
		\begin{align*}
			&\mathbb{P}\left(L_i^{(\theta)}\in \mathsf{A},  R_j^{(\theta)}  \in \mathsf{B} , \tau_{\Lambda_\alpha(\mathsf{S})}^{(\theta)} = i, \mathfrak{T}_{\Lambda_\alpha(\mathsf{S})}^{(\theta)} = j \right)\\
			&\qquad = \mathbb{P}\left( L_j^{(\Lambda_\alpha(\theta))}\in \Lambda_\alpha^{-1}(\mathsf{B}),  R_i^{(\Lambda_\alpha(\theta))}  \in \Lambda_\alpha^{-1}(\mathsf{A}) , \tau_{\mathsf{S}}^{(\Lambda_\alpha(\theta))} = j, \mathfrak{T}_{\mathsf{S}}^{(\Lambda_\alpha(\theta))} = i \right).
		\end{align*} \label{L: unstopped chains under rotation - with stopping times}
	\end{enumerate}
\end{lemma}

\begin{proof}
	\emph{To \ref{L: unstopped chains under rotation - no stopping times yet}:}
	Let $ \alpha \in \mathbb{R} $, $ i,j \in \mathbb{N}$ and $ \mathsf{A}_1, \ldots, \mathsf{A}_i, \mathsf{B}_1,\ldots,\mathsf{B}_j \in \mathcal{B}(\mathbb{R}) $.
	We have by Lemma \ref{L: joint distribution}
	\begin{align*}
		&\mathbb{P} \left(L_1^{(\theta)}\in \Lambda_\alpha(\mathsf{A}_1), \ldots, L_i^{(\theta)}\in \Lambda_\alpha(\mathsf{A}_i),R_1^{(\theta)}\in \Lambda_\alpha(\mathsf{B}_1), \ldots, R_j^{(\theta)}\in \Lambda_\alpha(\mathsf{B}_j) \right)\\
		& \quad = \frac{1}{w} \int_{0}^{w} \prod_{k=1}^{i} \mathbbm{1}_{\Lambda_\alpha(\mathsf{A}_k)}\big(\theta-u-(k-1)w\big) \prod_{l=1}^{j} \mathbbm{1}_{\Lambda_\alpha(\mathsf{B}_l)}\big(\theta-u + lw\big) \ \mathrm{Leb}_1(\d u)\\
		& \quad = \frac{1}{w} \int_{0}^{w} \prod_{k=1}^{i} \mathbbm{1}_{\mathsf{A}_k}\big(\alpha - \theta+u+(k-1)w\big) \prod_{l=1}^{j} \mathbbm{1}_{\mathsf{B}_l}\big(\alpha - \theta+u-lw\big) \ \mathrm{Leb}_1(\d u).
	\end{align*}
	Performing a change of variables $ \widetilde{u} = -u+w $ yields
	\begin{align*}
		&\mathbb{P} \left(L_1^{(\theta)}\in \Lambda_\alpha(\mathsf{A}_1), \ldots, L_i^{(\theta)}\in \Lambda_\alpha(\mathsf{A}_i),R_1^{(\theta)}\in \Lambda_\alpha(\mathsf{B}_1), \ldots, R_j^{(\theta)}\in \Lambda_\alpha(\mathsf{B}_j) \right)\\
		& \quad = \frac{1}{w} \int_{0}^{w} \prod_{k=1}^{i} \mathbbm{1}_{\mathsf{A}_k}\big(\alpha - \theta-\widetilde{u}+kw\big) \prod_{l=1}^{j} \mathbbm{1}_{\mathsf{B}_l}\big(\alpha - \theta-\widetilde{u} - (l-1)w\big) \ \mathrm{Leb}_1(\d \widetilde{u}).
	\end{align*}
	Using again Lemma \ref{L: joint distribution}, we obtain
	\begin{align*}
		&\mathbb{P} \left(L_1^{(\theta)}\in \Lambda_\alpha(\mathsf{A}_1), \ldots, L_i^{(\theta)}\in \Lambda_\alpha(\mathsf{A}_i),R_1^{(\theta)}\in \Lambda_\alpha(\mathsf{B}_1), \ldots, R_j^{(\theta)}\in \Lambda_\alpha(\mathsf{B}_j) \right)\\
		&\quad= \mathbb{P} \left(L_1^{\left(\Lambda_\alpha(\theta)\right)}\in \mathsf{B}_1, \ldots, L_j^{\left(\Lambda_\alpha(\theta)\right)}\in \mathsf{B}_j,R_1^{\left(\Lambda_\alpha(\theta)\right)}\in \mathsf{A}_1, \ldots, R_i^{\left(\Lambda_\alpha(\theta)\right)}\in \mathsf{A}_i \right).
	\end{align*}
	
	\emph{To \ref{L: unstopped chains under rotation - with stopping times}:
		Let $ i \in \{1, \ldots,m\} $, $ j \in \{1,\ldots, m+1-i\} $, $\theta,\alpha \in \mathbb{R}$ and $\mathsf{A}, \mathsf{B} \in \mathcal{B}(\mathbb{R}) $.
		For $ \upiota \in \{i,\ldots, m+1-j\} $ we define
		\begin{align*}
			\mathsf{A}_k^{\upiota} = \Lambda_\alpha(\mathsf{S}),\quad k \in \{1, \ldots, i-1\} \qquad &\text{and} \qquad \mathsf{A}_i^{\upiota} = \begin{dcases} \mathsf{A}, & \upiota = i\\ \mathsf{A} \cap  (\mathbb{R} \setminus\Lambda_\alpha(\mathsf{S})), & \upiota > i,\end{dcases}\\
			\mathsf{B}_l^{\upiota} = \Lambda_\alpha(\mathsf{S}),\quad  l \in \{1, \ldots, j-1\} \qquad &\text{and} \qquad 
			\mathsf{B}_j^{\upiota} =  \begin{dcases} \mathsf{B}, & \upiota = m+1-j\\ \mathsf{B} \cap  (\mathbb{R} \setminus \Lambda_\alpha(\mathsf{S})), & \upiota > m+1-j,\end{dcases} 
		\end{align*}
		as well as
		\begin{align*}
			\widetilde{\mathsf{A}}_k^{\upiota} =  \mathsf{S},\quad k \in \{1, \ldots, i-1\} \qquad &\text{and} \qquad \widetilde{\mathsf{A}}_i^{\upiota} = \begin{dcases} \Lambda_\alpha^{-1}(\mathsf{A}), & \upiota = i\\ \Lambda_\alpha^{-1}(\mathsf{A}) \cap  (\mathbb{R} \setminus\mathsf{S}), & \upiota > i,\end{dcases}\\
			\widetilde{\mathsf{B}}^{\upiota}_l = \mathsf{S},\quad  l \in \{1, \ldots, j-1\} \qquad &\text{and} \qquad 
			\widetilde{\mathsf{B}}^{\upiota}_j =  \begin{dcases} \Lambda_\alpha^{-1}(\mathsf{B}),  &\upiota = m+1-j\\ \Lambda_\alpha^{-1}(\mathsf{B}) \cap  (\mathbb{R} \setminus\mathsf{S}), & \upiota > m+1-j.\end{dcases}
		\end{align*}
		Observe that 
		\begin{align*}
			\Lambda_\alpha(\widetilde{\mathsf{A}}_k^{\upiota}) = \mathsf{A}_k^{\upiota}, \qquad \text{and} \qquad 
			\Lambda_\alpha(\widetilde{\mathsf{B}}_l^{\upiota}) = \mathsf{B}_l^{\upiota}, \qquad \textup{for }k \in \{1,\ldots, i\}, l\in \{1,\ldots, j\},
		\end{align*}
		as $ \Lambda_\alpha(\mathbb{R}) = \mathbb{R}$.
		Due to $ J \sim \mathrm{Unif}(\{1, \ldots,m\}) $ being independent of $(L_k^{(\theta)})_{k \in \mathbb{N}}$ and $(R_l^{(\theta)})_{l \in \mathbb{N}}$, and \eqref{Eq: Bound tau, T}, 
		using statement \ref{L: unstopped chains under rotation - no stopping times yet} gives
		\begin{align*}
			&\mathbb{P}\left(L_i^{(\theta)}\in \mathsf{A},  R_j^{(\theta)}  \in \mathsf{B} , \tau_{\Lambda_\alpha(\mathsf{S})}^{(\theta)} = i, \mathfrak{T}_{\Lambda_\alpha(\mathsf{S})}^{(\theta)} = j \right)\\
			&\qquad = \sum_{\upiota = i}^{m+1-j} \mathbb{P}\left(L_1^{(\theta)}\in \mathsf{A}_1^{\upiota},\ldots, L_i^{(\theta)}\in \mathsf{A}_i^{\upiota},R_1^{(\theta)}\in \mathsf{B}_1^{\upiota},\ldots, R_j^{(\theta)}\in \mathsf{B}_j^{\upiota}\right) \mathbb{P}(J = \upiota)\\
			&\qquad =\frac{1}{m}\sum_{\upiota = i}^{m+1-j} \mathbb{P}\left(L_1^{(\theta)}\in \Lambda_\alpha(\widetilde{\mathsf{A}}_1^{\upiota}),\ldots, L_i^{(\theta)}\in \Lambda_\alpha(\widetilde{\mathsf{A}}_i^{\upiota}),R_1^{(\theta)}\in \Lambda_\alpha(\widetilde{\mathsf{B}}_1^{\upiota}),\ldots, R_j^{(\theta)}\in\Lambda_\alpha(\widetilde{\mathsf{B}}_j^{\upiota})\right)\\
			&\qquad =\frac{1}{m}\sum_{\upiota = i}^{m+1-j} \mathbb{P}\left(L_1^{(\Lambda_\alpha(\theta))}\in \widetilde{\mathsf{B}}_1^{\upiota},\ldots, L_j^{(\Lambda_\alpha(\theta))}\in \widetilde{\mathsf{B}}_j^{\upiota},R_1^{(\Lambda_\alpha(\theta))}\in \widetilde{\mathsf{A}}_1^{\upiota},\ldots, R_i^{(\Lambda_\alpha(\theta))}\in\widetilde{\mathsf{A}}_i^{\upiota}\right)\\
			&\qquad = \mathbb{P}\left( L_j^{(\Lambda_\alpha(\theta))}\in \Lambda_\alpha^{-1}(\mathsf{B}),  R_i^{(\Lambda_\alpha(\theta))}  \in \Lambda_\alpha^{-1}(\mathsf{A}) , \tau_{\mathsf{S}}^{(\Lambda_\alpha(\theta))} = j, \mathfrak{T}_{\mathsf{S}}^{(\Lambda_\alpha(\theta))} = i \right).
	\end{align*}}
\end{proof}

\begin{proof}[of Lemma \ref{L: stepping out under rotation}]
	Let $ \mathsf{A},\mathsf{B} \in \mathcal{B}(\mathbb{R}) $.
	Using \eqref{Eq: Bound tau, T}, we obtain
	\begin{equation*}
		\begin{split}
			\xi_{\Lambda_\alpha(\mathsf{S})}^{(\theta)}(\mathsf{A} \times \mathsf{B}) 
			&=\mathbb{P}\left( \boldsymbol{L}^{(\theta)}_{\Lambda_\alpha(\mathsf{S})}\in \mathsf{A},  \boldsymbol{R}^{(\theta)}_{\Lambda_\alpha(\mathsf{S})} \in \mathsf{B} \right)\\
			&= \sum_{i = 1}^m \sum_{j = 1}^{m+1-i} \mathbb{P}\left( L_i^{(\theta)}\in \mathsf{A},  R_j^{(\theta)}  \in \mathsf{B} , \tau_{\Lambda_\alpha(\mathsf{S})}^{(\theta)} = i, \mathfrak{T}_{\Lambda_\alpha(\mathsf{S})}^{(\theta)} = j \right).
		\end{split}
	\end{equation*}
	Applying Lemma \ref{L: unstopped stepping put chain under rotation}.\ref{L: unstopped chains under rotation - with stopping times} and a reordering of the sum yields
	\begin{align*}
		\xi_{\Lambda_\alpha(\mathsf{S})}^{(\theta)}(\mathsf{A} \times \mathsf{B}) 
		&= \sum_{j = 1}^m \sum_{i = 1}^{m+1-j} \mathbb{P}\left( L_j^{(\Lambda_\alpha(\theta))}\in \Lambda_\alpha^{-1}(\mathsf{B}),  R_i^{(\Lambda_\alpha(\theta))}  \in \Lambda_\alpha^{-1}(\mathsf{A}) , \tau_{\mathsf{S}}^{(\Lambda_\alpha(\theta))} = j, \mathfrak{T}_{\mathsf{S}}^{(\Lambda_\alpha(\theta))} = i \right)\\
		& = \mathbb{P}\left( \boldsymbol{L}^{(\Lambda_\alpha(\theta))}_{\mathsf{S}}\in \Lambda_\alpha^{-1}(\mathsf{B}),  \boldsymbol{R}^{(\Lambda_\alpha(\theta))}_{\mathsf{S}} \in \Lambda_\alpha^{-1}(\mathsf{A}) \right)
		= \xi_{\mathsf{S}}^{(\Lambda_\alpha(\theta))}\left(\uplambda_\alpha^{-1}(\mathsf{A} \times \mathsf{B})\right),
	\end{align*}
	such that the probability measures $\xi_{\Lambda_\alpha(\mathsf{S})}^{(\theta)}$ and $ (\uplambda_\alpha)_\sharp\xi_{\mathsf{S}}^{(\Lambda_\alpha(\theta))}$ agree on the intersection stable generator  $\{\mathsf{A} \times \mathsf{B} \mid \mathsf{A}, \mathsf{B} \in \mathcal{B}(\mathbb{R}) \}$ of $\mathcal{B}(\mathbb{R}^2)$.
	Then \citep[Lemma 1.42]{Klenke} gives the result. 
\end{proof}

\section{Uniform simple slice sampling}\label{Sec: Slice sampling in measurable spaces}
Let $(\mathsf{X}, \mathcal{X}, \mu)$ be a $\sigma$-finite measure space.
For illustrative purposes we introduce uniform simple slice sampling for probability measures on $(\mathsf{X}, \mathcal{X})$ that are absolutely continuos with respect to $\mu$.
More precisely, let 
\[
p: \mathsf{X} \to (0, \infty)
\]
be a measurable function that satisfies 
\[
Z:= \int_{\mathsf{X}} p(x)\ \mu(\d x) \in (0, \infty).
\]
We consider the probability measure
\begin{equation}\label{Eq: Definition target - general measurable space}
	\pi(\d x) := \frac{1}{Z}p(x) \mu(\d x)
\end{equation}
that has unnormalised density $p$ with respect to $\mu$.
We define the level sets as
\[
L(t) := \{ x \in \mathsf{X} \mid p(x)> t\}, \qquad t \in (0, \infty),
\]
and the essential supremum norm of $p$
\[
\|p \|_{\text{ess-}\infty}:=\inf_{\mathsf{N} \in \mathcal{X}, \mu(\mathsf{N}) = 0} \sup_{x \in \mathsf{X}\setminus \mathsf{N}} p(x).
\]
By definition of $\|p \|_{\text{ess-}\infty}$ and because
\begin{align*}
	&\int_0^\infty \mu\big(L(t)\big) \Leb_1(\d t) = \int_0^\infty \int_{\mathsf{X}} \mathbbm{1}_{L(t)}(x) \ \mu(\d x)\, \Leb_1(\d t)\\
	&\qquad= \int_{\mathsf{X}} \int_0^\infty  \mathbbm{1}_{(0, p(x))}(t) \  \Leb_1(\d t)\, \mu(\d x)
	= \int_{\mathsf{X}} p(x)\ \mu(\d x) = Z < \infty,
\end{align*}
we have that for all $t \in (0, \|p \|_{\text{ess-}\infty})$ holds $\mu\big(L(t)\big) \in (0, \infty)$,
due to the monotonicity of $t \mapsto \mu(\mathsf{L}(t))$.
Consequently, the Markov kernel of the uniform simple slice sampler
\begin{equation}\label{Eq: Slice sampling kernel - measrurable space}
	H(x,\mathsf{A}) := \frac{1}{p(x)}\int_{0}^{p(x)} \mu_t(\mathsf{A})\ \Leb_1(\d t), \qquad x \in \mathsf{X}, \mathsf{A} \in \mathcal{X}, 
\end{equation}
where
\[
\mu_t := \frac{1}{\mu(L(t))}\mu\vert_{L(t)}, \qquad t \in (0, \|p \|_{\text{ess-}\infty}),
\]
is well defined.\footnote{Strictly speaking, $H$ is only $\mu$-almost surely well-defined. However, wlog we may assume $\|p \|_{\text{ess-}\infty} = \|p \|_{\infty}$, since changing $p$ on a null set leaves $\pi$ invariant. Under this condition, $H$ is indeed everywhere well defined.}

\begin{lemma}
	Let $(\mathsf{X}, \mathcal{X}, \mu)$ be a $\sigma$-finite measure space, and  let $p: \mathsf{X} \to (0, \infty)$
	be a measurable function with $Z:= \int_{\mathsf{X}} p(x)\ \mu(\d x) \in (0, \infty)$.
	Then $H$, as defined in \eqref{Eq: Slice sampling kernel - measrurable space}, is reversible with respect to $\pi$, given in \eqref{Eq: Definition target - general measurable space}.
\end{lemma}

\begin{proof}
	Let $\mathsf{A}, \mathsf{B} \in \mathcal{X}$.
	We have
	\begin{align*}
		&\int_B H(x,\mathsf{A}) \ \pi(\d x)
		= \int_B \frac{p(x)}{Z\, p(x)} \int_0^{p(x)} \mu_t(\mathsf{A}) \ \Leb_1(\d t) \, \mu(\d x)\\
		&\qquad= \frac{1}{Z} \int_0^{\infty}   \int_B\mu_t(\mathsf{A})  \mathbbm{1}_{(0, p(x))}(t)\  \mu(\d x)\, \Leb_1(\d t)
		= \frac{1}{Z} \int_0^{\infty} \frac{\mu(\mathsf{A} \cap L(t)) \mu(\mathsf{B} \cap L(t))}{\mu(L(t))}  \ \Leb_1(\d t).
	\end{align*}
	This expression is symmetric in $\mathsf{A}$ and $\mathsf{B}$, such that 
	\begin{align*}
		\int_{\mathsf{B}} H(x,\mathsf{A}) \ \pi(\d x) = 	\int_{\mathsf{A}} H(x,\mathsf{B}) \ \pi(\d x) .
	\end{align*}
\end{proof}

\FloatBarrier
\section{Supplementary material for the numerical experiments}\label{Sec: Numerics supplement}
\subsection{MCMC-methods used in the numerical experiments}\label{Sec: MCMC samplers}
First we comment on the structure on the Stiefel and the Grassmann manifold needed for the implementation of GSS.
\begin{remark}
		Let $n,k \in \mathbb{N}$ with $k \leq n$.
	\begin{enumerate}
		\item(\emph{Stiefel manifold}.) 
		The Stiefel manifold $\mathcal{V}(n,k)$, including an explicit formula for its geodesic and the uniform distribution on its tangent spheres, is introduced in Example \ref{Ex:Manifolds}.\ref{Ex:Manifolds_iii} and Example \ref{Ex:_Illustative_scenarios_for_GSS}.\ref{Ex:_Illustative_scenarios_for_GSS_iii}.
		It is worth noting that the computation of a geodesic requires evaluating the matrix exponential map on skew symmetric matrices of size $2k\times 2k$, which can be efficiently done using the eigenvalue decomposition of skew symmetric matrices.
		\item(\emph{Grassmann manifold}.)  
		The $k(n-k)$-dimensional Grassmann manifold $\mathcal{G}(n,k)$ can be defined as the set of all $k$-dimensional subspaces of the Euclidean space $\mathbb{R}^n$ (see \citep{bendokat2020grassmann} for a complete overview), i.e.,
		\begin{equation}
			\mathcal{G}(n,k)=\{W\subseteq \mathbb{R}^n \mid W \textrm{ is a } k\textrm{-dimensional subspace} \}.
		\end{equation}
		The following representation allows for an efficient implementation of points on the Grassmann manifold. First, observe that for any subspace $W \in \mathcal{G}(n,k)$, there exists a (non-unique) orthonormal basis formed by the columns of $\Gamma \in \mathcal{V}(n,k)$ such that $W = \operatorname{span} \Gamma$. Thus, by defining the equivalence relation $\sim$ on $\mathcal{V}(n,k)$ as $\Gamma_1 \sim \Gamma_2$ if and only if $\operatorname{span} \Gamma_1 = \operatorname{span} \Gamma_2$, the Grassmann manifold $\mathcal{G}(n,k)$ can be identified with the quotient manifold $\mathcal{V}(n,k)/ \sim$. As a result, an element $W$ of $\mathcal{G}(n,k)$ can be represented by an element $\Gamma$ of $\mathcal{V}(n,k)$.
		Fixing $\Gamma_\perp \in \mathbb{R}^{n \times (n-k)}$ such that the columns of $\Gamma$ and $\Gamma_\perp$ form an orthonormal basis of $\mathbb{R}^n$, the tangent space to $\mathcal{G}(n,k)$ at $W$ is given by
		\[
		T_W\mathcal{G}(n,k) := \{ \Gamma_\perp \Sigma \in \mathbb{R}^{n \times k} \mid \Sigma \in \mathbb{R}^{(n-k) \times k}\}.
		\]
		We equip $\mathcal{G}(n,k)$ with the Riemannian metric
		\[
		\mathfrak{g}_W(\Delta_1, \Delta_2) = \mathrm{tr}(\Delta_1^\top \Delta_2), \qquad \Delta_1, \Delta_2 \in T_W\mathcal{G}(n,k),
		\]
		such that $\Sigma \mapsto \Gamma_\perp \Sigma$ is an isometric isomorphism between $\mathbb{R}^{n \times (n-k)}$ and $T_W\mathcal{G}(n,k)$. Let $\Delta = \widehat{U}D\widehat{V}^\top$ be the compact singular value decomposition (SVD) of an element $\Delta \in T_W\mathcal{G}(n,k)$. Then, the geodesic at $W$ with direction $\Delta \in T_W \mathcal{G}(n,k)$ admits the following explicit expression
		\[
		\gamma_{(W, \Delta)} (\theta) = \big(\Gamma \widehat{V} \cos(D \theta) + \widehat{U} \sin(D\theta)\big) \widehat{V}^\top, \qquad \theta \in \mathbb{R}.
		\]
		Further details and derivations can be found in \citep{edelman1998geometry}.
	\end{enumerate}
\end{remark}

For the convenience of the reader we provide  more details on the three comparison samplers appearing in our  numerical experiments.
Each sampler defines a Markov chain $(X_i)_{i \in \mathbb{N}}$, and we describe the respective transition mechanisms from $X_i$ to $X_{i+1}$ for $i \in \mathbb{N}$.
\begin{description}
	\item[\textbf{ (a) adaptive random walk Metropolis Hastings (RMH). }]  
	Given $X_i\in \mathcal{V}(n,k)$ and a step size $h_a\in (0, \infty)$ this sampler uses the following Metropolis-Hastings (MH) like transition mechanism:
	\begin{enumerate}
		\item Sample $V_{i+1} \sim \mathcal{N}(\boldsymbol{0}_{nk}, \mathrm{Id}_{n k})$ interpreted as a matrix in $ \mathbb{R}^{n \times k}$.
		Let $\operatorname{proj}_{\mathcal{V}(n,k)}:\mathbb{R}^{n \times k} \to \mathcal{V}(n,k),\  V = \widehat{U}D\widehat{V}^\top \mapsto \widehat{U}\widehat{V}^\top $ be the projection on the Stiefel manifold, where $ \widehat{U}D\widehat{V}^\top $ is the SVD of $V$.
		Then define $ \widetilde{X}_{i+1}=\operatorname{proj}_{\mathcal{V}(n,k)}(X_i+h_a V_{i+1})$.
		\item Sample $\Upsilon_{i+1} \sim \mathrm{Unif}\big([0,1]\big)$ independent of all previously appearing random variables. If $\Upsilon_{i+1}<\upalpha (X_i,\widetilde{X}_{i+1})$, where 
		\begin{equation}
			\label{eq:acceptance_ratio}
			\upalpha(X_i,\widetilde{X}_{i+1})=\frac{p(\widetilde{X}_{i+1})}{p(X_i)},
		\end{equation}
		then set $X_{i+1}=\widetilde{X}_{i+1}$, otherwise $X_{i+1}=X_i$.
	\end{enumerate}
	The hyperparameter $h_a$ is adaptively tuned to target an acceptance probability of $0.234$, a constant proposed by optimal design analysis when using the same sampler in Euclidean space \citep{roberts2001optimal}.
	As remarked in \citep{zappa2018monte}, the proposal mechanism is not symmetric here. More precisely, if  the conditional distribution of $\widetilde{X}_{i+1}$ given $X_i=\Gamma\in \mathcal{V}(n,k)$ admits a  transition density denoted by $q$, this function is not symmetric, i.e., $q(\Gamma_1 |\Gamma_2)\neq q(\Gamma_2 |\Gamma_1)$ for $\Gamma_1, \Gamma_2 \in  \mathcal{V}(n,k)$.
	Therefore, using \eqref{eq:acceptance_ratio} as the acceptance ratio leads to a biased MCMC algorithm, which is justified in \citep{mantoux2021understanding} by the observation that, for sufficiently small $h_a > 0$, $q(\Gamma_1 |\Gamma_2) \approx q(\Gamma_2 |\Gamma_1)$ holds for all $\Gamma_1, \Gamma_2 \in \mathcal{V}(n,k)$. However, this introduces a bias that grows with increasing $h_a$.
	In Section \ref{sub_sec:practical_case}, we compare adaptive RMH with GSS using the same application as discussed in \citep{mantoux2021understanding}.\\
	It is worth noting that \citet{zappa2018monte} propose an MCMC algorithm that, by employing another projection and an additional accept-reject step, results in an unbiased MCMC algorithm. However, their approach does not exploit the specificity of a ``tractable geodesics''-framework.
	\item[\textbf{ (b) geodesic adaptive Metropolis Hastings (GeoRMH). }]
	This methods is a bias free modification of RMH.
	It essentially uses the same transition mechanism, but replaces the proposal at step $i\in \mathbb{N}$ with $\gamma_{(X_i, V_{i+1})}(h_a)$
	where $V_{i+1}$ is distributed according to a standard Gaussian on $T_{X_i}\mathcal{V}(n,k)$. In \citep{goyal2019sampling}, a similar algorithm is proposed to sample uniformly from a convex subset of a manifold with non-negative curvature.
	\item[\textbf{ (c) geodesic Metropolis adjusted Langevin algorithm (GeoMALA). }] 
	On the Grassmann manifold, GeoMALA is used and presented in \citep[Algorithm 1]{holbrook2016bayesian} to estimate parameters in Bayesian inference models involving dimensionality reduction.
	Fix a step size $h \in (0, \infty)$ and denote by $\operatorname{proj}_{T_W\mathcal{G}(n,k)}:\mathbb{R}^{n \times k} \to T_W\mathcal{G}(n,k),  V \mapsto (\mathrm{Id}_n - \Gamma \Gamma^\top) V$ for $W = \mathrm{span}\ \Gamma \in \mathcal{G}(n,k)$ the projection onto the tangent space.
	The transition mechanism of GeoMALA given $X_i\in \mathcal{G}(n,k)$ works as follows:
	\begin{enumerate}
		\item Sample $V_{i+1} \sim \mathcal{N}(\boldsymbol{0}_{nk}, \mathrm{Id}_{n k})$ interpreted as a matrix in $ \mathbb{R}^{n \times k}$. Then define $ \bar{V}_{i+1}=\operatorname{proj}_{T_{X_i}\mathcal{G}(n,k)}(V_{i+1})$ and $E_0=\log p(X_i)-  \mathrm{Tr}(\bar{V}_{i+1}^\top \bar{V}_{i+1})/2$.
		\item Define $\widetilde{V}_{i+1}=\operatorname{proj}_{T_{X_i}\mathcal{G}(n,k)}(\bar{V}_{i+1}+h\nabla \log p(X_i)/2) $ and $\bar{X}_{i+1}=\gamma_{(X_i, \widetilde{V}_{i+1})}(h)$, $\bar{V}_{i+1}=(\d \gamma_{(X_i, \widetilde{V}_{i+1})}/\d t)\vert_h$,
		as well as $V_{i+1}^*=\operatorname{proj}_{T_{X_i}\mathcal{G}(n,k)}(\bar{V}_{i+1}+h\nabla \log p(\bar{X}_i)/2)$ and $E_1=\log p(\bar{X}_{i+1})- \mathrm{Tr}(V_{i+1}^{*\top} V_{i+1}^*)/2$.
		\item Sample $\Upsilon_{i+1} \sim \mathrm{Unif}\big([0,1]\big)$ independently. If $\log(\Upsilon_{i+1})<E_1-E_0$
		then set $X_{i+1}=\bar{X}_{i+1}$, otherwise $X_{i+1}=X_i$.
	\end{enumerate}
	Note that GeoMALA is biased (Figure \ref{fig:bias_analysis_geomala}) for the same reason as RMH is biased: the proposal kernel is not symmetric because of the projection step.
\item [\textbf{ (d) uniform independent Metropolis Hastings (UIMH). }] Given $X_i\in \mathcal{V}(n,k)$, this sampler uses the following Metropolis-Hastings (MH) like transition mechanism :
\begin{enumerate}
	\item Sample $\bar{X}_i$ uniformly\footnote{This is performed using scipy.stats.ortho_group.rvs using \cite{mezzadri2006generate}.} on $\mathcal{V}(n,k)$. Then, define $E_0=\log p(X_{i})$ and $E_1=\log p(\bar{X}_{i+1})$.
	\item Sample $\Upsilon_{i+1} \sim \mathrm{Unif}\big([0,1]\big)$ independently. If $\log(\Upsilon_{i+1})<E_1-E_0$
	then set $X_{i+1}=\bar{X}_{i+1}$, otherwise $X_{i+1}=X_i$.
\end{enumerate} 
\end{description}

Note that for GSS, GeoRMH, and GeoMALA we always reproject the (final) state of the Markov chain at each step onto the Stiefel manifold to mitigate numerical errors arising from the geodesic computations.
This reprojection step ensures that the resulting samples lie on the manifold, preserving the desired manifold structure and improving the accuracy of the sampling algorithms.

\subsection{Sampling the von Mises–Fisher distribution}\label{Sec: Experiments von Mises-Fisher distribution}
In this section we present numerical experiments targeting the matrix von Mises–Fisher distribution.
Given a matrix-valued parameter $F\in \mathbb{R}^{n \times k}$, the matrix von Mises–Fisher distribution $\operatorname{vMF}(F)$ on the Stiefel manifold $\mathcal{V}(n,k)$ has unnormalized density with respect to its Riemannian measure
\begin{equation}
	\label{eq:Mises_fisher_distribution}
	p_{\operatorname{vMF}(F)}(\Gamma) = \exp\left(\mathrm{tr}(F^\top \Gamma )\right), \qquad \Gamma \in \mathcal{V}(n,k).
\end{equation}  
However, this expression can not be used for the Grassmann manifold since $\operatorname{span} \Gamma_1=\operatorname{span} \Gamma_2$ with $(\Gamma_1, \Gamma_2) \in \mathcal{V}(n,k)^2$ does not imply  $p_{\operatorname{vMF}(F)}(\Gamma_1)= p_{\operatorname{vMF}(F)}(\Gamma_2)$.
Note that an element $W$ of the Grassmann manifold can be identified with its orthogonal projector $\Gamma\Gamma^\top$, which does not depend on the choice of the representative $\Gamma\in \mathcal{V}(n,k)$ for $W$.
Therefore, given a positive semi-definite matrix $P\in \mathbb{R}^{n \times n}$, the von Mises–Fisher distribution $\operatorname{vMF}(P)$ (also called matrix Langevin distribution \citep{chikuse2003concentrated}) on the Grassmann manifold $\mathcal{G}(n,k)$ has unnormalized density
\begin{equation}
	\label{eq:Mises_fisher_distribution_grassman}
	p_{\operatorname{vMF}(P)}(W) = \exp\left(\mathrm{tr}(P^\top \Gamma\Gamma^\top )\right),\quad \operatorname{span} \Gamma=W \qquad (\Gamma,W)\in \mathcal{V}(n,k)\times \mathcal{G}(n,k).
\end{equation}

In the following, given $n$ and $k$, we always choose the parameters $F$ and $P$ to be of the form
\begin{equation}
	\label{eq:F_exp12}
	F=\begin{pmatrix} 
		D \\ \boldsymbol{0}_{n-k,k}
	\end{pmatrix} \in \mathbb{R}^{n \times k},\qquad P=FF^\top.
\end{equation}
where $D \in \mathbb{R}^{k\times k}$.

We consider three different experimental setups on $\mathcal{V}(n,k)$.
Firstly, $k$ is fixed at $2$, we vary $n$ in the set $\{3, 30, 100\}$ and set $D = \mathrm{diag}((1,\ldots, k)^\top) \in \mathbb{R}^{d\times d}$ in \eqref{eq:F_exp12}.
Secondly, for the same choice of $D$ as in the first experiment, we fix $n = 30$ and vary $k$ in the set $\{3, 30, 100\}$.
Thirdly, the pair $(n,k)$ is fixed at $(30,2)$, but we choose $D= \mathrm{diag}((1, \lambda)^\top)\in \mathbb{R}^{2\times 2}$, where $\lambda$ varies in the set $\{1, 10, 100\}$.
This allows us to study the impact of the target distribution's  anisotropy on the samplers' performance.

On the Grassmann manifold $\mathcal{G}(n,k)$ we repeat similar experiments.
However in the first two setups, we use $D = \mathrm{Id}_k$ in \eqref{eq:F_exp12}.
For the third experiment, we set $(n,k) = (3,2)$ and $D = \sqrt{\lambda} \mathrm{Id}_k$  for $\lambda \in \{1, 10, 100\}$.
There, we also vary the stepping-out parameter $m$ of GSS  in the set $\{1,3,10\}$.
In all the other experiments, we fix the hyperparameter $m$ of GSS to 1 for the sake of simplicity.
The samplers run for $N=100,000$ iterations, and the stepsize $h_a$ for RMH and GeoRMH is initialized to 0.01 and updated every 20 steps.
  
We define the initialization $X_1$ by first sampling $\widetilde{X}_1 \sim \mathrm{Unif}\left([0,1]^{n \times k}\right)$.
Then, $\widetilde{X}_1$ is projected onto $\mathcal{V}(n,k)$ using \citep[Lemma 2]{mantoux2021understanding} to obtain $X_1$, thereby identifying a point on $\mathcal{V}(n,k)$ with its equivalence class when we work on $\mathcal{G}(n,k)$.
To evaluate the performance of the samplers, we compute the effective sample size (ESS) of $(\log(p_{\operatorname{vMF}(F)}(x_i)))_{i\in\{1,\ldots,N\}}\in \mathbb{R}^N$, where $(x_i)_{i\in \{1,\ldots, N\}}$ denotes the samples generated by a single sampler.
The results are averaged over 10 different resamplings, but the initialization is kept fixed since the use of a burn-in period has no significant impact.

 \begin{table}[!h]
	  \caption{
	  Effective sample size [min, median, max] of 10 repetitions for varying dimension $n$, on the Stiefel manifold $\mathcal{V}(n,k)$.
	  All values have been multiplied by $10^{-3}$.
		}
	  \centering
	  \begin{tabular}{llll}
		\toprule 
		$(n,k)$& (3,2) & (30,2)& (100,2) \\
	   GSS $w=1$   & $[1.8, 2.1, 2.5]$  & $[1.6, 2.0, 2.2]$& $[1.6, 2.0, 2.1]$      \\
	   GSS $w=3$  & $ [9.3,10.2 , 11.4]$ &$ [12.2, 13.4, 15.0]$& $ [14.3, 15.7, 18.2]$    \\
	   GSS $w=5$  & $ [11.9, 14.0, \textbf{18.8}]$  &$[22.7, 26.3, 32.0]$& $ [31.9, 35.7, 42.1]$   \\
	   GSS $w=7$  & $ [13.1, 14.7, 18.2]$  &$[30.2,33.8,35.7]$& $ [43.2, 52.1, 57.6]$  \\
	   GSS $w=9$  & $ [12.8, 14.1, 16.6]$  &$[34.5,39.0,\textbf{48.8}]$& $ [47.7, 57.8, \textbf{68.8}]$  \\
	   GSS $w=11$  & $ [\textbf{14.4, 15.6}, 16.7]$&$[\textbf{35.2, 39.7},46.5]$& $ [51.7, 57.7, 62.8]$  \\
		 RMH:  & $ [9.9, 12.0, 13.9]$  &$[33.5,37.8,43.5]$& $ [\textbf{52.1, 60.0}, 68.1]$  \\
		 GeoRMH:  & $ [7.0, 8.4, 10.7]$  &$[28.4,34.1,37.7]$& $ [45.4, 51.1, 56.3]$  \\
	  \end{tabular}
	  \label{table:sample_exp1}
	  
	\end{table}
	
	\begin{table}[!h]
		  \caption{
		  Effective sample size [min, median, max] of 10 repetitions for varying dimension $k$, on the Stiefel manifold $\mathcal{V}(n,k)$.
		  All values have been multiplied by $10^{-3}$.
		  }
		  
		  \centering
		  \begin{tabular}{llll}
			\toprule 
			$(n,k)$& (30,5) & (30,10)& (30,20) \\
		   GSS $w=1$   & $[0.7, 0.8, 0.9]$  & $[0.4, 0.4, 0.4]$& $[0.32, 0.33, 0.34]$      \\
			GSS $w=3$  & $ [3.1, 3.6, 4.0]$ &$ [0.9, 1.1, 1.2]$& $ [0.36, 0.39, 0.43]$    \\
			 GSS $w=5$  & $ [\textbf{5.3, 5.8, 6.5}]$  &$[\textbf{1.1, 1.2, 1.3}]$& $ [\textbf{0.38, 0.40, 0.43}]$   \\
			 RMH:  & $ [2.7, 3.3, 4.3]$  &$[1.0,1.0,1.2]$& $ [0.34,0.37, 0.39]$  \\
			 GeoRMH:  & $ [1.0, 4.2, 4.7]$  &$[1.0,1.2,1.3]$& $ [0.36, 0.37, 0.39]$  \\
		  \end{tabular}
		  \label{table:sample_exp2}
		\end{table}

\bigskip	
\textbf{Varying $(n,k)$ on the Stiefel manifold.}
First, we observe in Table \ref{table:sample_exp1} and \ref{table:sample_exp2} that the larger the value of $w$, the higher the effective sample size (ESS) for GSS.
This is coherent with the fact that the maximal size of the region taken into account by the stepping-out and shrinkage procedure increases with $w$ such that 
the space is better explored.
However, note that the finite diameter of a compact manifold like $\mathcal{V}(n,k)$ sets a natural limit to this effect.
As a result, the gain in ESS reaches a plateau when $w\geq 7\geq 2\pi$.

The second observation is that for all samplers, as $n$ increases for fixed $k$, the ESS increases, and inversely, as $k$ increases with $n$ fixed, the ESS decreases. 
Considering \eqref{eq:F_exp12}, the number of directions that impact the density is equal to $k$ since $\mathrm{tr}(F^\top\Gamma)=\sum_{i=1}^k f_i^\top\Gamma_i$ for any $F=(f_i)_{i\in \{1,\ldots,k\}},\Gamma=(\Gamma_i)_{i\in \{1,\ldots,k\}}\in \mathbb{R}^{n\times k}$. Thus, as $k/n$ is small, 
the target density is ``flat'' in many directions.
As a result there are more directions in which GSS can move further and the risk of proposal rejection is smaller for the Metropolis-Hastings type algorithms, respectively.

In Table \ref{table:sample_exp1}, we see that RMH outperforms GSS and GeoRMH for $n=100$, and remains competitive for $n\in\{3,30\}$.
This, maybe at first glance surprising, performance of RMH can be explained by the fact that when $k/n$ is small, the number of constraints is low compared to the dimensionality of the space, making $\mathcal{V}(n,k)$ nearly Euclidean and the Stiefel projection nearly equal to the identity.
Consequently, since the optimal acceptance rate of $0.234$ for the adaptive mechanism was found in an Euclidean space, it is reasonably explainable that RMH performs better in this scenario.
This interpretation is further confirmed by Table \ref{table:sample_exp2}, where we see that GSS($w=5$) outperforms RMH and GeoRMH for any $k\in\{5,10,20\}$, though this advantage is less significant for GeoRMH, which also explores the space taking into account its intrinsic geometry.
Recall, to put our observations into perspective, that while RMH may outperform GSS and GeoRMH in terms of ESS, it is fundamentally biased.

In Table \ref{table:sample_exp1}, we observe that GeoRMH is not as efficient as RMH but performs comparably to GSS($w=7$) when $n\in \{30,100\}$.
This difference can be linked to the fact that GeoRMH and GSS only differ in the method they employ to move on a geodesic. The first uses a Metropolis Hastings mechanism whereas the latter runs a slice sampler.
Using slice sampling ensures that the sampled direction in the tangent space is not wasted, 
	but (likely) increases the computational cost of each transition step, 
	since each iteration within the stepping-out and shrinkage procedure involves new computations.
	This is affected by the choice of the parameter $w$,
	as it determines the considered size of the portion of the geodesic, and thus influences the probability to miss the level set, within the shrinkage procedure.
For example, in the case $n=3$, there are, on average, $1.11$ attempts when $w=1$, but $1.41$ attempts when $w=5$.
Therefore, there is a trade-off when selecting $w$ to optimize the time efficiency of sampling, as larger values of $w$ increase both the ESS and the computation time.

\bigskip
\textbf{Varying $(n,k)$ on the Grassmann manifold.}
In Table \ref{table:sample_exp12_grasman}, we present only a subset of our experiments to convey the following message: When GeoMALA is well tuned, the gradient-informed sampler outperforms GSS regardless of the choice of $w$.
 However, the tuning of GeoMALA is very sensitive and GeoMALA is biased especially for large step size (Figure \ref{fig:bias_analysis_geomala}). For instance, in the case $(n,k)=(100,2)$ with $h=1$, the gradient information encourages to focus on high density area and does not explore enough to outperform GSS, but if we increase the stepsize $h$ to $2$, GeoMALA is better than GSS.
Regarding the role of the hyperparameter $w$, the conclusions are consistent with those on the Stiefel manifold. 
Increasing $w$ beyond a certain point does not pay off, 
since the manifold is compact. 
The sweet spot appears to be around $w=7\approx 2\pi$.
Both methods have the same complexity, as in both cases, we need to sample a point on the tangent space and compute a geodesic using SVD.
	\begin{table}
		  \caption{Effective sample size [min, median, max] of 10 repetitions for varying dimension $(n,k)$ with a fixed shape of distribution, on the Grassmann manifold.
		  We show the results only for $w=7$ since it does not affect the comparison with GeoMALA, and the dependence according to $w$ is globally the same as on the Stiefel manifold. 
		  ``GeoMALA Best $h$'' means that we provide the result for the best stepsize parameter $h$ in $\{0.01,0.1,0.5,1,2\}$. All values have been multiplied by $10^{-3}$.}

		  \centering
			\begin{tabular}{llll}
			$(n,k)$ & (3,2) & (30,20) & (100,2) \\
			GSS $w=7$  & $ [29, 37, 41]$  &$[22, 24, 28]$& $ [27, 30, 36]$   \\
			GeoMALA Best $h$  & $ [\textbf{46, 51, 55}]$  & $[\textbf{39, 41, 48}]$& $ [\textbf{56, 61, 71}]$  \\
		\end{tabular}
		\label{table:sample_exp12_grasman}
	\end{table}
	
\bigskip	
\textbf{Varying anisotropy on the Stiefel manifold.}
\begin{table}
	\caption{Effective sample size [min, median, max] of 10 repetitions for varying anisotropy factor $\lambda$ when $(n,k)=(30,2)$, on the Stiefel manifold.
	All values have been multiplied by $10^{-3}$.}
	\centering
	\begin{tabular}{llll}
				$\lambda$& 1 & 10& 100 \\
		 GSS $w=5$  & $ [28.4, 34.3, 37.4]$  &$[\textbf{4.9, 5.3, 5.48}]$& $ [\textbf{1.15, 1.33, 1.45}]$   \\
		 RMH   & $ [49.4, \textbf{55.8}, 60.4]$  &$[1.4,2.1,2.5]$& $ [0.86,1.03, 1.31]$  \\
		 GeoRMH  & $ [50.0, 53.1, 67.1]$  & $[1.6,2.4,2.9]$& $ [0.96, 1.04, 1.13]$  \\
		 UIMH & $[\textbf{51.2},55.5,\textbf{69.8}]$ & $[1.2,1.9,2.2] $ &$[0.25,0.26,0.56] $ \\ 
	\end{tabular}
	\label{table:sample_exp3}
\end{table}
In Table \ref{table:sample_exp3}, the samplers' performances worsen as $\lambda$ increases.
This indicates that the sampling task becomes more difficult with growing anisotropy of the target density.
This effect, related to the departure of the target distribution from uniformity, is particularly striking for UIMH.
Namely, we observe that UIMH is among the best performing samplers as $\lambda=1$ and becomes the worst for $\lambda=100$.
At the same time UIMH, RMH and GeoRMH outperform GSS only when $\lambda=1$, highlighting the effectiveness of the slice sampling approach of GSS in dealing with anisotropic densities.

In our experiments, we notice that the number of attempts in the shrinkage procedure is more sensitive to the variance of the target distribution than the value of the parameter $w$.
 Therefore, in cases where GSS is a good fit for the sampling task, the computation time increases accordingly.

\bigskip
\textbf{Varying the variance on the Grassmann manifold.}
In Table \ref{table:sample_exp3_grassman}, we observe that GSS outperforms GeoMALA when $\lambda \in \{10,100\}$, indicating its advantage in situations where the density is concentrated.
In a second phase of analysis, we extended the grid search for the GeoMALA step size to include values in the set $\{0.05, 0.15, 0.2, 0.25, 0.3\}$ to evaluate whether additional tuning could further improve the method’s performance compared to GSS.
 Notably, selecting $h=0.05$ for $\lambda=100$ or $h=0.3$ for $\lambda=10$ allowed GeoMALA to outperform GSS, supporting the idea that gradient-based methods can achieve superior performance at the cost of more extensive tuning. 

Furthermore, we find that using a lower value for $w$ is more suitable when the distribution is sharp (e.g., $w=1, \lambda=100$). Additionally, choosing a large value of the stepping out parameter $m$ can improve the performance when $w$ is chosen too small.

This experiment reaffirms the previous findings, demonstrating that GSS adapts to the shape
of the density, which is promising for practical applications. Moreover, the performance of GSS appears to be quite robust across different choices of $w$, and taking a large value for $m$ strengthens this feature.

	\begin{table}
  \caption{Effective sample size [min, median, max] of 10 repetitions for varying variance factor $\lambda$ when $(n,k)=(3,2)$, on the Grassmann manifold. ``GeoMALA Best $h$'' means that we provide the result for the best stepsize parameter $h$ in $\{0.01,0.1,0.5,1\}$.
  All values have been multiplied by $10^{-3}$.}
  \centering
  \begin{tabular}{llll}
  	$\lambda$& 1 & 10& 100 \\
  	GSS $w=1,m=1$  & $ [9, 10, 11]$  &$[11, 12, 13]$& $ [16, 19, 23]$   \\
  	GSS $w=1,m=3$  & $ [31, 35, 42]$  &$[13, 16, 19]$& $ [12, 20, 21]$   \\
  	GSS $w=1,m=10$  & $ [32, 36, 42]$  &$[\textbf{17, 19, 22}]$& $ [16, 19, 23]$   \\
  	GSS $w=7,m=1$  & $ [33, 36, 41]$  &$[15, 17, 21]$& $[\textbf{18, 19}, 23]$   \\
  	GeoMALA Best $h$   & $ [\textbf{44, 49, 56}]$  &$[11,13,14]$& $[0.395,2.6, \textbf{43}]$  \\
  \end{tabular}
  \label{table:sample_exp3_grassman}
  \end{table}

\bigskip
We discuss some further aspects of the conducted numerical experiments regarding complexity in Supplementary material \ref{appendix:complexity}.

	


\begin{remark}\emph{(Choice of the hyperparameters.)}
	The performance of GSS  seems to be quite robust to the choice of the hyperparameters $m$ and $w$ as soon as $mw \geq 2\pi$ (see Tables \ref{table:sample_exp1}, \ref{table:sample_exp2}, \ref{table:sample_exp3_grassman}).
	In Supplementary material \ref{appendix:sensitivy_analysis}, we further explore this in a sensitivity analysis with respect to $w$ and $m$ for the varying anisotropy experiment on the Stiefel manifold.
	
	Observe that the number $2\pi$ coincides with twice the diameter of the Grassmann and the Stiefel manifold.
	Therefore, as a heuristic we propose to choose $m$ and $w$ such that $mw$ is about twice the diameter for compact manifolds in general.
	Note that developing a theoretical foundation or clear empirical criteria for selecting the hyperparameters of the stepping-out procedure remains an open research question, even in the Euclidean setting \citep{latuszynski2024convergence, power2024weak}. 

\end{remark}
While this heuristic ensures that the effective sample size (ESS) remains close to optimal according to the GSS hyperparameters, it may not yield the optimal ESS per target evaluation.
Optimizing this ratio is meaningful since the method employs a zeroth-order oracle. For example, in a grid search with $w \in \{0.1, 0.5, 1, 3, 5, 7, 9\}$ and $m \in \{1, 3, 5, 7, 9\}$ for the anisotropic experiment with the matrix von Mises-Fisher distribution (see Table 8),
 we observe that the best hyperparameter configuration with respect to ESS per target evaluation is $(m = 1, w = 0.5)$ for $\lambda = 100$ (see Table \ref{table:comparison_nbeval}).
 Notably, this configuration does not adhere to the proposed heuristic.
 This is probably, because in this case most of the mass of the probability distribution is concentrated on a set which has a diameter that is much smaller than the diameter of the Stiefel manifold.
 \begin{table}

	\begin{tabular}{crrr}
		$\lambda$ & 1 & 10 & 100 \\ 
	Best hyperparameter & $(m=1,w=9)$ & $(m=1,w=3)$ & $(m=1,w=0.5)$ \\ 
		Our proposed method & $449$ & $\textbf{23.5}$ & $4.8$\\ 
		Random Walk Metropolis Hasting & $\textbf{542}$ & $23.1$ & $\textbf{8.7}$ \\ 
	\end{tabular}
	\caption{Comparison of ESS divided by the number of target evaluation for the anisotropy experiment on the Stiefel manifold. All values have been multiplied by $10^3$.}
	\label{table:comparison_nbeval}
\end{table}
When comparing our method to the RMH algorithm in Table \ref{table:comparison_nbeval}, which is tuned during a burn-in phase for the anisotropic experiment, we find that our method is competitive only for $\lambda = 10$.
In the worst case, RMH outperforms our method by a factor of two.
 To ensure a fair comparison in real data experiments where MCMC sampling is used for parameter estimation, we employ four times more MCMC iterations for RMH than for our method. 

Our method demonstrates competitiveness in real data experiments because the optimization process updates the target density at each step, making RMH tuning less effective.
 In contrast, our method remains efficient without requiring extensive tuning.
 Our proposed method is robust in the sense that sufficiently high values of $mw$ ensure ESS performance close to optimal, regardless of the anisotropy of the target distribution.
 However, this trick may not be relevant when computational efficiency is a concern. 

Although the proposed method may not be competitive in terms of ESS per target evaluation, it is particularly advantageous when the target density changes over time, as is common in optimization processes characteristic of empirical Bayesian estimation \cite{kuhn2004coupling}.

 When comparing GSS with GeoMALA in the anisotropy experiment, based on the ESS divided by computation time (see Table~\ref{table:comparison_time_geomala}), where GeoMALA is tuned by selecting the best stepsize \( h \in \{0.01, 0.1, 0.5, 1\} \), we observe that our best-performing method is consistently outperformed by GeoMALA.
  However, our heuristic approach with \( mw \geq 2\pi \) consistently performs better than GeoMALA with its second-best stepsize.
 We would like to emphasize that, while GeoMALA exhibits better performance in terms of ESS, the method is fundamentally biased due to the projection steps it employs. This bias is illustrated by the experiments shown in Figure \ref{fig:bias_analysis_geomala}.
 In contrast, GSS is theoretically guaranteed to be unbiased, a fact reflected in Figure~\ref{fig:bias_analysis_geomala} by its very narrow confidence intervals.

 \begin{table}

	\begin{tabular}{crrr}
	$\lambda$ & 1 & 10 & 100 \\ 
		\hline
	GeoMALA Best $h$         & $ \textbf{1361}$ & $\textbf{135.4}$ & $\textbf{41.7}$ \\ 
	Our best method       & $982$ & $93.8$ & $18.3$ \\ 
	\hline
	GeoMALA Second Best $h$  & $346$ & $61.7$ & $9.4$\\ 
		Our heuristic method ($w=7,m=1$) & $\textbf{874}$ & $\textbf{91.1}$ & $\textbf{15.9}$ \\ 
	\end{tabular}
	\caption{Comparison of the median effective sample size (ESS) divided by computation time of 10 repetitions for the anisotropy experiment on the Stiefel manifold.}
	\label{table:comparison_time_geomala}
\end{table}
\begin{figure}[t]
	\begin{center}
	\includegraphics[width=120mm]{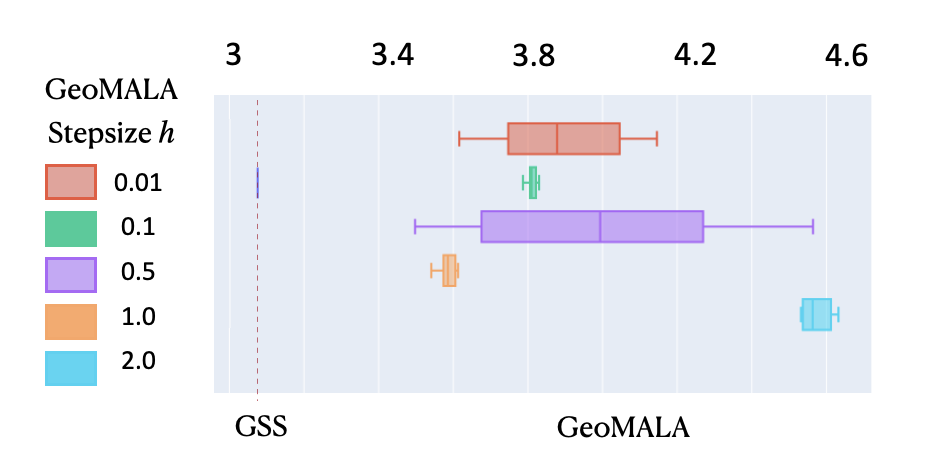}
	\end{center}
	\caption{We report the empirical average of the log density, \(\frac{1}{N} \sum_{k=1}^{N} \text{tr}(F^\top \tilde{X}_k)\), where \((\tilde{X}_k)_{k=1,\ldots,N}\) represents the samples generated by the algorithm for the different samplers when $\lambda=1$.
	The confidence interval are computed over 10 repetitions.}
	\label{fig:bias_analysis_geomala}
\end{figure}

\subsection{Complexity}
\label{appendix:complexity}
A discerning reader might consider the comparison of ESS between RMH, GSS, and GeoRMH unfair, as GSS involves additional computations due to the stepping-out and shrinkage procedure.

We illuminate this in more detail under the assumption that the main computational bottleneck in all methods is the computation of the geodesics, which comprises the reprojection of the proposal on the Stiefel manifold at the end of each step using an SVD  $O(nk^2)$, the sampling on the tangent space using a QR decomposition $O(k^3)$, and the eigenvalues of the skew-symmetric matrices $O((2k)^3)$.

In the case of the stepping-out and shrinkage procedure in GSS, an eigenvalue decomposition is initially computed, enabling the generation of geodesics with some matrix products for the subsequent attempts. Thus, from a computational perspective, the only difference between GeoRMH and GSS is the cost of these \textit{matrix products}.

However, when comparing RMH and GSS, we need to account for the additional computational cost of the QR decomposition and eigenvalue decomposition, both in $O(k^3)$.
Notably, the computational cost related to constraints naturally increases with $k$.
In practice, without any engineering optimization of the codes, we observe that GSS takes between two to four times longer to execute than RMH (which is biased), but it is nearly equivalent to GeoRMH in terms of computation time when $n$ is not excessively large.
The relative speed of RMH compared to GSS has to be weighted by its bias.


Moreover, the computation time of GSS can generally be reduced by using parallel computations.
 The shrinkage and stepping-out procedures involve numerous log-target and geodesic evaluations,
 which can be computed in parallel as long as the stopping rules are applied retrospectively at the end of all evaluations.
 For instance, for the stepping-out procedure one can proceed as follows:
 \begin{enumerate}
	\item First, a value $u$ is sampled in $[0,w]$ on a central server.
	\item Then, instead of computing the geodesic $\gamma$ at $(u+iw)_{i\in \{1,\ldots ,m\}}$ consecutively in the while loop of Algorithm 2, it is evaluated in parallel on different computers.
	 The values $(\gamma(u+iw))_{i\in \{1,\ldots ,m\}}$ are then sent to the central server.
	\item Finally, on the central server, the right value of $u+iw$ is selected a posteriori by executing the while loop of Algorithm 2 with the pre-computed quantities $(\gamma(u+iw))_{i\in \{1,\ldots ,m\}}$.
 \end{enumerate} 
 The same approach can be applied to the shrinkage procedure.


\subsection{A practical case: Understanding the variability in graph data sets.\\ Experiments of estimation on synthetic data.}
\label{appendix:estimation}
We follow the same procedure as \citep{mantoux2021understanding} to generate synthetic data $ (\Upphi^{(j)}_*,\Gamma^{(j)}_*,\kappa^{(j)}_*,\epsilon^{(j)}_*)_{j\in \{1, \ldots, J\}}$
with $J=100$, $(n,k)=(30,5)$ and $(n,k)=(3,2)$. 
The generation parameters are fixed as $\sigma_\epsilon^2=0.1$, $\sigma_\kappa^2=2 $, $\mu=(10,2,\ldots,2)\in \mathbb{R}^k $,
and  $F^* \in \mathbb{R}^{n \times k}$ is chosen as the matrix with columns $a_1 f_1, \ldots, a_k f_k \in \mathbb{R}^{n}$, where $(f_i)_{i\in \{1, \ldots, k\}}\in \mathcal{V}(n,k)$ is sampled uniformly from $\mathcal{V}(n,k)$, and $a_1=\lambda\in \{1,100\}$, $a_i=1$ for any $i\in\{2,\ldots, k\}$.
The factor $\lambda$ in this context serves as an anisotropy factor incorporated to examine how performance is influenced by anisotropy.
The optimization process involves random initialization of the parameters $F$ and $\kappa$, and we deterministically set $\sigma_\kappa^2=\sigma^2_\epsilon=1$.
Then, $(\Gamma^{(j)})_{j \in  \{1, \ldots, J\}}$ is initialised either by performing 200 iterations of GSS on $p_\Gamma$ or 800 iterations of RMH.
The results are averaged over 10 repetitions with different random initializations and generation parameters.

We run the MCMC-SAEM procedure for 100 iterations, and at each step of MCMC-SAEM, 20 iterations of MCMC are performed when GSS is used, while 80 iterations are performed when RMH is used for the E-step of the EM algorithm.
\begin{figure}[!h]
	\begin{center}
		\includegraphics[width=140mm]{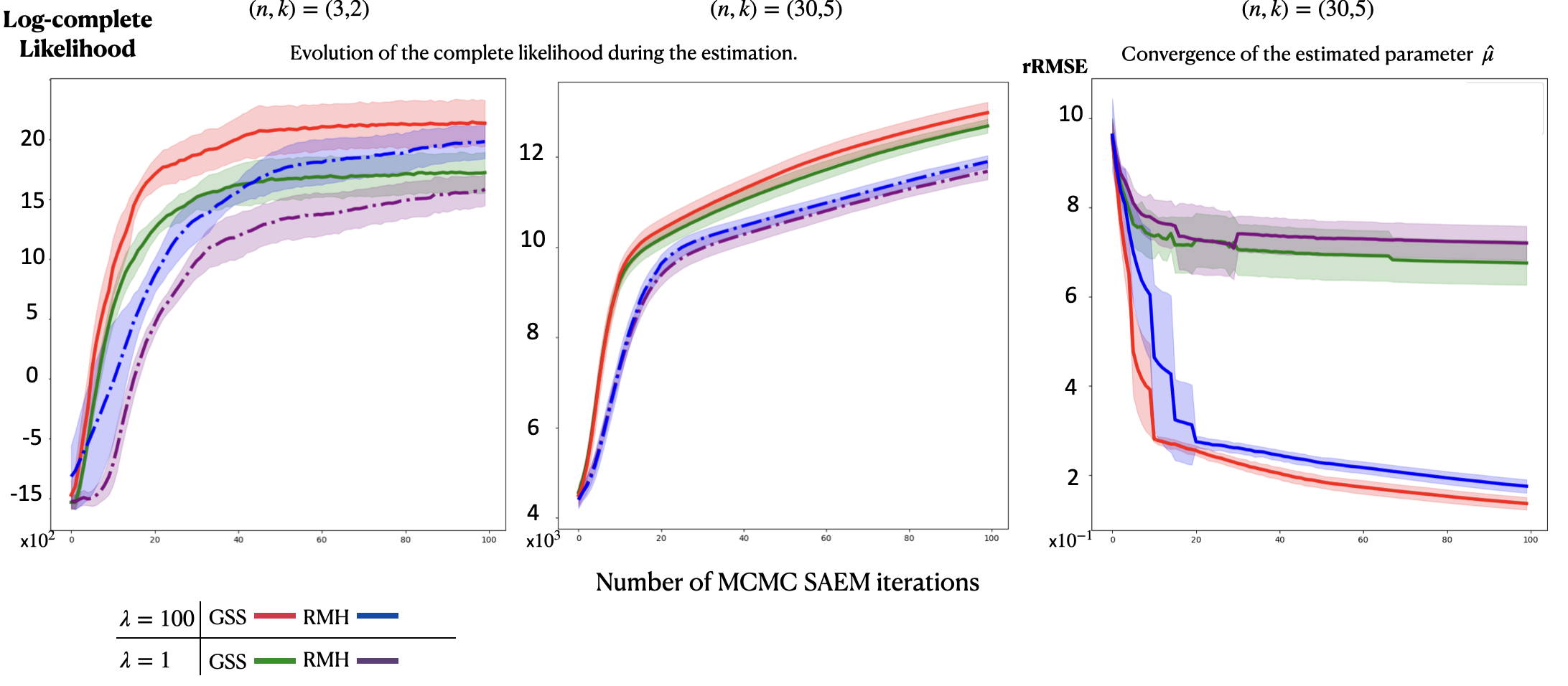}
	\end{center}
	\caption{We compare the performance of RMH and GSS($w=1,m=5$) used as MCMC-method within the MCMC-SAEM procedure to maximise the likelihood of the model.
		The complete log-likelihood is used as a proxy to perform the comparison since the log-likelihood is not tractable.
		The rRMSE is computed for the difference between the estimated parameter $\hat{\mu}$ and the true parameter $\mu$.
		The convergence of the estimated parameter $\hat{\mu}$ is depicted for the case $(n,k)=(30,5)$.
		To ensure a fair comparison, we use four times more MCMC steps for RMH compared to GSS, since our implementation of RMH is four times faster.}
	\label{fig:Exp_networks30,5}
\end{figure}

First, upon examining Figure \ref{fig:Exp_networks30,5}, we can observe that the complete log-likelihood curve\footnote{Denoting by $\theta$ the model parameter, the complete log-likelihood
	is $\log p(\Upphi^{(j)},\kappa^{(i)}, \Gamma^{(i)}|\theta)$ and the log-likelihood is $\log p(\Upphi^{(j)}|\theta)$.}
exhibits higher values when $\lambda=100$ in contrast to $\lambda=1$. This can be attributed to the unobserved variables $(\Gamma^{(j)})_{j\in {1,\ldots,J}}$ being more concentrated in the direction indicated by the first column of $F^*$, making recovery easier.

Secondly, GSS outperforms RMH regarding the complete log-likelihood. The curve increases more rapidly with the number of iterations and attains a higher final value. This improvement is particularly striking for $(n,k)=(30,5)$, which is consistent with the findings in Table \ref{table:sample_exp2}. Moreover, as illustrated by the third graph of Figure \ref{fig:Exp_networks30,5}, the relative root mean square error (rRMSE) of the estimated parameter to the true parameter is always smaller when using GSS, especially when $\lambda=100$.

\subsection{Sensitivity analysis}
\label{appendix:sensitivy_analysis}
In this section, the experiment with varying anisotropy on the Stiefel manifold of Section 3.1 is extended by a sensitivity analysis of GSS on the hyperparameters $w,m$.
The ESS median of GSS is reported in the Figures \ref{fig:sensitivity_analysis_lambda_1} to \ref{fig:sensitivity_analysis_lambda_100} for any $w\in [0.1,0.5,1,2,3,7,9]$, $m\in [1,3,5,7,10]$ and $\lambda=\{1,10,100\}$
We observe that the results are satisfactory on average as long as $mw \geq 2\pi$ and that $(w,m)=(5,1)$ is a fair choice on average even if suboptimal.
For higher $\lambda$, satisfactory results are obtained for lower $mw$ values. 

\begin{figure}[t]
	\begin{center}
	\includegraphics[width=120mm]{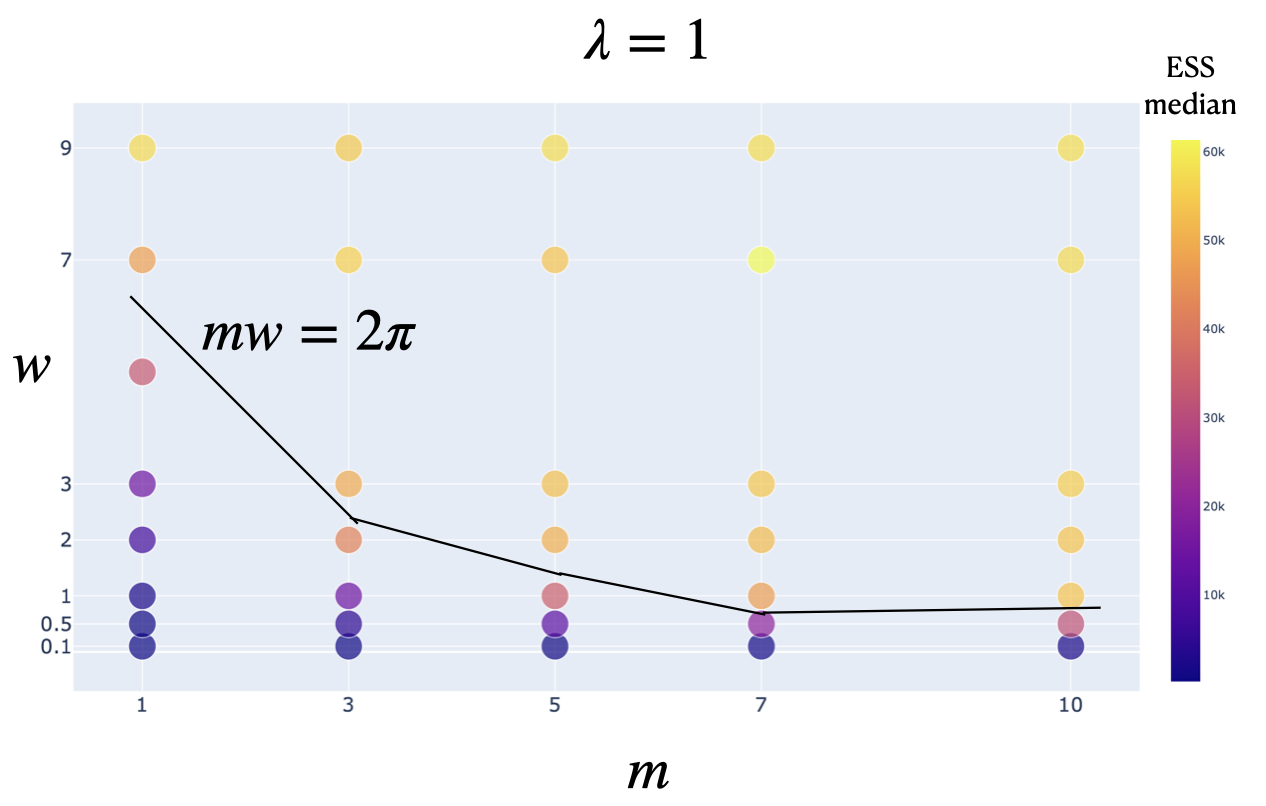}
	\end{center}
	\caption{Sensitivy analysis of the anisotropy experiment on the Stiefel manifold with $w\in [0.1,0.5,1,2,3,7,9]$ and $m\in [1,3,5,7,10]$ for $\lambda=1$.}
	\label{fig:sensitivity_analysis_lambda_1}
\end{figure}
\begin{figure}[t]
	\begin{center}
	\includegraphics[width=120mm]{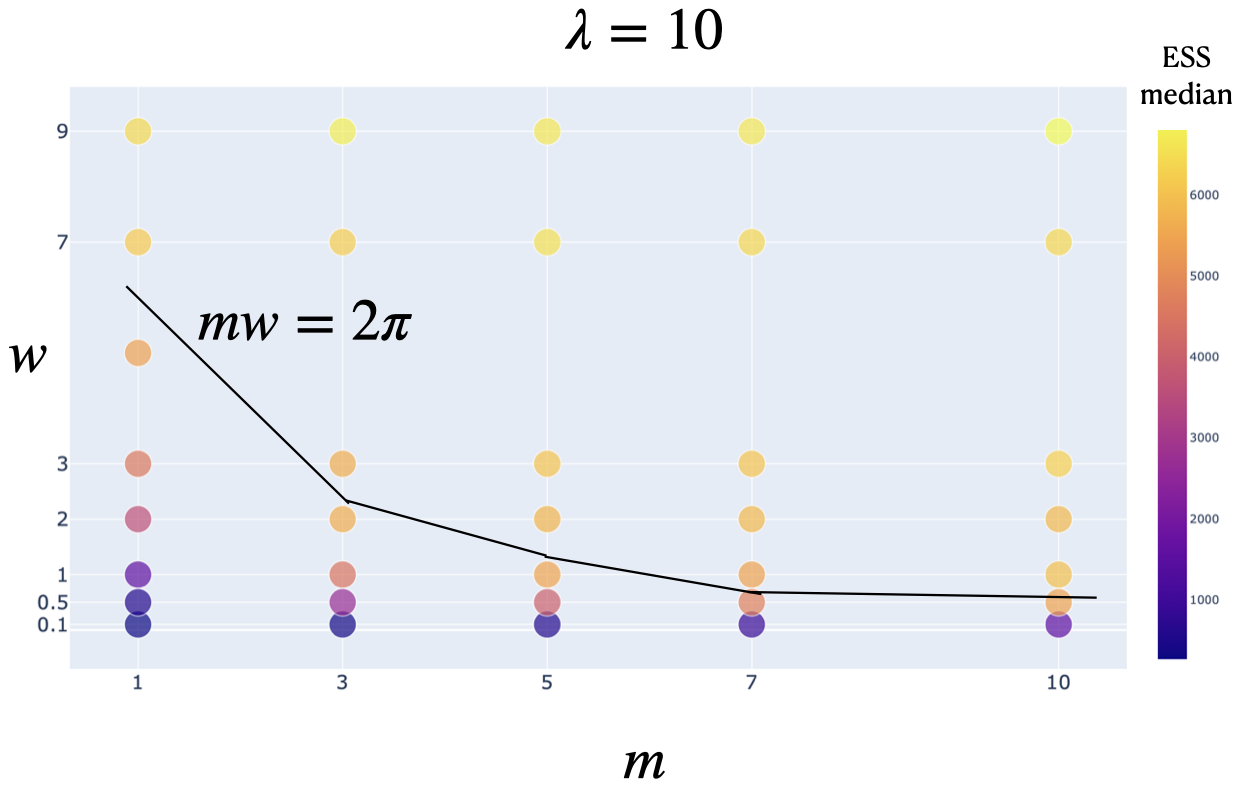}
	\end{center}
	\caption{Sensitivy analysis of the anisotropy experiment on the Stiefel manifold with $w\in [0.1,0.5,1,2,3,7,9]$ and $m\in [1,3,5,7,10]$ for $\lambda=10$.}
	\label{fig:sensitivity_analysis_lambda_10}
\end{figure}
\begin{figure}[t]
	\begin{center}
	\includegraphics[width=120mm]{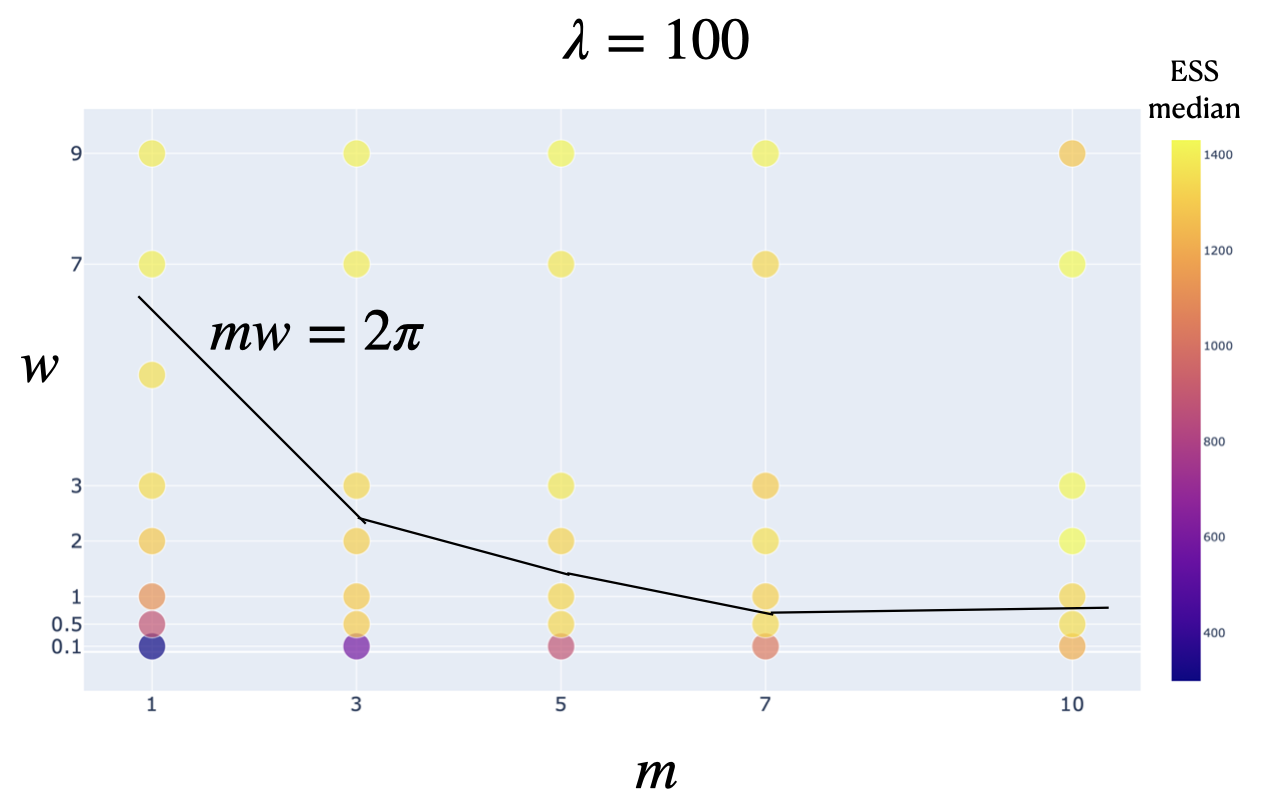}
	\end{center}
	\caption{Sensitivy analysis of the anisotropy experiment on the Stiefel manifold with $w\in [0.1,0.5,1,2,3,7,9]$ and $m\in [1,3,5,7,10]$ for $\lambda=100$.}
	\label{fig:sensitivity_analysis_lambda_100}
\end{figure}

\subsection{Additional experiment using the ARMA model}
\label{appendix:ARMA}
Time series related to different types of data, such as dynamic textures, shape sequences and videos, are often modelled as auto-regressive and moving average (ARMA) models
\citep{doretto2003dynamic,aggarwal2004system,bissacco2001recognition,veeraraghavan2005matching}.
Provided observations $z=(z_t)_{t\in\{1,\ldots, T\}} \in (\mathbb{R}^n)^{T} $ at $T$ time points, the ARMA model is described by the equations
\begin{eqnarray*}
	&z_t=H x_t+w_t, \quad w_t  \overset{\text{i.i.d.}}{\sim} \mathcal{N}(\boldsymbol{0}_n, R), \\
	&x_{t+1}=B x_t+v_t, \quad  v_t  \overset{\text{i.i.d.}}{\sim}  \mathcal{N}(\boldsymbol{0}_k, Q), \qquad t\in \{1,\ldots,T\},
\end{eqnarray*}
where $x=(x_t)_{t\in \{1,\ldots,T\}}\in (\mathbb{R}^k)^T$ is the hidden state vector, $B\in \mathbb{R}^{k \times k}$ the transition matrix between the states at consecutive time points, $H\in \mathcal{V}(n,k)$ the measurement matrix and $R\in \mathbb{R}^{n\times n}, Q\in \mathbb{R}^{k\times k}$ covariance matrices. 
We can wrap this model in a Bayesian framework by considering the priors $x_1 \sim \mathcal{N}(\boldsymbol{0}_k,Q) $ and $ H \sim \operatorname{vMF}(F)$ with $F\in \mathbb{R}^{n\times k}$, such that $(F,B,R,Q)$ are seen as hyperparameters.

In this experiment, we compare the ESS related to the sampling of the posterior $ p(H|z) \propto p(z|H)p(H)$ by using GSS and RMH.
More precisely, we compute the ESS of $(\mathrm{tr}(F^\top H_i))_{i\in\{1, \ldots,T\}}$ where $(H_i)_{i \in\{1,\ldots,T\}}$ are the samples.
The expression of $p(H)$ is given in \eqref{eq:Mises_fisher_distribution} and $p(z|H)$ can be computed using Kalman filter updates. 
The observations $(z_t)_{t \in \{1, \ldots,T\}}$ are synthetically generated from the model where the parameters are chosen randomly with $T=10$ and varying dimensions $(n,k)$.
We choose a low variance for the observation covariance matrix $R$ and a prior close to the true parameter in order to have a concentrated posterior.

In Table \ref{table:ARMAEXP}, GSS outperforms RMH.
Moreover, for large $w$ we obtain higher ESS for GSS.
Surprisingly, increasing $n$ improves the ESS for GSS, but not for RMH, contrary to the experiments with the matrix von Mises–Fisher distribution in Section \ref{Sec: Experiments von Mises-Fisher distribution}.
This highlights how the target distribution influences the sampling quality.
For the case $(n,k)=(30,5)$ we observe that increasing $m$ while reducing $w$ increases the performances.
This suggest that if large transitions can be achieved only in specific situations, then large $m$ allows this when necessary. 
\begin{table}
	\caption{Effective sample size [min, median, max] of 10 repetitions for varying dimensions $(n,k)$ of the Stiefel manifold.
	All values have been multiplied by $10^{-2}$.}
	\centering
	\begin{tabular}{lllll}
		\toprule 
		($n,k$) & $(30,2)$ & $(30,5)$& $(100,2)$ \\
		GSS $w=1, m=1$ :  & [3.9, 4.2, 4.5] &[3.9, 4.0, 4.4]& [3.0, 3.1, 3.4]   \\
		GSS $w=5, m=2$ :  & [4.2, 4.7, 4.8] &[\textbf{4.7, 5.1, 5.6}]& [1.1, 1.4, 1.8]   \\
		GSS $w=10, m=1$ :  & [\textbf{5.1, 5.5, 5.9}] &[4.4, 4.8, 5.3]& [\textbf{1.5, 2.0, 2.7}]   \\
		RMH :  & [2.9, 3.7, 3.9]  &[4.0, 4.2, 4.7]& [2.9, 3.0, 3.2] \\
		
	\end{tabular}
	\label{table:ARMAEXP}
\end{table}

\subsection{Additional Experiments on the Manifold of Symmetric Positive Definite Matrices $\mathsf{S}_d^{++}$}

Previous experiments were conducted on compact manifolds. Here, we consider the manifold of symmetric positive definite matrices defined by
\[
\mathsf{S}_d^{++} = \left\{ A \in \mathbb{R}^{d \times d} : A = A^\top,\ x^\top A x > 0,\ \forall x \in \mathbb{R}^d \setminus \{0\} \right\}.
\]
The target distribution is a Riemannian Gaussian measure \cite{said2017riemannian} centered at the identity, given by
$
\pi(x) \propto \exp(-\lambda\, \dist(x, I_d)^2),
$
where $\dist$ denotes the geodesic distance induced by the affine-invariant Riemannian metric, and $\lambda > 0$.
 Further details can be found in \cite[Section 3.3]{pennec2006riemannian} including the expression of geodesics and the Riemannian metric.

We compare GeoRMH, GSS (\(w = 1, m = 10\)), and RMH using the ESS associated with the statistic $M \in \mathsf{S}_d^{++} \mapsto \trace(M)$ on $N=10^5$ samples. The dimension $d$ is 10 and the samplers are initialized with the identity.
 RMH applies a pseudo-projection of symmetric matrices onto $\mathsf{S}_d^{++}$ by thresholding the eigenvalues as $x \mapsto \max(\min(x, 10^8), 10^{-4})$ to avoid numerical instabilities.

\begin{table}
	\centering
		\caption{ESS [min, median, max] of 10 repetitions for varying variance factor $\lambda$ for the experiment on the SPD manifold. All values are scaled by \(10^{-3}\).}
	\label{table:SPD}
	\begin{tabular}{crrr}
		$\lambda$ & 1 & 10 & 100 \\ 
		GeoRMH & [0.6, 0.6, 0.7] & [1.9, 2.1, 2.1] &  [2.9, 3.2, 3.7]\\ 
		GSS (\(w = 1, m = 10\)) & [\textbf{0.7}, \textbf{0.7}, \textbf{0.8}] & [\textbf{2.5}, \textbf{2.6}, \textbf{2.9}] & [\textbf{4.1}, \textbf{4.7}, \textbf{5.2}]\\ 
		RMH & [0.2, 0.2, 0.2] & [\textbf{0.4}, 0.4, 0.5] & [0.5, 0.5, 0.6] \\ 
	\end{tabular}
\end{table}

As shown in Table \ref{table:SPD}, RMH performs poorly compared to GSS and GeoRMH. Moreover, GSS outperforms GeoRMH for any \(\lambda \in \{1, 10, 100\}\) for the same reason mentioned earlier: its ability to adapt to the target geometry.

The performance gap between RMH and GeoRMH can be attributed to the fact that GeoRMH uses the exponential map to define geodesics, enabling large transitions that are well-suited to the target distribution, which involves the geodesic distance
$
\dist(M, I_d)^2 = \sum_{i=1}^d \log^2(\zeta_i),
$
where $(\zeta_i)_{i=1}^d$ are the eigenvalues of $M$.

\subsection{A practical case: Understanding the variability in graph data sets.\\ Used parameters.}

\label{appendix:numerical_details_missing_link_imputation}
\begin{table}[hbt!]
    \caption{Parameters used in the experiments of Section \ref{sub_sec:practical_case}. $\operatorname{Unif}(n,k$) denotes the uniform sampling on $\mathcal{V}(n,k)$ and $\cdot$ the colomn wise multiplication.}
    \centering
    \begin{tabular}{lllll}
              \toprule 
              Parameter & $\sigma_\epsilon$ &$\sigma_\kappa$ & $\mu $ & $F$ \\
              Synthetic data estimation $(n,k,\lambda)$ &0.1 &2 &$[10,2,\ldots]$& $[\lambda,1,\ldots]\cdot\operatorname{Unif}(n,k) $   \\
              Missing link imputation  & 0.3  &2& $ [20,10,5,2,-10]$ &[60,20,20,20,5]$\cdot \operatorname{Unif}(20,5)$  \\

            \end{tabular}
          \label{table:set_up_exp_practical}
      \end{table}



\end{document}